\documentclass[12pt]{article}
 \usepackage{amsmath, amsthm, amssymb, bbm, setspace, bigints}
 \usepackage[margin=1in]{geometry}
\usepackage{caption, enumitem}
\usepackage{subcaption}
\usepackage[toc,page]{appendix}
\usepackage{graphicx, amsfonts, tikz, multirow}
\usepackage{natbib}
\usetikzlibrary{bayesnet}
\usepackage{pdfpages}
\usepackage{epstopdf}
\usepackage{booktabs}
\usepackage[ruled]{algorithm2e}
\usepackage[colorlinks,citecolor=blue,urlcolor=blue]{hyperref}
\usepackage[utf8]{inputenc}
\usepackage{authblk}
\usepackage{float}

\def \T {\top}
\def \M {\mathbf{M}}

\def \Y {\mathbf{Y}}
\newcommand{\wt}[1]{\widetilde{#1}}
\newcommand{\wh}[1]{\widehat{#1}}
\def \iid {\text{iid}}
\def \ind {\text{ind}}

\newtheorem{theorem}{Theorem}[section]
\newtheorem{lemma}[theorem]{Lemma}
\newtheorem{assumption}[theorem]{Assumption}

\newtheorem{proposition}[theorem]{Proposition}

\graphicspath{{Pictures/}}

\date{}
\title{Blessing of dimension in Bayesian inference on covariance matrices}

\author[1]{Shounak Chattopadhyay\thanks{shounakch@virginia.edu}}
\author[2,3,4]{Anru R. Zhang
}
\author[4,5]{David Dunson 
}
\affil[1]{Department of Statistics, University of Virginia}
\affil[2]{Department of Biostatistics and Bioinformatics, Duke University} 
\affil[3]{Department of Computer Science, Duke University} 
\affil[4]{Department of Statistical Science, Duke University}
\affil[5]{Department of Mathematics, Duke University}

\begin{document}
\maketitle
\begin{abstract}
Bayesian factor analysis is routinely used for dimensionality reduction in modeling of high-dimensional covariance matrices. 
Factor analytic decompositions express the covariance as a sum of low-rank and diagonal matrices. 
In practice, Gibbs sampling algorithms are typically used for posterior computation, alternating between updating the latent factors, loadings, and residual variances. 
In this article, we exploit a blessing of dimensionality to develop a provably accurate posterior approximation for the covariance matrix that bypasses the need for Gibbs or other variants of Markov chain Monte Carlo sampling. 
Our proposed Factor Analysis with BLEssing of dimensionality (FABLE) approach relies on a first-stage singular value decomposition (SVD) to estimate the latent factors, and then defines a jointly conjugate prior for the loadings and residual variances. 
The accuracy of the resulting posterior approximation for the covariance improves with increasing samples as well as increasing dimensionality. 
We show that FABLE has excellent performance in high-dimensional covariance matrix estimation, including producing well-calibrated credible intervals, both theoretically and through simulation experiments. 
We also demonstrate the strength of our approach in terms of accurate inference and computational efficiency by applying it to a gene expression dataset.
\end{abstract}


\section{Introduction}
\label{sec:intro}

Inference on covariance matrices in high-dimensional data is a key focus in many application areas, motivating a rich literature on associated statistical methods. 
One thread of this literature focuses on frequentist estimators of high-dimensional covariance matrices.
Such approaches avoid direct modeling of the data and simply estimate the covariance under various assumptions on its inherent low-dimensional structure, including (but not limited to) banded covariance \citep{bickel2008regularized}, low-rank structure \citep{shikhaliev2019low}, low-rank with sparsity \citep{richard2012estimation}, sparse covariance \citep{bien2011sparse}, and sparse inverse precision matrix estimation \citep{zhang2014sparse}.
Our interest is instead in model-based Bayesian approaches, which have advantages in terms of naturally accommodating complexities in the data and quantifying uncertainty, while having disadvantages in terms of computational efficiency, particularly in implementing Markov chain Monte Carlo (MCMC) algorithms to sample from the posterior distribution.
 
Our particular focus is on factor models, which express the covariance matrix as a sum of low-rank and diagonal matrices, and represent a popular probabilistic model-based alternative to principal components analysis (PCA). 
Factor analysis introduces a lower-dimensional latent factor $\eta_i \in \mathbb{R}^k$ corresponding to each observation $y_i \in \mathbb{R}^p$ with $k \ll p$, such that $y_i = \Lambda \eta_i + \epsilon_i$ for a factor loadings matrix $\Lambda \in \mathbb{R}^{p \times k}$ and independent idiosyncratic errors $\epsilon_i \sim N_p(0, \Sigma)$ with diagonal $\Sigma$. 
Following typical practice, we let the latent factors have independent Gaussian priors $\eta_i \sim N_k(0, \mathbb{I}_k)$.
Marginalizing out the latent factors leads to $y_i \sim N_p(0, \Psi)$, with the covariance matrix $\Psi = \Lambda \Lambda^{\top} + \Sigma$. 
This provides a convenient decomposition of the covariance, which dramatically reduces the number of free parameters required to model the $p \times p$ covariance matrix from $\mathcal{O}(p^2)$ to $\mathcal{O}(pk)$, since $\Lambda \in \mathbb{R}^{p \times k}$ and $\Sigma \in \mathbb{R}^{p \times p}$ is diagonal.

There is a vibrant recent literature improving upon and expanding the scope of Bayesian factor analysis  \citep{schiavon2022generalized, de2021bayesian, fruhwirth2023generalized, roy2021perturbed, ma2022posterior, bolfarine2022decoupling, xie2022bayesian}. 
Even with increasingly rich classes of priors and data types such as massive binary/count data, the canonical approach for posterior computation remains Gibbs samplers that iterate between updating the latent factors, factor loadings, residual variances, hyperparameters controlling the hierarchical prior, and other model parameters. 
This approach is simple to implement in broad model classes, but it commonly faces problems with slow mixing, particularly as data dimensionality and complexity increase. 
Issues with slow mixing and bottlenecks with large datasets have motivated alternative posterior sampling algorithms that apply Hamiltonian Monte Carlo (HMC) after marginalizing out the latent factors, exploiting sufficient statistics and parallel computation to facilitate implementation for large sample sizes \citep{chandra2023inferring, man2022mode}. 
Nonetheless, all current MCMC sampling algorithms for large covariance matrices face substantial computational hurdles as the number of dimensions increases.



The lack of sufficiently scalable MCMC algorithms has motivated a rich literature on  developing more computationally efficient algorithms for Bayesian inference in factor models. 
Variational Bayes approaches \citep{hansen2023fast, wang2021empirical} attempt to approximate the posterior, but often the accuracy of the approximation is poor, with a tendency to dramatically underestimate posterior uncertainty. 
A concise review comparing variational approximations and posterior sampling approaches for factor analysis is provided in \cite{foo2021comparison}. 
Maximum {\em a posteriori} (MAP) estimation under sparsity priors \citep{srivastava2017expandable} and expectation-maximization (EM) algorithms \citep{rovckova2016fast,  avalos2022heterogeneous, zhao2016bayesian} provide computationally efficient point estimates of the covariance matrix. 
\cite{sabnis2016divide} develop a distributed computing framework aggregating local estimates of the covariance across different cores.
\cite{srivastava2017expandable}, \cite{rovckova2016fast}, and \cite{sabnis2016divide} provide theoretical results regarding estimation accuracy and posterior concentration, without providing any guarantees on uncertainty quantification (UQ). 
A key open problem in the literature remains how to obtain an efficient posterior approximation for a high-dimensional covariance that has guarantees in terms of valid frequentist coverage of credible intervals.

Fast algorithms for Bayesian factor analysis capable of scaling efficiently to high-dimensional data typically sacrifice the ability to provide an accurate characterization of uncertainty. 
The focus of this article is on proposing a simple approach for overcoming this limitation, providing computationally efficient inference on a high-dimensional covariance matrix.
Assuming the rank $k$ remains fixed, we have more and more variables loading on the same small number $k$ of latent factors with increasing dimensions.
As a consequence, we obtain a blessing of dimensionality phenomenon allowing us to first pre-estimate the latent factors and then leverage this point estimate to obtain an approximate posterior distribution of the factor loadings $\Lambda$ and error variances $\Sigma$ under conjugate priors on these parameters. 
This in turn induces a posterior approximation for the covariance matrix $\Psi = \Lambda \Lambda^{\top} + \Sigma$.
The error in estimating the covariance based on this posterior approximation decreases as both the number of samples and number of dimensions increases. 
Our proposed Factor Analysis with BLEssing of dimensionality (FABLE) approach completely bypasses MCMC. 
Furthermore, we develop a data-adaptive coverage-correction methodology which ensures that the credible intervals for the entries of the covariance have guarantees regarding valid frequentist coverage.

\cite{fan2023latent} also pre-estimate latent factors, but in the context of regression with high-dimensional correlated covariates. 
In their approach, a first-stage principal components analysis is carried out for the high-dimensional covariates. 
The results are then used to obtain a modified covariate vector consisting of principal component scores and de-correlated versions of the original covariates. 
Using this new covariate vector in the second stage high-dimensional regression avoids multicollinearity. 
This focus, and the corresponding methodology, are substantially different from ours.
A blessing of dimensionality phenomenon is also explored in \cite{li2018embracing}, where estimation of the covariance between a subset of the variables of interest is improved by incorporating auxiliary information from other variables. 
Their methodology, which focuses on optimizing a constrained weighted least-squares problem, is also considerably different from ours, although we demonstrate similar gains with FABLE in a real-world gene expression data application when estimating covariance submatrices of relevant subsets of genes.


In Section \ref{sec:method}, we develop our FABLE methodology and provide an approach to choose key hyperparameter values. In Section \ref{sec:theory}, we provide theoretical results on posterior contraction rates and UQ, leveraging a blessing of dimensionality phenomenon. In Section \ref{sec:simulation}, we validate our approach through comparison with other state-of-the-art approaches from the viewpoint of estimation error, frequentist coverage, and computational runtime in various numerical experiments. In Section \ref{sec:application}, we present a relevant application of our approach to a large gene expression dataset and highlight distinct advantages over competitors. In Section \ref{sec:discussion}, we provide a discussion on future research directions.

\section{Proposed Methodology}
\label{sec:method}

\subsection{Overview}
\label{subsec:genapproach}

The observed data consist of $\mathbf{Y} = \left[y_1, \ldots, y_n \right]^{\T} \in \mathbb{R}^{n \times p}$, where $y_i \in \mathbb{R}^{p}$ for $i=1,\ldots,n$. 
We consider a factor model for $y_i$ given by
\begin{equation}
\label{eq:LFM}
y_i = \Lambda \eta_i + \epsilon_i,\qquad \epsilon_i \overset{\iid}{\sim} N_p \left(0, \Sigma \right),
\end{equation}
where we have omitted the intercept term, assuming the data have been centered prior to analysis. 
Here, $\Lambda \in \mathbb{R}^{p \times k}$ is an unknown matrix of factor loadings, $\eta_i \overset{\iid}{\sim} N_k \left(0, \mathbb{I}_k \right)$ is the $i$th latent factor, and $\epsilon_i$ is a zero-mean idiosyncratic error having diagonal covariance $\Sigma = \mbox{diag} \left(\sigma_1^2, \ldots, \sigma_p^2 \right)$ for $i=1,\ldots,n$. 
Integrating out the latent factors leads to 
\begin{equation}
\label{eq:LFMmarginal}
    y_i \mid \Lambda, \Sigma \, \overset{\iid}{\sim} \, N_p \left(0, \Lambda \Lambda^{\T} + \Sigma \right).
\end{equation}
In this paper, our goal is to estimate $L = \Lambda \Lambda^{\top}, \Sigma,$ and
$\Psi = \Lambda \Lambda^{\T} + \Sigma$. 
It is typical to assume that $k \ll p$.
This facilitates dimension reduction of the high-dimensional $y_i \in \mathbb{R}^p$ to the lower-dimensional $\eta_i \in \mathbb{R}^k$, along with ensuring that the high-dimensional covariance matrix $\Psi \in \mathbb{R}^{p \times p}$ of the $y_i$s may be interpreted as the sum of low-rank and diagonal matrices.
We first illustrate our methodology assuming $k$ is known and later discuss an approach to estimate $k$ in Section \ref{subsubsec:tunek}.
To ease exposition, we suppress the dependence of relevant quantities on $k$ unless mentioned otherwise.

Let $\mathbf{M} = \left[\eta_1,\ldots,\eta_n \right]^{\T} \in \mathbb{R}^{n \times k}$ and $\Lambda = \left[\lambda_1, \ldots, \lambda_p \right]^{\T} \in \mathbb{R}^{p \times k}$, with $\lambda_j^{\T}$ denoting the $j$th row of $\Lambda$. 
We also denote the $j$th column of $\mathbf{Y}$ as $y^{(j)}$, so that $\mathbf{Y} = [y^{(1)}, \ldots, y^{(p)} ].$ The latent factor model \eqref{eq:LFM} can be alternatively expressed as
\begin{equation}
\label{eq:LFM2}
    y^{(j)} = \mathbf{M} \lambda_j + \epsilon^{(j)},
\end{equation}
where $\epsilon^{(j)}$ is the $j$th column of the matrix $E = \left[\epsilon_1, \ldots, \epsilon_n \right]^{\T} = [\epsilon^{(1)}, \ldots, \epsilon^{(p)} ],$ with $\epsilon^{(j)} \overset{\ind}{\sim} N_n \left(0, \sigma_j^2 \mathbb{I}_n \right)$ for $j=1,\ldots,p$. Writing in matrix form, we obtain
\begin{equation}
    \label{eq:LFM-matrix}
    \mathbf{Y} = \mathbf{M} \Lambda^{\T} + E.
\end{equation}

We now outline our strategy to obtain an approximation for the posterior of $\Psi$ from which computationally efficient samples can be drawn. 
First, we estimate $\mathbf{M}$ by a plug-in estimator $\widehat{\mathbf{M}}$
based on the singular value decomposition (SVD) of $\mathbf{Y}$ and then consider a surrogate model to \eqref{eq:LFM2} where $\mathbf{M}$ is replaced by $\widehat{\mathbf{M}}$. Conditional on the observed data, $\widehat{\mathbf{M}}$ is fixed. 
Thus, the surrogate model is equivalent to $p$ parallel regression problems. 
Next, in the surrogate model, we endow the regression coefficient and the error variance for the $j$th regression with conjugate normal-inverse gamma prior distributions for $j=1,\ldots,p$.  
We then sample the regression coefficients and error variances jointly from the posterior approximations, across the $p$ different regressions in parallel. 
Finally, we combine the obtained samples of the regression coefficients and the error variances to obtain a sample of $\Psi$ drawn from the posterior approximation. 

To illustrate our choice for $\widehat{\mathbf{M}}$, we start with the SVD of $\mathbf{Y}$, given by
\begin{equation}
    \label{eq:SVD}
    \mathbf{Y} = U D V^{\T} + U_{\perp} D_{\perp} V_{\perp}^{\T},
\end{equation}
where $U \in \mathbb{R}^{n \times k}, D \in \mathbb{R}^{k \times k}, V \in \mathbb{R}^{p \times k}, U_{\perp} \in \mathbb{R}^{n \times (r-k)}, D_{\perp} \in \mathbb{R}^{(r-k) \times (r-k)},$ and $V_{\perp} \in \mathbb{R}^{p \times (r-k)}$, with $r = n \wedge p.$ 
The columns of $U, U_{\perp}, V, V_{\perp}$ consist of orthonormal vectors, with $U^{\T} U_{\perp} = V^{\T} V_{\perp} = \mathbf{O}_{k \times (r-k)}$, where $\mathbf{O}_{p_1 \times p_2}$ denotes the $p_1 \times p_2$ matrix with all entries equal to $0$. 
Let $\mathbf{A} = \mathbf{Y}V / \sqrt{p} = UD / {\sqrt{p}}$ and  $\widehat{C} \in \mathbb{R}^{k \times k}$ be an invertible matrix satisfying 
\begin{equation}
    \label{eq:Chat}
    \widehat{C} \widehat{C}^{\T} = \dfrac{1}{n} \mathbf{A}^{\T} \mathbf{A} = \dfrac{D^2}{np}.
\end{equation}
An immediate choice of $\widehat{C}$ satisfying \eqref{eq:Chat} is given by $\widehat{C} = D / \sqrt{np}$. 
Then, we let 
\begin{equation}
    \label{eq:Mhat}
    \widehat{\mathbf{M}} = \mathbf{A} \left(\widehat{C}^{\T} \right)^{-1}.
\end{equation}
For the choice of $\widehat{C} = D / \sqrt{np},$ it is straightforward from \eqref{eq:Mhat} that $\widehat{\mathbf{M}} = \sqrt{n} U.$ 
We will denote this particular choice of $\widehat{\M} = \sqrt{n} U$ as the {\em canonical choice} of $\widehat{\mathbf{M}}$.
However, other choices of $\wh{\mathbf{M}}$ are also possible for different choices of $\widehat{C}$ satisfying \eqref{eq:Chat}, such as $\wh C = DQ / \sqrt{np}$ for any $k \times k$ rotation matrix $Q$ satisfying $QQ^\top = \mathbb{I}_k$.
If either $n$ or $p$ is very large, approximate SVD approaches \citep{halko2011finding} could be implemented instead of \eqref{eq:SVD} to improve computational efficiency.
We provide a heuristic leading to \eqref{eq:Mhat} in Section \ref{subsec:preest}.

Given a particular choice of $\widehat{\mathbf{M}}$, we consider the following surrogate model:
\begin{equation}
    \label{eq:LFM-surrogate}
    y^{(j)} = \widehat{\mathbf{M}} \widetilde{\lambda}_j + \widetilde{\epsilon}^{ (j)},\quad 
    \widetilde{\epsilon}^{(j) } \overset{\ind}{\sim} N_n \left(0, \widetilde{\sigma}_j^{ 2} \mathbb{I}_n \right)
\end{equation}
for $j=1,\ldots,p$. 
The model \eqref{eq:LFM-surrogate} could be interpreted as a version of \eqref{eq:LFM2} with $\widehat{\mathbf{M}}$ substituted for the original matrix of latent factors $\mathbf{M}$ and surrogate model parameters $\widetilde{\lambda}_j \in \mathbb{R}^{k}$, $\widetilde{\sigma}_j^{2} > 0$ for $j=1,\ldots,p$, endowed with independent normal-inverse gamma (NIG) priors $\left(\widetilde{\lambda}_j, \widetilde{\sigma}_j^2 \right) \overset{\iid}{\sim} \mbox{NIG}\left(0_k, \tau^2 \mathbb{I}_k, \gamma_0/2, \gamma_0 \delta_0^2 / 2 \right)$. 
That is, we let
\begin{align}
\label{eq:prior}
    \widetilde{\lambda}_j \mid \widetilde{\sigma}_j^{2} & \sim N_k \left(0_k, \widetilde{\sigma}_j^2 \, \tau^2 \, \mathbb{I}_k \right),\quad 
    \widetilde{\sigma}_j^{2} \sim \mbox{IG}\left(\dfrac{\gamma_0}{2}, \dfrac{\gamma_0 \delta_0^2}{2}\right).
\end{align}
The global shrinkage parameter $\tau^2$ allows us to \textit{a priori} shrink the factor loadings towards zero, regularizing $\widetilde{\Lambda} = \left[\widetilde{\lambda}_{1}, \ldots, \widetilde{\lambda}_{p} \right]^{\T}$. We discuss a data-driven strategy to estimate $\tau^2$ in Section \ref{subsubsec:tunetau}.

The surrogate model \eqref{eq:LFM-surrogate} and the conjugate prior specification \eqref{eq:prior} motivate the following updating scheme.
We first define for $j = 1, \ldots, p$:
\begin{align}
    \label{eq:hyper}
    \begin{split}
        \mu_j & = \left(\widehat{\mathbf{M}}^{\T} \widehat{\mathbf{M}} + \dfrac{\mathbb{I}_k}{\tau^2} \right)^{-1}\widehat{\mathbf{M}}^{\T} y^{(j)} = \dfrac{\wh{\mathbf{M}}^\top y^{(j)}}{n + \tau^{-2}} ,\\
    \mathbf{K} & = \left(\widehat{\mathbf{M}}^{\T} \widehat{\mathbf{M}} + \dfrac{\mathbb{I}_k}{\tau^2} \right)^{-1} = \dfrac{\mathbb{I}_k}{n + \tau^{-2}},\\ 
    \gamma_n & = \gamma_0 + n,\\
    \gamma_n \delta_j^2 &= \gamma_0 \delta_0^2 +  \left(y^{(j) \T}y^{(j)} - \mu_j^{\T} \mathbf{K}^{-1} \mu_j\right),
    \end{split}
\end{align} 
and proceed to sample
\begin{equation}
\label{eq:FABLESampling}
    \widetilde{\lambda}_j \mid \widetilde{\sigma}_j^2, \Y \sim N_k\left(\mu_j, \rho^2 \widetilde{\sigma}_j^2 \mathbf{K} \right), \quad \widetilde{\sigma}_j^2 \mid \Y \sim \mbox{IG}\left(\gamma_n/2, \gamma_n \delta_j^2 / 2 \right), \, \, \text{for $1 \leq j \leq p.$}
\end{equation}
The sampling scheme in \eqref{eq:FABLESampling} is obtained by treating $\widehat{\mathbf{M}}$ as a fixed design matrix for the $j$th surrogate regression and then applying Bayes' rule, along with an additional coverage-correction coefficient $\rho^2$ when sampling the loadings $\widetilde{\lambda}_j$ conditional on the sampled $\widetilde{\sigma}_j^2$. 
The coverage-correction coefficient $\rho^2$ is chosen to ensure 
that the average asymptotic frequentist coverage of the $100(1-\alpha)\%$ credible intervals across the $p(p+1)/2$ distinct entries of the resulting covariance matrix also equals $1-\alpha$, for a given $\alpha \in (0,1)$. 
We describe the strategy to appropriately choose $\rho^2$ in Section \ref{subsec:cov_correction}.
The matrix inverse in the expressions for $\mu_j$ and $\mathbf{K}$ in \eqref{eq:hyper} is easily obtained by observing that $\wh{\mathbf{M}}^\top \wh{\mathbf{M}} = n\, \mathbb{I}_k$, 
from Proposition \ref{prop:mhat}.
We denote the posterior approximation in \eqref{eq:FABLESampling} by $\widetilde{\Pi}_j$, defined as
\begin{align}
\label{eq:pseudo-posterior}
    \begin{split}
\widetilde{\Pi}_j \left(\widetilde{\lambda}_j, \widetilde{\sigma}_j^{2} \right) & = \mbox{NIG}\left(\widetilde{\lambda}_j,  \widetilde{\sigma}_j^{2} \mid \mu_j, \rho^2 \mathbf{K}, \gamma_n / 2, \gamma_n \delta_j^2 / 2 \right)\\
& = N_k\left(\widetilde{\lambda}_j \mid \mu_j, \rho^2 \widetilde{\sigma}_j^2 \mathbf{K}\right) \mbox{IG}\left(\widetilde{\sigma}_j^2 \, \middle\vert \, \dfrac{\gamma_n}{2} , \dfrac{\gamma_n \delta_j^2}{2} \right), \quad 1 \leq j \leq p.
    \end{split}
\end{align}

To obtain independent posterior samples of $\Psi$, we first draw independent samples $(\widetilde{\lambda}_j, \widetilde{\sigma}_j^{2} ) \overset{\ind}{\sim} \widetilde{\Pi}_j$ for $j=1,\ldots,p$. We then let $\widetilde{\Lambda} = [\widetilde{\lambda}_1, \ldots, \widetilde{\lambda}_p]^{\T}, \widetilde{L} = \widetilde{\Lambda} \widetilde{\Lambda}^{\top}, \widetilde{\Sigma} = \mbox{diag}(\widetilde{\sigma}_1^{2}, \ldots, \widetilde{\sigma}_p^{2}),$ and $\widetilde{\Psi} = \widetilde{L} + \widetilde{\Sigma}$ denote the posterior samples of $\Lambda, L = \Lambda \Lambda^{\top}$, $\Sigma$, and $\Psi$, respectively. 
Our proposed Factor Analysis with BLEssing of dimensionality (FABLE) approach obtains independent posterior draws of the covariance matrix $\Psi$ in an embarrassingly parallel fashion, entirely bypassing the need to carry out MCMC. 
We refer to the posterior distribution arising from the sampling scheme in \eqref{eq:FABLESampling} as the FABLE-posterior throughout this paper.

The FABLE-posterior mean  $\widehat{\Psi}$ is available explicitly,
with its $(u,v)$th entry for $1 \leq u,v \leq p$ given by
\begin{align}
\label{eq:pseudoPostMean}
\begin{split}
\widehat{\Psi}_{uv} & = \left\{\begin{array}{ll}
\mu_u^{\top} \mu_v, & \text{if $u \neq v,$}\\
\|\mu_u\|_2^2 + \left(1 + \dfrac{k\rho^2}{n + \tau^{-2}}\right)\dfrac{\gamma_n \delta_u^2}{\gamma_n - 2}, & \text{if $u = v,$}
\end{array}
\right.
\end{split}
\end{align}
provided $\gamma_n > 2$. For finite $k, \rho^2,$ and $\tau^2$, the FABLE-posterior mean $\widehat{\Psi}$ can be approximated with $\widehat{\Psi} \approx G_0 G_0^{\top} + \Delta,$ where $G_0 = [\mu_1, \ldots, \mu_p]^{\top}$, $\Delta = \mbox{diag}(\delta_1^2, \ldots, \delta_p^2)$, and accuracy improves as $n$ increases. For simplicity, we refer to $\widehat{\Psi} = G_0 G_0^{\top} + \Delta$ as the FABLE-posterior mean of $\Psi$ in what follows.

\subsection{Pre-estimating the Latent Factors}
\label{subsec:preest}




We now describe the heuristic leading to \eqref{eq:Mhat}. 
From \eqref{eq:LFM-matrix}, the matrix $\mathbf{A} = \mathbf{Y}V / \sqrt{p}$ satisfies
$$\mathbf{A} = \mathbf{M} \, \left(\dfrac{\Lambda^{\T}V}{\sqrt{p}} \right) + \dfrac{EV}{\sqrt{p}} = \mathbf{M} \, C^\top + \dfrac{EV}{\sqrt{p}},$$
where $C = V^{\top} \Lambda / \sqrt{p} \in \mathbb{R}^{k \times k}$. 
Based on the consistency of spectral estimates \citep{chen2021spectral}, we expect $E$ to be approximately independent of $V$ as both $n, p$ grow. 
As a result, we expect $EV / \sqrt{p} \approx 0$ for increasing $p$. This leads us to
\begin{equation}
    \label{eq:A}
    \mathbf{A} \approx \mathbf{M} \, C^{\top}
\end{equation}
or equivalently, $a_i \approx C\eta_i$, where $a_i$ is the $i$th row of $\mathbf{A}$, for $i = 1, \ldots, n$. 
Since $\eta_i \sim N_k(0, \mathbb{I}_k),$ the marginal density of $a_i$ is approximately given by $a_i \overset{\iid}{\sim} N_k(0, CC^{\top})$, for $i = 1, \ldots, n$. 
This motivates the following estimator $\widehat{C}\widehat{C}^{\top}$ of $CC^{\top}$:
$$\widehat{C}\widehat{C}^{\top} = \dfrac{1}{n} \mathbf{A}^{\top} \mathbf{A} = \dfrac{D^2}{np},$$
leading us to \eqref{eq:Chat}. 
Given any invertible $\widehat{C}$ satisfying \eqref{eq:Chat}, we use \eqref{eq:A} to propose $\widehat{\M} = \mathbf{A} (\widehat{C}^{\top})^{-1}$ as an estimator of $\mathbf{M}.$ 
 In particular, the canonical choice $\widehat{\M} = \sqrt{n} U$ obtained with $\wh{C} = D / \sqrt{np}$ is the same as the estimator obtained by carrying out principal components analysis (PCA) on the matrix $\Y \Y^{\top}$ \citep{fan2023latent}. 
However, \eqref{eq:Mhat} permits other choices corresponding to different choices of $\widehat{C}$ satisfying \eqref{eq:Chat}, providing a general framework for obtaining estimates of the latent factors. 

A natural concern is whether different choices of $\widehat{\M}$ satisfying \eqref{eq:Mhat} would affect the FABLE-posterior distribution of $\widetilde{\Psi}$. The following result, with proof in Section \ref{app:propIdentify} of the Supplementary Material, ensures that this is not the case.
\begin{proposition}
\label{prop:mhat}
    (i) For any $\widehat{\mathbf{M}}$ satisfying \eqref{eq:Mhat}, we have $\widehat{\mathbf{M}}^{\T} \widehat{\mathbf{M}} = n \mathbb{I}_k$ and $\widehat{\mathbf{M}} \widehat{\mathbf{M}}^{\T} = n UU^{\T}.$
(ii) The FABLE-posterior distribution of $\widetilde{L}, \widetilde{\Sigma},$ and $\widetilde{\Psi}$ obtained from the {\em FABLE} approach only depends on $\widehat{\M}$ through  $\widehat{\M}^{\top}\widehat{\M}$ and $\widehat{\M}\widehat{\M}^{\top}$.
\end{proposition}
\noindent Unless mentioned otherwise, we let $\widehat{\mathbf{M}} = \sqrt{n} U$ from here on.

\subsection{Coverage Correction}
\label{subsec:cov_correction}

Based on numerous simulations, we observed that after sampling
$$\widetilde{\lambda}_j \mid \widetilde{\sigma}_j^2 \sim N_k \left(\mu_j,  \widetilde{\sigma}_j^2 \mathbf{K} \right), \quad \widetilde{\sigma}_j^2 \sim \mbox{IG}\left(\gamma_n/2, \gamma_n \delta_j^2 / 2 \right),$$
the credible intervals of the entries of $\widetilde{\Psi}$ usually underestimated the uncertainty associated with estimating the entrywise elements of the true covariance matrix, under repeated sampling. 
To alleviate this issue, the FABLE algorithm introduces a coverage-correction coefficient $\rho^2 \geq 1$ in the conditional variance when sampling the factor loadings $\widetilde{\lambda}_j$ conditional on $\wt \sigma_j^2$s, as described in \eqref{eq:FABLESampling}. 
This coefficient $\rho^2$ is adaptively chosen to restore nominal average frequentist coverage over all the entries.
The FABLE-posterior mean $\wh \Psi$ also depends on $\rho$, but the dependence diminishes with growing $n$, as highlighted in \eqref{eq:pseudoPostMean}.
 In contrast, the effect of choosing $\rho$ on UQ remains even with large $n, p$.
 Thus, we 
 focus on guaranteeing valid UQ as the metric when choosing $\rho$.

We choose $\rho$ to ensure that a summary statistic, such as the average or minimum, of the $\binom{p}{2}$ entrywise asymptotic coverages equals $1-\alpha$ and describe the methodology as follows.
First, for $1 \leq u \leq v \leq p$, let $\wh{q}_{uv}(\rho)$ represent the estimated asymptotic coverage as a function of $\rho$ for the $(u,v)$th entry of the true covariance when estimated by $\wt \Psi_{uv} = \wt \lambda_u^\top \wt \lambda_v + \wt{\sigma}_u^2 \, \mathbbm{1}(u = v)$. 
We derive the form of $\wh{q}_{uv}(\rho)$ and show that it is monotonically increasing for $\rho \geq 1$ in Section \ref{subsec:theoryUQ}.
Next, we let
\begin{align}
\label{eq:b_uv}
\begin{split}
b_{uv} & = \left\{\begin{array}{ll}
\left(1 + \dfrac{\|\mu_u\|_2^2 \|\mu_v\|_2^2 + (\mu_u^{\top} \mu_v)^2}{\mathcal{V}_u^2 \|\mu_v\|^2 + \mathcal{V}_v^2 \|\mu_u\|_2^2}\right)^{1/2}, & \text{if $u \neq v,$}\vspace{1.3em}\\ 
\left(1 + \dfrac{\|\mu_u\|_2^2}{2 \mathcal{V}_u^2}\right)^{1/2}, & \text{if $u = v,$}
\end{array}
\right.
\end{split}
\end{align}
for $1 \leq u, v \leq p$, with $\mu_j = \sqrt{n} \, U^{\T} y^{(j)} / (n + \tau^{-2})$ and $\mathcal{V}_j^2 = \|(\mathbb{I}_n - UU^{\top})y^{(j)}\|_2^2 / n$ for $j =1,\ldots, p$. 
The $b_{uv} \geq 1$ are defined so that $\wh{q}_{uv}(b_{uv}) = 1 - \alpha.$ 
That is, $b_{uv}$ represents the coefficient for correcting the coverage of the $(u,v)$th entry of the true covariance with $\wt \Psi_{uv}$.
A computationally convenient expression for $\wh{q}_{uv}(\rho)$ entirely in terms of $b_{uv}$ is available in Section \ref{suppsec:algorithm} of the Supplementary Material.
As a consequence of the monotonicity of $\wh{q}_{uv}(\rho)$ for $\rho \geq 1$, we can ensure that all the entrywise asymptotic coverages are at least $1-\alpha$ by letting $\rho = \sup_{u,v} b_{uv}$, as $\wh{q}_{uv}(\sup_{j,j'} b_{jj'}) \geq \wh{q}_{uv}(b_{uv}) = 1 - \alpha$ for any given $1 \leq u \leq v \leq p.$
Alternatively, we can choose $\rho$ to ensure that the average coverage over all the $p(p+1)/2$ distinct entries of the covariance equals $1 - \alpha$. This is obtained by solving $\rho$ such that
\begin{equation}
    \label{eq:rhoChoiceMean}
    \dfrac{1}{p(p+1)/2} \sum_{1 \leq u \leq v \leq p} \wh{q}_{uv}(\rho) = 1  - \alpha.
\end{equation}
In general, we recommend solving \eqref{eq:rhoChoiceMean} to obtain $\rho$.
We provide provably accurate guarantees on frequentist coverage of the proposed approach in Section \ref{subsec:theoryUQ}.

\subsection{Hyperparameter Choice}
\label{subsec:hyperchoice}

\subsubsection{Tuning $k$}
\label{subsubsec:tunek}

In this Subsection, we explicitly indicate the dependence of all relevant quantities on $k$.
To estimate $k$, we implement the approach described in \cite{chen2022determining}, based on minimizing a joint-likelihood-based information criterion $\mbox{JIC}(k)$ by adding the penalty $\nu(k, n, p) = k (n \vee p) \log(n \wedge p)$ to twice the fitted negative log-likelihood, obtained as follows.
From \eqref{eq:LFM-matrix}, the joint likelihood function of $\left(\mathbf{M}^{(k)}, \Lambda^{(k)}, \Sigma\right)$ for a given $k$ is 
   $$\mathcal{L}\left(\mathbf{M}^{(k)}, \Lambda^{(k)}, \Sigma\right) = \prod_{i=1}^{n} \prod_{j=1}^{p} \dfrac{1}{\sqrt{2\pi} \sigma_j} \exp\left\{-\dfrac{1}{2\sigma_j^2} \left(y_{ij} - \mu_{ij}^{(k)}\right)^2\right\},$$ 
   where $\mu_{ij}^{(k)}$ is the $(i, j)$th entry of the signal matrix $\mathbf{M}^{(k)} \Lambda^{(k)\top}.$
Although $\mathbf{M}^{(k)} \in \mathbb{R}^{n \times k}$ and $\Lambda^{(k)} \in \mathbb{R}^{p \times k}$ have different dimensions as $k$ changes, the signal matrix $\mathbf{M}^{(k)} \Lambda^{(k)\top} \in \mathbb{R}^{n \times p}.$
Motivated from the SVD of $\mathbf{Y}$ for a given $k$ as in \eqref{eq:SVD}, we estimate $\mathbf{M}^{(k)}\Lambda^{(k)\top}$ with  $\wh{\mathbf{M}^{(k)}\Lambda^{(k)\top}} = U^{(k)} D^{(k)} V^{(k)\top}$.
Next, we consider the quantity $\mathcal{L}\left(\widehat{\mathbf{M}^{(k)}}, \wh{\Lambda^{(k)}}, \Sigma\right)$ as a function of $\Sigma = \mbox{diag}\left(\sigma_1^2, \ldots, \sigma_p^2\right)$, obtained by plugging the signal estimator into the joint likelihood:
$$\mathcal{L}\left(\widehat{\mathbf{M}^{(k)}}, \wh{\Lambda^{(k)}}, \Sigma\right) = \prod_{i=1}^{n} \prod_{j=1}^{p} \dfrac{1}{\sqrt{2\pi}\sigma_j} \exp\left\{-\dfrac{1}{2\sigma_j^2} \left(y_{ij} - \hat{\mu}_{ij}^{(k)}\right)^2 \right\},$$ 
   where $\hat{\mu}_{ij}^{(k)}$ is the $(i,j)$th entry of $\widehat{\mathbf{M}^{(k)} \Lambda^{(k)\top}} = U^{(k)} D^{(k)} V^{(k)\top}$. 
   Maximizing $\mathcal{L}(\wh{\mathbf{M}^{(k)}}, \wh{\Lambda^{(k)}}, \Sigma)$ with respect to $\sigma_1^2, \ldots, \sigma_p^2$, 
   we estimate $\sigma_j^2$ with
   $\hat{\sigma}_j^{(k)2} =  \sum_{i=1}^{n} \left(y_{ij} - \hat{\mu}_{ij}^{(k)}\right)^2 / n$ for $j=1,\ldots,p$.
   Finally, with $\wh{\Sigma^{(k)}} = \mbox{diag}\left(\hat{\sigma}_1^{(k)2}, \ldots, \hat{\sigma}_p^{(k)2}\right)$, the fitted likelihood is given by 
   $$\mathcal{L}\left(\widehat{\mathbf{M}^{(k)}}, \wh{\Lambda^{(k)}}, \wh{\Sigma^{(k)}}\right) = \prod_{i=1}^{n} \prod_{j=1}^{p} \dfrac{1}{\sqrt{2\pi}\hat{\sigma}_j^{(k)}} \exp\left\{-\dfrac{\left(y_{ij} - \hat{\mu}_{ij}^{(k)}\right)^2}{2\hat{\sigma}_j^{(k)2}}  \right\} = \dfrac{(2\pi e)^{-np/2}}{\prod_{j=1}^{p} \hat{\sigma}_j^{(k)n}}.$$
   Let $\wh {\mathcal{L}}_k := \mathcal{L}\left(\widehat{\mathbf{M}^{(k)}}, \wh{\Lambda^{(k)}}, \wh{\Sigma^{(k)}}\right)$. 
   The JIC as a function of $k$ is given by 
   $$\mbox{JIC}(k) = -2 \log \wh{\mathcal{L}}_k + k (n \vee p) \log (n \wedge p).$$
   Given a plausible upper bound $\mathcal{K}_0$ for the possible values of $k$, we estimate $\widehat{k} = \min_{1 \leq k \leq \mathcal{K}_0} \mbox{JIC}(k).$ 

To choose 
$\mathcal{K}_0$, we consider the singular values $s_1 \geq s_2 \geq \ldots \geq s_{n \wedge p}$ of $\mathbf{Y}$ and let $\mathcal{K}_0$ be the smallest integer $\mathcal{K}$ such that $\left(\sum_{j=1}^{\mathcal{K}} s_{j} \right) / \left(\sum_{j=1}^{n \wedge p} s_j\right) \geq S_0$ for a given $S_0 \in (0,1).$
Based on our experience with numerous simulation experiments, setting $S_0 = 0.75$ as a default simultaneously allows
the search space $\{1, 2, \ldots, \mathcal{K}_0\}$ to be sufficiently large while avoiding the boundary, and maintains computational efficiency of the procedure, leading to a good performance of $\wh k.$
In general, we would not recommend choosing $S_0$ to be as large as possible, as candidate values of $k$ near the boundary can cause near-singular fits \citep{ichikawa1988empirical}.

Although the JIC is similar in spirit to the Akaike information criterion (AIC) or the Bayesian information criterion (BIC) \citep{akaike1974new, schwarz1978estimating}, there are some key differences.
The AIC/BIC are routinely used in existing literature on rank selection in factor models based on the marginal likelihood
$y_i \mid \Lambda^{(k)}, \Sigma \overset{\iid}{\sim} N_p(0, \Lambda^{(k)} \Lambda^{(k)\top} + \Sigma),$
obtained by marginalizing out $\eta_i^{(k)} \overset{\iid}{\sim} N_k(0, \mathbb{I}_k)$ \citep{akaike1987factor}. 
In contrast, the JIC is based on the joint likelihood of $\left(\mathbf{M}^{(k)}, \Lambda^{(k)}, \Sigma\right)$ and bypasses the inversion of $\Psi^{(k)}.$
Additionally, the JIC generalizes elegantly to non-Gaussian errors or non-linear outcome measurement structures where marginal likelihoods may not be explicitly available but joint likelihoods are, thus facilitating extensions of the proposed method.
We provide a thorough comparison of the JIC with the AIC and BIC, including simulation studies, in Section \ref{suppsec:JICvsAICBIC} of the Supplementary Material.
\subsubsection{Tuning Prior Hyperparameters}
\label{subsubsec:tunetau}

After estimating $k$, we next estimate $\tau^2$, which can be interpreted as the global shrinkage parameter of the loadings. For this, we employ an empirical Bayes (EB) approach. Since $\widetilde{\lambda}_j \mid \widetilde{\sigma}_j^2, \tau^2 \sim N_k\left(0, \tau^2 \widetilde{\sigma}_j^2 \mathbb{I}_k\right)$ for $j=1,\ldots,p$ {\it a priori}, we obtain an estimate of $\tau^2$ by conditioning on $(\widetilde{\lambda}_j, \widetilde{\sigma}_j^2 )$ and then maximizing the conditional likelihood, leading to
$\widehat{\tau}_0^2 =  \sum_{j=1}^{p} \left({\|\widetilde{\lambda}_j\|_2^2}/{\widetilde{\sigma}_j^2}\right) / (kp).$
Following from  \eqref{eq:hyper} and \eqref{eq:pseudo-posterior}, we simply estimate $k$ by $\widehat{k}$ as before, $\|\widetilde{\lambda}_j\|_2^2$ by $ \mathcal{L}_j^2 = \|U^{\top} y^{(j)}\|_2^2 / n$, and $\widetilde{\sigma}_j^2$ by $\mathcal{V}_j^2 = \|(\mathbb{I}_n - UU^{\top}) y^{(j)}\|_2^2 / n$ for $j=1,\ldots,p$, where $U$ is dependent on $\widehat{k}$.
This leads us to the plug-in estimate of $\tau^2$, given by
\begin{equation}
    \label{eq:tausq-est}
    \widehat{\tau}^2 = \dfrac{1}{p\widehat{k}} \sum_{j=1}^{p} \dfrac{\mathcal{L}_j^2}{\mathcal{V}_j^2}.
\end{equation}
The consistency of the point estimators $\mathcal{L}_j^2$ and $\mathcal{V}_j^2$ for $1 \leq j \leq p$ is proven in Section \ref{proof:lemmaPostCon}, Lemma \ref{lemma:post-mean-conv} of the  Supplementary Material. 
As long as $\wh{k} \to k$ as $n \to \infty$ and we assume that $(1/p) \sum_{j=1}^{p} \mathcal{L}_j^2 / \mathcal{V}_j^2$ converges to a positive and finite quantity $\|\lambda_0\|_2^2 / \sigma_0^2$ as $p \to \infty$, we obtain $\wh{\tau}^2 \to (1/k) \left(\|\lambda_0\|_2^2 / \sigma_0^2\right)$. 
Here, $\|\lambda_0\|_2^2 / \sigma_0^2$ may be interpreted as the true average signal-to-noise ratio (SNR). 
Thus, the EB estimate $\wh{\tau}^2$ automatically adapts to the appropriate SNR underlying the data, with lower SNR settings producing smaller $\wh{\tau}^2$ that lead to smaller loadings estimates. 
In simulations, we obtained better estimation accuracy using the EB estimate $\wh{\tau}^2$, instead of assuming a vague prior for the loadings $\wt \lambda_j$ by letting $\tau^2 \to \infty$.

Lastly, we observed that the FABLE procedure is not sensitive to the hyperparameters of the prior for the variances as in \eqref{eq:prior}, namely $\left(\gamma_0, \delta_0^2\right)$. 
Results for a sensitivity analysis by varying $\left(\gamma_0, \delta_0^2 \right)$ are available in Section \ref{suppsec:invgamma} of the Supplementary Material.
Thus, we use the default values $\gamma_0 = \delta_0^2 = 1$ in practice.
Once the number of factors $k$ and the common variance $\tau^2$ have been estimated using $\widehat{k}$ and $\widehat{\tau}^2$, respectively, we then proceed to obtain FABLE-posterior samples of $\Psi$ as described earlier. We provide a summary of implementing the proposed method in Algorithm \ref{algorithm:fable}, available in Section \ref{suppsec:algorithm} of the Supplementary Material.

\section{Theoretical Support}
\label{sec:theory}

\subsection{Setup and Assumptions}
\label{subsec:theorySetup}
In this Section, we provide theoretical guarantees on the FABLE-posterior for high-dimensional covariance matrices. Most of the existing literature \citep{pati2014posterior, bhattacharya2011sparse, rovckova2016fast, srivastava2017expandable} on
frequentist asymptotic guarantees for Bayesian factor models focuses on posterior contraction rates. 
Theorems on frequentist coverage of Bayesian credible sets for such models are lacking.
\cite{xie2023eigenvector} consider a signal-plus-noise model and derive theoretical properties of a quasi-posterior based approach, with guarantees on asymptotic frequentist coverage of the quasi-posterior intervals. 
However, their model only considers a square signal matrix, with a different underlying focus from high-dimensional covariance estimation. Due to the (innovative) form of the FABLE-posterior, we cannot use established tools for showing Bernstein-von Mises-type results in Bayesian models.
We overcome this challenge by leveraging on a blessing of dimensionality phenomenon which requires both $n$ and $p = p_n$ to grow, providing results on both  contraction and UQ of the FABLE-posterior. The blessing of dimensionality is key for accurate estimation of the latent factor subspace up to rotational ambiguity. 
The proofs of all the theorems in this Section can be found in Sections \ref{proof:post-con}-\ref{proof:asymp-law} of the Supplementary Material.

We assume the following data-generating model:
\begin{equation}
    \label{eq:data-gen}
    y_i = \Lambda_0 \eta_{0i} + \epsilon_i,
\end{equation}
where $\epsilon_i \overset{\iid}{\sim} N_p(0, \Sigma_0)$ for $i=1,\ldots,n$ with $\Sigma_0 = \mbox{diag}(\sigma_{01}^2, \ldots, \sigma_{0p}^2),$ and $\eta_{0i} \overset{\iid}{\sim} N_k(0, \mathbb{I}_k)$ for $i=1,\ldots,n$. 
Here, $\Lambda_0 = [\lambda_{01}, \ldots, \lambda_{0p}]^{\top}$ is the true factor loadings matrix while $\eta_{0i}$ are the true latent factors; integrating them out provides $y_i \overset{\iid}{\sim} N_p(0, \Lambda_0 \Lambda_0^{\top} + \Sigma_0)$ for $i=1,\ldots,n$ as the marginal distribution of the data. 
Let $\Y = [y_1,\ldots,y_n]^{\top} \in \mathbb{R}^{n \times p}$ and $M_0 = [\eta_{01},\ldots,\eta_{0n}]^{\top} \in \mathbb{R}^{n \times k}$ be the data matrix and the true matrix of latent factors, respectively, and $E = [\epsilon_1, \ldots, \epsilon_n]^\top$, so that the true data generating model may be written as 
\begin{equation}
    \label{eq:data-gen-mat}
    \Y = M_0 \Lambda_0^{\top} + E.
\end{equation}
Our primary goal is inference for the covariance matrix $\Psi_0 = \Lambda_0 \Lambda_0^{\T} + \Sigma_0.$ For a matrix $A \in \mathbb{R}^{n_1 \times n_2}$, we denote its singular values by $s_1(A) \geq \ldots \geq s_{n_1 \wedge n_2}(A)$. 
Let $\|A\| = \underset{\|x\|_2 = 1}{\sup} \|Ax\|_{2} = s_1(A)$ denote the operator norm of $A$ and $\|A\|_{\infty} = \max_{ij} |A_{ij}|$ denote the max-norm of $A$, with $A_{ij}$ denoting the $(i,j)$th entry of $A$. 
For two sequences $a_m, b_m \geq 0$, we say $a_m \asymp b_m$ if $a_m = \mathcal{O}(b_m)$ and $b_m = \mathcal{O}(a_m)$ as $m \to \infty$. 
We assume the following conditions on the true data-generating model:
\begin{assumption}
\label{assumption:dim}
    $p_n \to \infty$ and $(\log p_n) / n = o(1)$
    as $n \to \infty$. 
    
\end{assumption}
\begin{assumption}
\label{assumption:norm}
    $s_k(\Lambda_0) \asymp \|\Lambda_0\| \asymp \sqrt{p_n}$ as $n \to \infty$, $\|\Lambda_0\|_{\infty} = \mathcal{O}(1)$, and $\underset{1 \leq j \leq p_n} {\min}\|\lambda_{0j}\|_2 \geq c_1$ for some finite constant $c_1 > 0$.
\end{assumption}
\begin{assumption}
\label{assumption:var}
    The true error variances satisfy $\underset{1 \leq j \leq p_n}{\max} \sigma_{0j}^2 = \mathcal{O}(1)$ and $\underset{1 \leq j \leq p_n}{\min} \sigma_{0j}^2 \geq c_2$ for some finite constant $c_2 > 0$.
\end{assumption}
\begin{assumption}
\label{assumption:hyper}
    The hyperparameters $k, \tau^2, \gamma_0, \delta_0^2$, and $\rho$ are fixed constants.
\end{assumption}
\begin{assumption}
\label{assumption:uq}
    $\sqrt{n} / p_n \to 0$ as $n \to \infty$.
\end{assumption}

Such assumptions are standard in the literature on the asymptotic properties of latent factor models \citep{pati2014posterior, bhattacharya2011sparse, xie2022bayesian, rovckova2016fast}.  
Assumption \ref{assumption:dim} allows the number of dimensions $p_n$ to scale as any polynomial function of $n$. 
Assumption \ref{assumption:norm} ensures that the true loadings matrix $\Lambda_0$ is well-conditioned, with the low-rank portion $\Lambda_0 \Lambda_0^{\top}$ identifiable from noise in the asymptotic regime, and also ensures that none of the columns of $\mathbf{Y}$ consist purely of noise. 
Assumptions \ref{assumption:var} and \ref{assumption:hyper} assume the scalar error variances and model hyperparameters are finite. 
Assumption \ref{assumption:uq} imposes a lower bound on the rate at which $p_n$ increases and is necessary for the ``blessing of dimensionality'' to take effect rapidly enough so that asymptotically accurate UQ can be obtained.
For our theoretical requirements, we will assume that $k, \tau^2,$ and $\rho$ are known and fixed.
To ease exposition, we will often suppress the dependence of $p = p_n$ on $n$ and express it when needed.

Although Assumptions \ref{assumption:dim}-\ref{assumption:hyper} are sufficient to establish the posterior concentration of relative estimation errors with FABLE, we additionally require the stronger Assumption \ref{assumption:uq} to obtain accurate UQ guarantees for entrywise FABLE credible intervals. 
This assumption is standard in existing literature on UQ.
For example, Theorem 2 in \cite{bai2003inferential} utilizes Assumption \ref{assumption:uq} to obtain the limiting distribution of estimated factor loadings.
Compared to assumptions in the literature on matrix denoising theory, such as $n / p_n \to c \in (0,1) \text{ or } (0, \infty)$ \citep{ledoit2004well, wang2017asymptotics,hong2023optimally, benaych2012singular}, Assumption \ref{assumption:uq} is substantially weaker.
In general, it is common to assume stronger conditions to obtain UQ guarantees than those needed for consistency or rate results for high-dimensional covariance or precision estimation \citep{van2014asymptotically, jankova2015confidence, ning2017general, yao2023rates}.

Let $M_0 \Lambda_0^{\top} = U_0 D_0 V_0^{\top}$ be the singular value decomposition of the signal, with $U_0 \in \mathbb{R}^{n \times k}, V_0 \in \mathbb{R}^{p \times k}$ having orthonormal columns and $D_0 \in \mathbb{R}^{k \times k}$ a diagonal matrix of positive singular values. 
Suppose the singular value decomposition of $\Y$ is
$$\Y = UDV^{\top} + U_{\perp} D_{\perp} V_{\perp}^{\top},$$
where $U \in \mathbb{R}^{n \times k}, V \in \mathbb{R}^{p \times k}$ have orthonormal columns, and $D \in \mathbb{R}^{k \times k}$ contains the $k$ largest singular values of $\Y$. 
We first provide a result showcasing the blessing of dimensionality when estimating $U_0 U_0^{\top}$ by $UU^{\top}$ along with a few other key results, which form the basis of the theorems that follow.
Let us denote the induced FABLE-posterior measure, the true data-generating measure, and the expectation under the true data-generating measure by $\widetilde{{\Pi}}$, $P_0$, and $E_0$, respectively.
\begin{proposition}
\label{proposition:bod}
    Suppose Assumptions \ref{assumption:dim} -- \ref{assumption:hyper} hold. 
    Fix $u$ such that $1 \leq u \leq p_n$.
    Then, there exist constants $G_1, G_2, G_3 > 0$ such that 
    \begin{enumerate}
        \item[(a)] $\underset{n \to \infty}{\lim}P_0\left\{\|UU^{\top} - U_0U_0^{\top}\| > G_1\left(\dfrac{1}{\sqrt n} + \dfrac{1}{\sqrt{p_n}} \right)\right\} = 0.$
    \item[(b)]  $\underset{n \to \infty}{\lim}P_0\left\{\|M_0^\top(UU^{\top} - U_0U_0^{\top})M_0\| > G_2\left(1 + \dfrac{n}{p_n} \right)\right\} = 0.$
    \vspace*{0.2cm}
    \item[(c)] $\underset{n \to \infty}{\lim}P_0\left[ \|(UU^{\top} - U_0U_0^{\top})\, \epsilon^{(u)}\|_2 > G_3\left\{\sqrt{\log n}\left(\dfrac{1}{\sqrt{n}} + \dfrac{1}{\sqrt{p_n}}\right) + \dfrac{\sqrt{n}}{p_n} \right\}\right] = 0$, where $\epsilon^{(u)}$ is the $u$th column of the error matrix $E$.
    \end{enumerate}
\end{proposition}

The proof of Proposition \ref{proposition:bod} is available in Section \ref{proofSubSec:bod} of the Supplementary Material. 
First, in Theorem \ref{theorem:post-con}, we provide FABLE-posterior contraction rates when estimating $\Lambda_0 \Lambda_0^{\T}, \Sigma_0, $ and $\Psi_0$ using $\wt{L} = \wt \Lambda \wt{\Lambda}^{\top}, \wt{\Sigma}$, and $\wt{\Psi} = \wt{L} + \wt\Sigma$, respectively. 
Later, using Theorems \ref{theorem:bvm} and \ref{theorem:asymp-law}, we justify the UQ of the entrywise elements of $\Psi_0$ with the FABLE-posterior credible intervals. 

\subsection{FABLE-posterior Contraction Rates}
\label{subsec:theoryContraction}

\begin{theorem}
\label{theorem:post-con}
Suppose Assumptions \ref{assumption:dim} -- \ref{assumption:hyper} hold. Then, as $n \to \infty$, there exist finite constants $D_1, D_2, D > 0$ such that
\begin{enumerate}
\item[(a)] $E_0\left[\wt{\Pi}\left\{\dfrac{\|\wt L - \Lambda_0 \Lambda_0^{\top}\|}{\|\Lambda_0 \Lambda_0^{\top}\|} > D_1 \left(\sqrt{\dfrac{\log n}{n}} + \dfrac{1}{\sqrt{p_n}}\right) \right\}\right] \to 0.$
\item[(b)] 
\vspace*{0.2cm}
$E_0\left[\wt{\Pi}\left\{\|\wt \Sigma - \Sigma_0\| > D_2 \left(\left(\dfrac{\log p_n}{n}\right)^{1/3} + \dfrac{1}{\sqrt{p_n}}\right) \right\}\right]\to 0.$
\vspace{0.2cm}
\item[(c)] $E_0\left[\wt{\Pi}\left\{\dfrac{\|\wt \Psi - \Psi_0\|}{\|\Psi_0\|} > D \left(\sqrt{\dfrac{\log n}{n}} + \dfrac{1}{\sqrt{p_n}}\right)\right\}\right] \to 0.$
\end{enumerate}
\end{theorem}
\noindent Theorem \ref{theorem:post-con} shows the FABLE-posterior concentration of the relative errors around $0$ when estimating $\Lambda_0 \Lambda_0^\top, \Sigma_0,$ and $\Psi_0$.
The result in part (b) is also a statement on relative error when estimating $\Sigma_0$, as $0 < c_2 \leq \| \Sigma_0 \| = \max_{1 \leq j \leq p_n}\,  \sigma_{0j}^2 = \mathcal{O}(1)$ from Assumption \ref{assumption:var}.
From part (c), the relative error in estimating the true covariance matrix converges to $0$ at the rate of $n^{-1/2} + p_n^{-1/2}$, up to logarithmic factors. 
This showcases the blessing of dimensionality, as concentration of the FABLE-posterior is obtained if and only if both the number of samples $n$ and the number of dimensions $p_n$ increase. 
The proof of Theorem \ref{theorem:post-con} is available in Section \ref{proof:post-con} of the Supplementary Material.


\subsection{Uncertainty Quantification}
\label{subsec:theoryUQ}

We now consider UQ of entrywise elements of the covariance matrix. 
For fixed indices $(u,v)$ such that $1 \leq u, v \leq p$, let $\wt{\Psi}_{uv}$ and $\Psi_{0,uv}$ denote the $(u,v)$th element of $\wt{\Psi}$ and $\Psi_0$, respectively.
Then, we have $\Psi_{0,uv} = \lambda_{0u}^{\top} \lambda_{0v} + \sigma_{0u}^2 \mathbbm{1}(u=v)$, where $\mathbbm{1}$ denotes the indicator function. 
We first discuss theoretical results for a fixed value of the coverage-correction coefficient $\rho$ and then provide an approach 
to choose an appropriate $\rho$ that guarantees valid UQ for FABLE.
We let $T_{uv} = \mu_u^{\top} \mu_v + \delta_u^2 \mathbbm{1}(u=v)$ be an estimator of $\Psi_{0,uv}$, with $\mu_u, \delta_u^2$ as defined in \eqref{eq:hyper}.
This is the $(u,v)$-th entry of $\wh \Psi = G_0G_0^\top + \Delta$ as in \eqref{eq:pseudoPostMean}.
Then, for any fixed $\rho \geq 1$, the following result approximates the FABLE-posterior of $\wt{\Psi}_{uv}$ with a suitable Gaussian distribution centered at $T_{uv}$, as both $n$ and $p_n$ increase. Let $\Phi$ denote the cumulative distribution function of the $N(0,1)$ distribution.
\begin{theorem}
\label{theorem:bvm}
Suppose Assumptions \ref{assumption:dim} -- \ref{assumption:hyper} hold. Fix $\rho \geq 1$ and $1 \leq u, v \leq p_n$.
Let
\begin{align*}
\begin{split}
l_{0,uv}^2 (\rho) & = \left\{\begin{array}{ll}
      \rho^2(\sigma_{0v}^2 \|\lambda_{0u}\|_2^2 + \sigma_{0u}^2 \|\lambda_{0v}\|_2^2), \quad & \text{for $u \neq v$},\\
     2\sigma_{0u}^4 + 4\rho^2 \, \sigma_{0u}^2 \|\lambda_{0u}\|_2^2, \quad & \text{for $u = v$.}
\end{array}
\right.
\end{split}
\end{align*}
Then, as $n \to \infty$,
    $$\sup_{x \in \mathbb{R}} \, \left\vert \wt{\Pi}\left\{\dfrac{\sqrt{n}(\wt{\Psi}_{uv} - T_{uv})}{ l_{0,uv}(\rho)} \leq x\right\} - \Phi(x) \right\vert \overset{P_0}{\to} 0.$$ 
\end{theorem}
\noindent Theorem \ref{theorem:bvm} allows us to approximate the asymptotic FABLE-posterior distribution of each element of the covariance matrix, after suitable centering and scaling, using a Gaussian distribution with mean $0$ and variance $l_{0,uv}^2(\rho)$. The proof of Theorem \ref{theorem:bvm} is in Section \ref{proof:bvm} of the Supplementary Material. 

We next state a result regarding the asymptotic law of the quantity $\sqrt{n}\,(T_{uv} - \Psi_{0, uv})$ and then illustrate how this result can be used to show asymptotic frequentist validity of entrywise  credible intervals for $ \Psi_{0, uv}$. We will require Assumption \ref{assumption:uq} together with Assumptions \ref{assumption:dim}-\ref{assumption:hyper} that were used to establish Theorems \ref{theorem:post-con} and \ref{theorem:bvm}.
For generic random variables $X_n$ and $X$, we denote $X_n$ converging in distribution to $X$ by $X_n \implies X.$
\begin{theorem}
    \label{theorem:asymp-law}
    Suppose Assumptions \ref{assumption:dim}--\ref{assumption:uq} hold. Fix $1 \leq u, v \leq p_n$. Let 
    \begin{align*}
\begin{split}
\mathcal{S}_{0,uv}^2 & = \left\{\begin{array}{ll}
      \sigma_{0v}^2 \|\lambda_{0u}\|_2^2 + \sigma_{0u}^2 \|\lambda_{0v}\|_2^2 + \|\lambda_{0u}\|_2^2 \|\lambda_{0v}\|_2^2 + (\lambda_{0u}^{\top} \lambda_{0v})^2, \quad & \text{for $u \neq v$},\\
     2(\|\lambda_{0u}\|_2^2 + \sigma_{0u}^2)^2, \quad & \text{for $u = v$.}
\end{array}
\right.
\end{split}
\end{align*}
Then, as $n \to \infty$, one has 
    $\sqrt{n}(T_{uv} - \Psi_{0, uv}) / \mathcal{S}_{0,uv} \implies N(0, 1).$ 
\end{theorem}
\noindent The proof of Theorem \ref{theorem:asymp-law} is available in Section \ref{proof:asymp-law} of the Supplementary Material. 
Under Theorem \ref{theorem:bvm}, the $100(1-\alpha) \%$ asymptotic credible interval of $\wt{\Psi}_{uv}$ is given by $$\mathcal{C}_{uv}(\rho) = \left[T_{uv} - z_{1 - (\alpha/2)}\, \dfrac{ l_{0,uv}(\rho)}{\sqrt{n}}, \, T_{uv} + z_{1 - (\alpha/2)} \,\dfrac{ l_{0,uv}(\rho)}{\sqrt{n}} \right],$$
where $z_{1 - (\alpha/2)} = \Phi^{-1}\left\{1 - (\alpha/2)\right\}.$
Thus, the probability of $\mathcal{C}_{uv}(\rho)$ covering $\Psi_{0, uv}$ under repeated sampling is given by
\begin{align}
\label{eq:asympcoverage}
    P_0\left\{\Psi_{0,uv} \in \mathcal{C}_{uv}(\rho)\right\} & = P_0\left\{\dfrac{\sqrt{n} \, \left| T_{uv} - \Psi_{0,uv} \right|}{\mathcal{S}_{0,uv}} \leq  z_{1 - (\alpha/2)} \dfrac{l_{0,uv}(\rho)}{\mathcal{S}_{0,uv}}\right\} \notag \\
    & \to q_{uv}(\rho) := 2\Phi\left\{z_{1 - (\alpha/2)} \dfrac{l_{0,uv}(\rho)}{\mathcal{S}_{0,uv}}\right\} - 1,
\end{align}
as $n \to \infty,$ using Theorem \ref{theorem:asymp-law}. 

We next obtain a consistent estimator of the asymptotic coverage $q_{uv}(\rho)$ for any $\rho > 0.$ 
This will allow us to estimate $\rho$ as outlined in Section \ref{subsec:cov_correction}.
First, we define the quantities $\wh{l}_{0,uv}(\cdot)$ and $\wh{\mathcal{S}}_{0,uv}$ as plug-in estimators of $l_{0,uv}(\cdot)$ and $\mathcal{S}_{0,uv}$, respectively. 
These estimators are obtained by replacing $\lambda_{0u}$ and $\sigma_{0u}^2$ by the
estimators $\mu_u$ and $\mathcal{V}_u^2$, respectively, as defined in Section \ref{subsec:cov_correction}. 
Lemma \ref{lemma:post-mean-conv} in Section \ref{proof:lemmaPostCon} of the Supplementary Material provides a proof of their consistency. 
Next, we define $\wh{q}_{uv}(\rho)$ by replacing $l_{0,uv}(\rho)$ and $\mathcal{S}_{0,uv}$ with $\wh{l}_{0,uv}(\cdot)$ and $\wh{\mathcal{S}}_{0,uv}$, respectively, in \eqref{eq:asympcoverage}.
Since $\wh{l}_{0,uv}(\rho)$ is monotonically increasing in $\rho$ for $\rho \geq 1$, so is $\wh{q}_{uv}(\rho).$
The function $\wh{q}_{uv}(\rho)$ may now be utilized to obtain a data-adaptive estimate of $\rho$, starting from the fact that $\wh{q}_{uv}(b_{uv}) = 1 - \alpha$ and then following the discussion in Section \ref{subsec:cov_correction}.
The expression of $b_{uv}$ as in \eqref{eq:b_uv} is obtained by solving the equation $\wh{l}_{0,uv}(\rho) = \wh{S}_{0, uv}$ for $\rho$, for $1 \leq u \leq v \leq p$.
Since Lemma \ref{lemma:Bmat} in Section \ref{proof:lemmaPostCon} of the Supplementary Material ensures that $\underset{1 \leq u \leq v \leq p_n}{\max} b_{uv} = O_{P_0}(1)$, all the choices of $\rho$ as described in Section \ref{subsec:cov_correction} are bounded.

\section{Simulation Results}
\label{sec:simulation}

\subsection{Setup}
\label{subsec:preliminaries}

In this Section, we compare the performances of FABLE with competitors in terms of  estimation error, UQ, and computational efficiency.
The competitors for each category are as follows:
\begin{enumerate}
    \item[(a)] For estimation error, we compare with: 
    (i) the multiplicative gamma shrinkage prior approach of \cite{bhattacharya2011sparse} denoted by MGSP, 
    (ii) the automatic rotation to sparsity approach of \cite{rovckova2016fast} denoted by ROTATE, 
    (iii) the hard thresholding approach of \cite{bickel2008covariance} denoted by HT, 
    (iv) the SCAD penalty \citep{fan2001variable} applied to entries of the sample covariance and then following a compromise between hard and soft thresholding as in  \cite{rothman2009generalized}, and 
    (v) the \cite{ledoit2004well} linear shrinkage estimator denoted by LW.
    \item[(b)] For UQ, we consider the Bayesian approaches FABLE and MGSP, as the rest of the approaches only provide point estimates.
    \item[(c)] 
    Runtime analysis includes simulations comparing empirical runtimes for the different approaches, with a theoretical computational complexity analysis for FABLE available in Section \ref{suppsec:runtimeTheory} of the Supplementary Material.
\end{enumerate}


Next, we describe the true data generating models for each of the numerical experiments described above.
Suppose $\Psi_0 = \Lambda_0 \Lambda_0^\top + \Sigma_0$ denotes the true covariance matrix with $\Lambda_0 = (\Lambda_{0, jl})_{1 \leq j \leq p, 1 \leq l \leq k}$ the true loadings matrix and $\Sigma_0 = \mbox{diag}(\sigma_{01}^2, \ldots, \sigma_{0p}^2)$ the true idiosyncratic error variance matrix. 
We consider two different schemes of generating $(\Lambda_0, \Sigma_0)$:
\begin{enumerate}
    \item[(a)] \textbf{Spike-and-slab:} We generate the $(j,l)$th entry of factor loadings matrix $\Lambda_0$ as $\Lambda_{0,jl}  \overset{\ind}{\sim} \pi_0 \, \widetilde{\delta}_0 + (1 - \pi_0) \, N(0, 0.5^2)$ for all $1 \leq j \leq p, 1 \leq l \leq k$, and $\sigma_{0j}^2 \overset{\ind}{\sim} \mathcal{U}(0.5, 5),$ for all $1 \leq j \leq p,$
with $\widetilde{\delta}_0$ denoting a point mass at $0$. 
Thus, we simulate exact sparsity in the loadings, representing a challenging misspecified case for FABLE and MGSP, which use continuous shrinkage priors. 
We consider the following settings of $(n, p, k, \pi_0)$:
\begin{itemize}
    \item[(i)] $\pi_0 = 0.5$, $(n,p) \in \{500, 1000\} \times \{1000, 5000\}$, and $k = 10$.
    \item[(ii)] $\pi_0 = 0.85$ and $(n,p,k)$ varied as in (i).
    \item[(iii)] $n = 100$, $k = 10$, and $(\pi_0, p) \in \{0.5, 0.85\} \times \{100, 500\}.$
    \item[(iv)] $k = 50$ and $(n, p, \pi_0)$ varied as in (i).
\end{itemize}
Settings (i) and (ii) assess the performance of the methods when both $n$ and $p$ are large, with moderate ($\pi_0 = 0.5$) and high ($\pi_0 = 0.85$) sparsity in the loadings. 
Settings (iii) and (iv) modify (i) by considering smaller $(n,p)$ and larger $k$, respectively.

    \item[(b)] \textbf{Block-diagonal:} 
This example is taken from Section 4 of \cite{rovckova2016fast}, with $\Lambda_0$ having a block-diagonal structure with $85\%$ sparsity.
For a given $p$, the true loading matrix has $n_1(p) = \lfloor (1 - 0.85) \, p \rfloor = \lfloor 0.15 \, p \rfloor$ non-zero elements in each column and an overlap of $n_2(p) = \lceil 0.37 \, n_1(p) \rceil$ non-zero entries with the successive column, with all non-zero entries of $\Lambda_0$ set to $1$.
As earlier, we generate $\sigma_{0j}^2 \overset{\ind}{\sim} \mathcal{U}(0.5, 5)$ for $1 \leq j \leq p.$
\end{enumerate}

For each combination $(n,p)$, we use the generated $\Psi_0 = \Lambda_0 \Lambda_0^{\top} + \Sigma_0$ and replicate the data generating process $R$ times. 
For the estimation error experiments, we take $R = 50$ and consider both the spike-and-slab and block-diagonal setups. 
For the sake of exposition, we only present the results for setting (i) here for the spike-and-slab setup (a), with results for the rest of the settings (ii)-(iv) deferred to Section \ref{suppsec:sim} of the Supplementary Material.
For the estimation error experiments with the block-diagonal setup (b), we let $k = 10$ and vary $(n,p) \in \{500, 1000\} \times \{1000, 5000\}$.
For the UQ experiments, we consider the spike-and-slab setup (a) with setting (i) and consider $R = 100$ replicates.


We assess the difficulty of each simulation setup with the average proportion of variance explained by the signal expressed in percentage, computed as $$P_{av} = \dfrac{1}{p}\sum_{j=1}^{p}\left\{ \dfrac{E^*\|\lambda_{0j}\|_2^2}{E^*\|\lambda_{0j}\|_2^2 + E^*(\sigma_{0j}^2)} \right\},$$ 
with lower $P_{av}$ indicating more difficult scenarios.
Here, $\lambda_{0j}^\top$ is the $j$th row of $\Lambda_0$, and $E^*$ denotes the expectation with respect to the true generating measure of $(\Lambda_0, \Sigma_0)$.
For instance, the spike-and-slab setup (a) with setting (i) yields $P_{av} \approx 31\%$.
For the block-diagonal case (b), we have $P_{av} \approx 35\%$ for both $p = 1000$ and $p = 5000$.
The derivations of $P_{av}$ for different setups are provided in Section \ref{suppsec:sim} of the Supplementary Material.

Next, we discuss the performance metrics for judging the estimation error and UQ experiments. 
Given an estimator $\widehat{\Psi}_0$ of $\Psi_0$ obtained from a particular implementation, we assess its efficacy with the relative spectral error, defined as
$$\mathcal{L}(\Psi_0, \widehat{\Psi}_0) = \dfrac{\|\Psi_0 - \widehat{\Psi}_0\|}{\|\Psi_0\|}.$$
When using a Bayesian approach, we let $\widehat{\Psi}_0$ be the corresponding posterior mean. 
For a given $(n,p)$ and a particular replicate $r = 1, \ldots, R$ yielding the estimate $\wh{\Psi}_0^{(r)}$, we obtain $\mathcal{L}^{(r)}(\Psi_0, \widehat{\Psi}_0^{(r)})$ and proceed to report the average, $2.5\%$ quantile, and $97.5\%$ quantile of this quantity over the $R = 50$ replicates. 
For the UQ experiments, we ease computational burden by considering the coverage of a randomly chosen $100 \times 100$ submatrix of $\Psi_0$, corresponding to the covariance of $100$ randomly chosen variables.
For each $(n,p)$, these variable indices are held fixed across replicates. 
We report the average, $2.5\%$ quantile, and $97.5\%$ quantile of the frequentist coverages of the corresponding $95\%$ entrywise credible intervals across the $R = 100$ replicates.
For a given entry, its $95\%$ credible interval is given by $[L, U]$, where $L$ and $U$ are the $2.5\%$ and $97.5\%$-th empirical quantiles of its posterior samples, respectively.
All quantiles reported are based on linear interpolation between the corresponding adjacent order statistics of the replicate values, using the \texttt{quantile} function in \texttt{R}.

All the methods are implemented in the \texttt{R} programming language \citep{rcoreteam2021}. 
Open-source code for implementing FABLE is available at \url{https://github.com/shounakch/FABLE}.
We implement MGSP with the \href{https://github.com/cran/infinitefactor}{\texttt{infinitefactor}} package \citep{poworoznek2021efficiently}.
Code used to implement the ROTATE approach was obtained from \url{http://veronikarock.com/FACTOR\textunderscore ANALYSIS.zip}.
The methods HT, SCAD, and LW, were implemented with the \texttt{thresholdingEst}, \texttt{scadEst}, and \texttt{linearShrinkLWEst} functions, respectively, from the \texttt{cvCovEst} package \citep{boileau2022cvcovest}.
For FABLE, we collect $1000$ Monte Carlo samples; for estimation error experiments, we use the FABLE-posterior mean, which is explicitly available without sampling. 
For MGSP, we obtain $3000$ MCMC iterates, discard the first $1000$ as burn-in, and carry out inference based on the remaining $2000$ samples.

Throughout all simulations, the rank $k$ of the signal when fitting FABLE, MGSP, and ROTATE is treated as an unknown parameter and is thus estimated from the data.
For the ROTATE approach, $k$ is estimated with an Indian buffet process (IBP) prior as described in \cite{rovckova2016fast}.
For the MGSP approach, $k$ is estimated within the MCMC algorithm as described in \cite{bhattacharya2011sparse}.
For FABLE, we estimate $k$ following Section \ref{subsubsec:tunek}.

The ROTATE code requires the choice of a hyperparameter \texttt{lambda0} that crucially affects its performance. 
As recommended in Section 7 of \cite{rovckova2016fast}, we treat \texttt{lambda0} as an inverse temperature parameter, and  gradually increase it from $0.001$ to $18.001$ in increments of $2$. 
As described in their paper, the fit for a sufficiently large \texttt{lambda0} approximates the maximum \textit{a posteriori} (MAP) estimator of the covariance matrix.
For HT, SCAD, and LW, we obtain improved performance by standardizing the data before analysis and transforming the estimated covariance back to the original units in a post-hoc step. 
HT, SCAD, and LW are implemented with 5-fold cross-validation by default.
Further details regarding the tuning of HT, SCAD, and LW are available in Section \ref{suppsec:sim} of the Supplementary Material. 


\subsection{Estimation Performance}
\label{subsec:esterror}

\begin{table}[]
\centering
\caption{\normalsize Comparison of estimation error between multiple approaches, with $50\%$ sparsity in spike-and-slab factor loadings. The ``Mean'' and ``Range'' columns show the average and $2.5\%-97.5\%$ quantiles across replicates, respectively.}
\label{tab:simError}
\begin{tabular}{lcccccccc}
\hline
\multicolumn{1}{c}{$(n,p)$} & \multicolumn{2}{c}{$(500,1000)$} & \multicolumn{2}{c}{$(1000,1000)$} & \multicolumn{2}{c}{$(500,5000)$} & \multicolumn{2}{c}{$(1000,5000)$} \\
       Method & Mean & Range & Mean & Range & Mean & Range & Mean & Range \\
\hline
FABLE & 0.32 & 0.28 -- 0.38
& 0.23 & 0.20 -- 0.27 
& 0.33 & 0.30 -- 0.38 
& 0.24 & 0.21 -- 0.28 \\
MGSP     & 0.33 & 0.27 -- 0.42 
& 0.22 & 0.18 -- 0.27 
& 0.38 & 0.28 -- 0.50 
& 0.24 & 0.19 -- 0.31 \\ 
ROTATE   & 0.46 & 0.41 -- 0.51  
& 0.30 & 0.27 -- 0.35           
& 0.48 & 0.44 -- 0.56           
& 0.32 & 0.28 -- 0.35           \\
HT       & 0.32    & 0.27 -- 0.39           & 0.22    & 0.20 -- 0.25           
& 0.35    & 0.30 -- 0.40           
& 0.23    & 0.21 -- 0.27           \\
SCAD     & 0.34    & 0.29 -- 0.42           & 0.23    & 0.20 -- 0.28           
& 0.37    & 0.31 -- 0.43           
& 0.25    & 0.22 -- 0.30           \\
LW       & 0.35    & 0.31 -- 0.39           & 0.25    & 0.21 -- 0.29           
& 0.37    & 0.34 -- 0.42           
& 0.26    & 0.23 -- 0.30           \\
\hline
\end{tabular}
\end{table}


\begin{table}[]
\centering
\caption{\normalsize Comparison of estimation error between multiple approaches, with $85\%$ sparsity in block-diagonal factor loadings. The ``Mean'' and ``Range'' columns show the average and $2.5\%-97.5\%$ quantiles across replicates, respectively.}
\label{tab:simErrorSparsity}
\begin{tabular}{lcccccccc}
\hline
\multicolumn{1}{c}{$(n,p)$} & \multicolumn{2}{c}{$(500,1000)$} & \multicolumn{2}{c}{$(1000,1000)$} & \multicolumn{2}{c}{$(500,5000)$} & \multicolumn{2}{c}{$(1000,5000)$} \\
       Method & Mean & Range & Mean & Range & Mean & Range & Mean & Range \\
\hline
FABLE & 0.24 & 0.21 -- 0.30 
& 0.17 & 0.14 -- 0.22 
& 0.24 & 0.21 -- 0.31 
& 0.17 & 0.14 -- 0.22 \\
ROTATE   & 0.28 & 0.20 -- 0.38           & 0.12 & 0.08 -- 0.18           & 0.23 & 0.11 -- 0.35           & 0.09 & 0.07 -- 0.13           \\
HT       & 0.17    & 0.13 -- 0.25           & 0.09    & 0.07 -- 0.13           & 0.17    & 0.13 -- 0.24           & 0.09    & 0.07 -- 0.13            \\
SCAD     & 0.17    & 0.14 -- 0.24           & 0.13    & 0.10 -- 0.17           & 0.17    & 0.13 -- 0.24           & 0.13    & 0.09 -- 0.17           \\
LW       & 0.26    & 0.19 -- 0.34           & 0.18    & 0.15 -- 0.21           & 0.26    & 0.20 -- 0.34           & 0.18    & 0.15 -- 0.22           \\
\hline
\end{tabular}
\end{table}
Results of simulation experiments for the spike-and-slab setup with setting (i) are available in Table \ref{tab:simError}.
FABLE performs competitively for all four choices of $(n, p)$ compared to approaches that induce exact sparsity, such as ROTATE, HT, and SCAD. 
The performance of MGSP declines with increasing dimension more than the competitors, possibly due to mismatch between the continuous shrinkage prior  and the exact sparsity in the simulated loadings. 
The MGSP MCMC algorithm exhibited efficient mixing, with the average effective sample size (ESS) of the $p$ idiosyncratic error variances $\sim 95\%$ of the post-burn-in samples.
We also manually tuned the \texttt{lambda0} parameter when implementing ROTATE and obtained estimation errors almost identical to that of FABLE for \texttt{lambda0 = 2.001}. 
However, manual tuning requires knowledge of the true covariance matrix and thus is not feasible in practice. 

For the block-diagonal setup, the results are in Table \ref{tab:simErrorSparsity}, excluding MGSP which is not competitive with the other approaches. 
ROTATE, HT, and SCAD outperform FABLE for most of the choices of $(n,p)$. 
However, FABLE remains competitive in terms of estimation error for both choices of $p$ when $n = 500$. 
It is not surprising that ROTATE and the other approaches enforcing exact sparsity perform particularly well in the highly sparse regime; we expect FABLE and MGSP to outperform competitors when loadings are simulated under continuous light-tailed priors. 
However, it is notable that FABLE badly outperforms MGSP in this highly sparse case; it may be that pre-estimating the factors conveys greater robustness to a variety of true loadings structures.

\subsection{Frequentist Coverage}
\label{subsec:covsim}


\begin{table}[]
\centering 
\caption{\normalsize Comparison of frequentist coverages and interval widths across entrywise credible intervals obtained from FABLE and MGSP, with $50\%$ sparsity in spike-and-slab factor loadings. 
The ``Mean'' and ``Range'' columns show the mean and $2.5 \% - 97.5 \%$ quantiles of average coverages (over all entries) across replicates, respectively.
} 
\label{tab:coverage}

\begin{tabular}{c cc cc cc cc}
\hline
Method & \multicolumn{4}{c}{FABLE} & \multicolumn{4}{c}{MGSP} \\
Metric & \multicolumn{2}{c}{Coverage} & \multicolumn{2}{c}{Width}
 & \multicolumn{2}{c}{Coverage} & \multicolumn{2}{c}{Width} \\
$(n,p)$ & Mean & Range & Mean & Range & Mean & Range & Mean & Range \\
\hline
$(500, 1000)$ 
& 0.95 & 0.94 -- 0.97 
& 0.46 & 0.45 -- 0.47
& 0.85 & 0.83 -- 0.86
& 0.45 & 0.44 -- 0.46 \\
$(1000, 1000)$ 
& 0.95 & 0.94 -- 0.96
& 0.32 & 0.32 -- 0.33
& 0.80 & 0.79 -- 0.81
& 0.32 & 0.31 -- 0.33 \\
$(500, 5000)$ 
& 0.96 & 0.94 -- 0.97
& 0.48 & 0.46 -- 0.49
& 0.83 & 0.81 -- 0.84
& 0.44 & 0.43 -- 0.45 \\
$(1000, 5000)$ 
& 0.95 & 0.94 -- 0.97
& 0.34 & 0.33 -- 0.34
& 0.77 & 0.76 -- 0.78
& 0.31 & 0.31 -- 0.32 \\
\hline
\end{tabular}
\end{table}

As described earlier, for each $(n,p)$, we first consider $100$ randomly chosen variables and hold them fixed across $R=100$ replicates of the data. 
Next, we obtain the average coverage and interval width of $95\%$ posterior credible intervals corresponding to the relevant $100 \times 100$ submatrix of $\Psi_0$. 
For all the simulation instances, we set the coverage-correction factor $\rho$ as recommended in Section \ref{subsec:cov_correction}. 
The results, provided in Table \ref{tab:coverage}, suggest that FABLE provides superior entrywise coverage when compared to MGSP, with comparable interval width. 
In all the cases, the average coverage of entrywise intervals obtained from MGSP falls short of the nominal value $0.95$, while FABLE meets the nominal coverage. 
Furthermore, the average entrywise coverage obtained from MGSP decreases when the number of dimensions increases for a fixed sample size, while the results of FABLE are not affected.

\subsection{Computational Efficiency}
\label{subsec:runtime}

In Figure \ref{fig:runtimePlots}, we compare FABLE with the other competitors in terms of obtaining posterior samples and point estimates of the covariance matrix. 
We consider the spike-and-slab setup as in Section \ref{subsec:preliminaries} with $\pi_0 = 0.5$.
We fix $n = 500$ and vary $p$ from $500$ to $4000$ in increments of $500$, obtaining runtime measurements across $R=20$ independent replicates for each choice of $p$. 
As before, we collect $1000$ Monte Carlo samples for FABLE. 
For MGSP, we collect $3000$ MCMC iterations and discard the first $1000$ iterates as burn-in.
The experiments were carried out on an M1 MacBook Pro with 32 GB of RAM (random-access memory). 
We compare FABLE with MGSP when obtaining samples of the factor loadings and the error variances, and with the other approaches ROTATE, HT, SCAD, and LW, when obtaining a point estimate of the covariance matrix.
In both cases, FABLE is faster than its competitors. 
When obtaining samples, FABLE is about $150-200$ times faster than MGSP, with this ratio remaining fairly constant as $p$ increases. 
When obtaining a point estimator of the covariance, FABLE is at least four times as fast as ROTATE.
Furthermore, obtaining posterior samples with FABLE is around $2.5$ times slower than obtaining the FABLE-point estimate, when averaged across replicates and dimensions.
We provide an analysis of the theoretical computational complexity of FABLE in Section \ref{suppsec:runtimeTheory} of the Supplementary Material.

\begin{figure}[]
\centering
\includegraphics[scale=0.35]{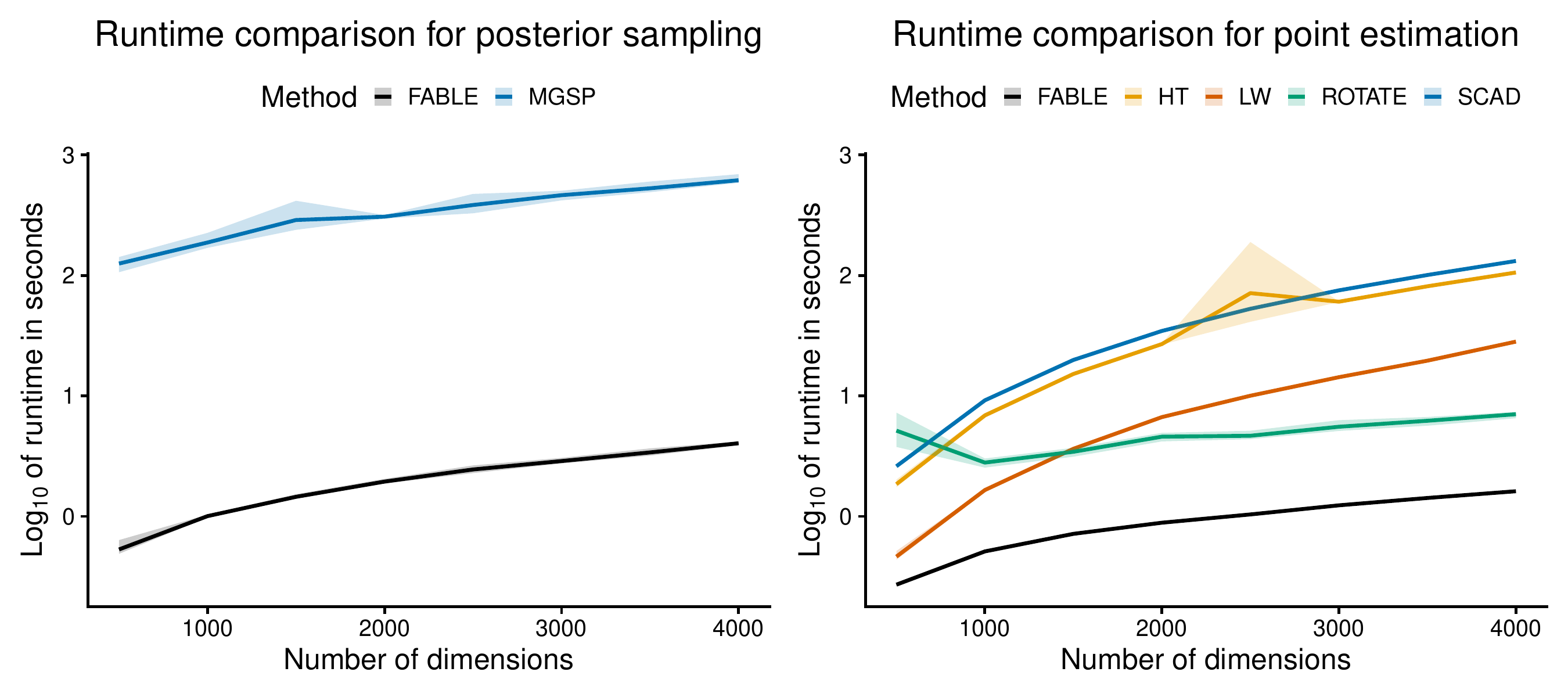}
\caption{Plots comparing the runtimes in seconds (in $\log_{10}$ scale) of FABLE with the competitors. 
The left panel compares FABLE with MGSP when obtaining posterior samples. 
The right panel compares FABLE, ROTATE, HT, SCAD, and LW, when obtaining a point estimate. 
Solid lines indicate average runtimes across replicates while shaded areas indicate minimum and maximum runtimes over replicates.}
\label{fig:runtimePlots}
\end{figure}

\section{Application}
\label{sec:application}







\subsection{Setup}
\label{subsec:appsetup}

We used FABLE to identify associations between gene expressions in a high-throughput sequencing immunocellular dataset. Identifying gene expression associations in varying immune cell populations and between multiple lineages is critical to understanding the inflammatory response to infectious and chronic diseases \citep{ota2021dynamic, saini2022gene}. 
We consider the GSE109125 dataset, which contains RNAseq data from 127 highly purified immune cell lineages, including adaptive and innate lymphocytes, myeloid cells, mast cells, and neutrophils. These data were profiled using the ImmGen ULI pipeline \citep{yoshida2019cis}. 

We first log-transform raw data counts $R_{ij}$ for the $i$th cell and $j$th gene as $y^{*}_{ij} = \log_2(R_{ij} + 1)$ and center the data; 
there are no cells with missing entries.
We then use the \texttt{genefilter} package in \texttt{R} \citep{gentleman2015genefilter, bioconductor2015} to filter relevant genes before fitting the latent factor model. 
In this case, we filter genes according to their variance and consider the top $10\%$ of genes with the highest variances.  
After carrying out appropriate pre-processing of the dataset, we obtain expression data $\mathbf{Y} \in \mathbb{R}^{n \times p}$ for $p = 5300$ genes measured on $n = 205$ cells. 

We compare results from applying the three methods FABLE, MGSP, and ROTATE on the data. 
All three approaches fit the Bayesian latent factor model as described in Section \ref{sec:method}.
The estimated rank when using the $\mbox{JIC}(k)$ criterion as highlighted in Section \ref{subsubsec:tunek} is $\widehat{k} = 30$.
For the other approaches MGSP and ROTATE, $k$ is estimated as  described in Section \ref{subsec:preliminaries}.
Fitting the ROTATE approach by gradually increasing \texttt{lambda0} as described in Section \ref{subsec:preliminaries} did not perform well for this application.
Instead, we obtain superior performance by starting with \texttt{lambda0 = 1} and gradually increasing it to \texttt{lambda0 = 5} and \texttt{lambda0 = 10}, and provide the results for all three choices.

To investigate the overall fit of the Bayesian factor model to the observed data, we carried out several posterior predictive checks.
We first obtain the average coverage of entrywise $95\%$ predictive intervals across all the entries of $\mathbf{Y}$.
The approaches FABLE, MGSP, and ROTATE (with all three choices of \texttt{lambda0}) each have average predictive coverages between $94\% - 96\%$, indicating appropriate calibration. 
Next, we also looked into the average proportion of variance explained by the factor model across all the genes, based on each implementation, which remained between $83\% - 85\%$ across the different approaches.
Both of these checks indicate that factor modeling is a good assumption for this dataset.
Lastly, we found the Gaussian assumption on the entries of $\mathbf{Y}$ to be well supported.
However, a natural direction for future research is to extend the FABLE methodology to directly model count data, such as the $R_{ij}$s themselves.


\subsection{Illustrating Blessing of Dimensions}
\label{subsec:appblessing}

We first highlight the blessing of dimensionality when implementing FABLE through the task of covariance submatrix estimation.
We denote the full covariance matrix by $\Psi \in \mathbb{R}^{p \times p}$ where $p = 5300$ and the covariance of any subset $\mathcal{I} \subset \{1, 2, \ldots, p\}$ of the variables by $\Psi_{\mathcal{I}} \in \mathbb{R}^{|\mathcal{I}| \times |\mathcal{I}|}$. 
That is, $\Psi_{\mathcal{I}}$ is  the submatrix of $\Psi$ corresponding to the indices in $\mathcal{I}$.
Without loss of generality, we assume that the variable indices are in descending order of their variance; that is, variable $1$ has the largest variance and variable $p$ has the smallest variance.
Our objective is to estimate the $100 \times 100$ covariance submatrix corresponding to the $100$ genes with the highest variability. 
Suppose the indices of these genes are $\mathcal{I}_0 = \{1,2,\ldots, 100\}$.
To estimate $\Psi_{\mathcal{I}_0}$ with  FABLE, ROTATE, or MGSP, we could adopt two schemes, described as follows: 

\begin{enumerate}
    \item[(i)] Consider gene expression data for variables with indices $\mathcal{I}_0$ and estimate $\Psi_{\mathcal{I}_0}$ only based on this data.
    \item[(ii)] For a given $p_S \geq 1$, let $\mathcal{A}(p_S) = \{|\mathcal{I}_0| + 1, \ldots, |\mathcal{I}_0| + p_S\} = \{100 + 1, \ldots, 100 + p_S\}$ denote the set of indices of genes with the next $p_S$ highest variances. 
    We estimate $\Psi_{\mathcal{I}_0 \cup \mathcal{A}(p_S)}$ based on data for the variables with indices $\mathcal{I}_0 \cup \mathcal{A}(p_S) = \{1, \ldots, 100 + p_S\}$, and extract the submatrix $\Psi_{\mathcal{I}_0}$ from $\Psi_{\mathcal{I}_0 \cup \mathcal{A}(p_S)}$.
\end{enumerate}



\noindent Scheme (i) is simply a special case of scheme (ii) with $p_S = 0$ and $\mathcal{A}(0) = \phi$.
We vary $p_S \in \{0\} \, \cup \, \{100j \, : \, j=1,\ldots,10\} \, \cup \, \{2000, 4000\}$ and extract the relevant estimate of the covariance between the genes with indices in $\mathcal{I}_0$, namely $\wh \Psi_{\mathcal{I}_0}(p_S)$, for each $p_S$. 

To obtain $\wh \Psi_{\mathcal{I}_0}(p_S)$ for a particular approach, we perform a random train-test split on $\mathbf{Y}$ with $50$ test samples and $n - 50 = 155$ training samples, and train our model on the $155 \times (100 + p_S)$ submatrix of $\mathbf{Y}$ as described earlier.
For each train-test split, we sort the variables according to their variance  with respect to the full $205 \times 5300$ data set.
We evaluate the performance of the estimate $\wh \Psi_{\mathcal{I}_0}(p_S)$ with the out-of-sample log-likelihood (OOSLL) of the $50 \times 100$ test data corresponding to the indices $\mathcal{I}_0$.
We carry out this analysis for $R = 10$ independent replicates of the train-test split, and compare the methods FABLE, MGSP, and ROTATE on the basis of average OOSLL across replicates.
The results are illustrated in Figure \ref{fig:OOSLLApplication}, with the OOSLL values available in Section \ref{suppsec:app} of the Supplementary Material.

\begin{figure}[]
    \centering
\includegraphics[scale=0.45]{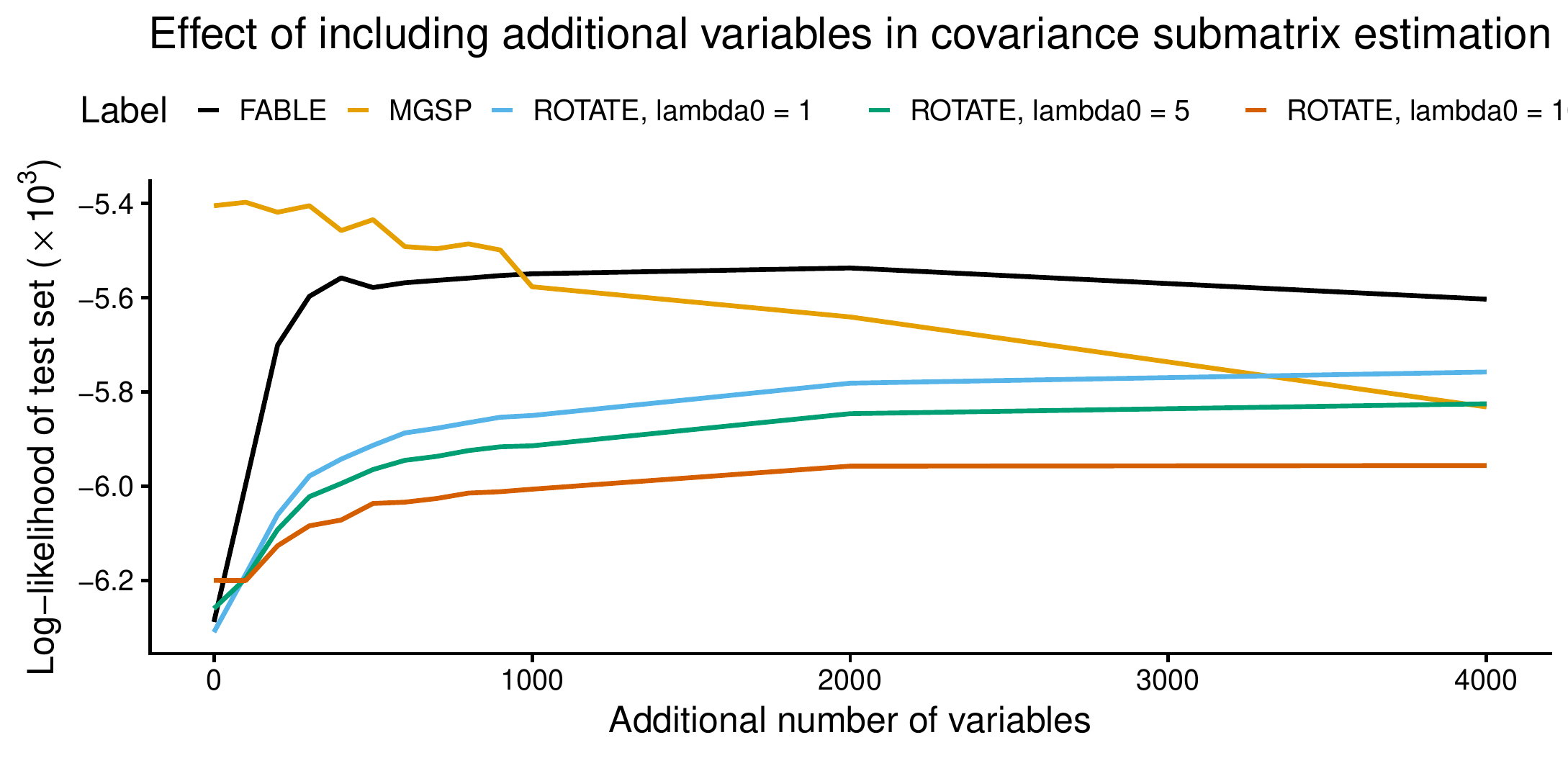}
    \caption{Figure showing log-likelihood of the test set (y-axis) averaged over $R = 10$ replicates of the train-test split when using additional genes (x-axis) to estimate the covariance of the relevant set of $100$ variables, with FABLE, MGSP, and ROTATE. The training set has $155$ cells and the test set has $50$ cells for each of the $R = 10$ train-test split replicates.
    Higher values of test log-likelihood are better.}
    \label{fig:OOSLLApplication}
\end{figure}


From Figure \ref{fig:OOSLLApplication}, it is clear that both FABLE and MGSP exhibit superior performance to ROTATE. 
MGSP performs the best for smaller values of $p_S$, 
along with exhibiting an essentially monotonically decreasing trend as extra genes are added, so that including additional genes does not improve inferences on the genes of interest by using these methods.
In sharp contrast, FABLE shows a rapid initial improvement in performance as extra genes are added, exhibiting a clear blessing of dimensionality. 
After a few hundred genes, the gain levels off, and eventually there is a modest decline, perhaps due to the need to add additional factors not related to the genes of interest when large numbers of additional genes are added.
Along with FABLE, all three variants of ROTATE also seem to benefit from considering additional variables when training the model, thus highlighting a ``blessing of dimensionality'' phenomenon for these approaches.

Thus, the proposed method shows promising results and compares favourably with two state-of-the-art approaches to Bayesian factor analysis.
The results obtained from FABLE come at a fraction of the computational budget compared to MGSP, due to its embarrassingly parallel sampling scheme requiring no MCMC. 
When fitted on the full dataset on an M1 MacBook Pro with $32$ GB of RAM, FABLE only took $1.1$ seconds to compute the posterior mean 
while running MGSP with $3000$ MCMC and $1000$ burn-in iterates took $\sim 27$ minutes. This translates to a speed-up of close to $1600$ times for FABLE. 

\subsection{Train-Test Split}
\label{subsec:apptraintest}

We also carry out a more straightforward train-test split exercise for the gene expression data application without subsetting the number of variables, 
thus working with all $p = 5300$ variables.
We first hold out 35 samples or equivalently a $35 \times 5300$ submatrix as the test set for OOSLL evaluation, and consider a sequence of $4$ training sample sizes $n_T \in \{110, 130, 150, 170\}$ among the rest of the samples.
For each such $n_T$, we fit FABLE, MGSP, and the three variants of ROTATE with \texttt{lambda0} in $\{1, 5, 10\}$ on the training set, obtain corresponding covariance matrix estimates, 
and then obtain the OOSLL of the $35$ held-out test samples with this estimate.
We carry out this analysis for $R = 10$ independent replicates of the train-test split and report average OOSLL for each $n_T$.
The results are available in Figure \ref{fig:trainTestApplication}, with the OOSLL values available in Section \ref{suppsec:app} of the Supplementary Material.
Overall, MGSP performed the best, followed by FABLE, with FABLE nearly outperforming all variants of the ROTATE approach with different values of the tuning parameter \texttt{lambda0}.
All approaches show a naturally increasing trend for the OOSLL as the training sample size increases.

\begin{figure}[]
    \centering
\includegraphics[scale=0.35]{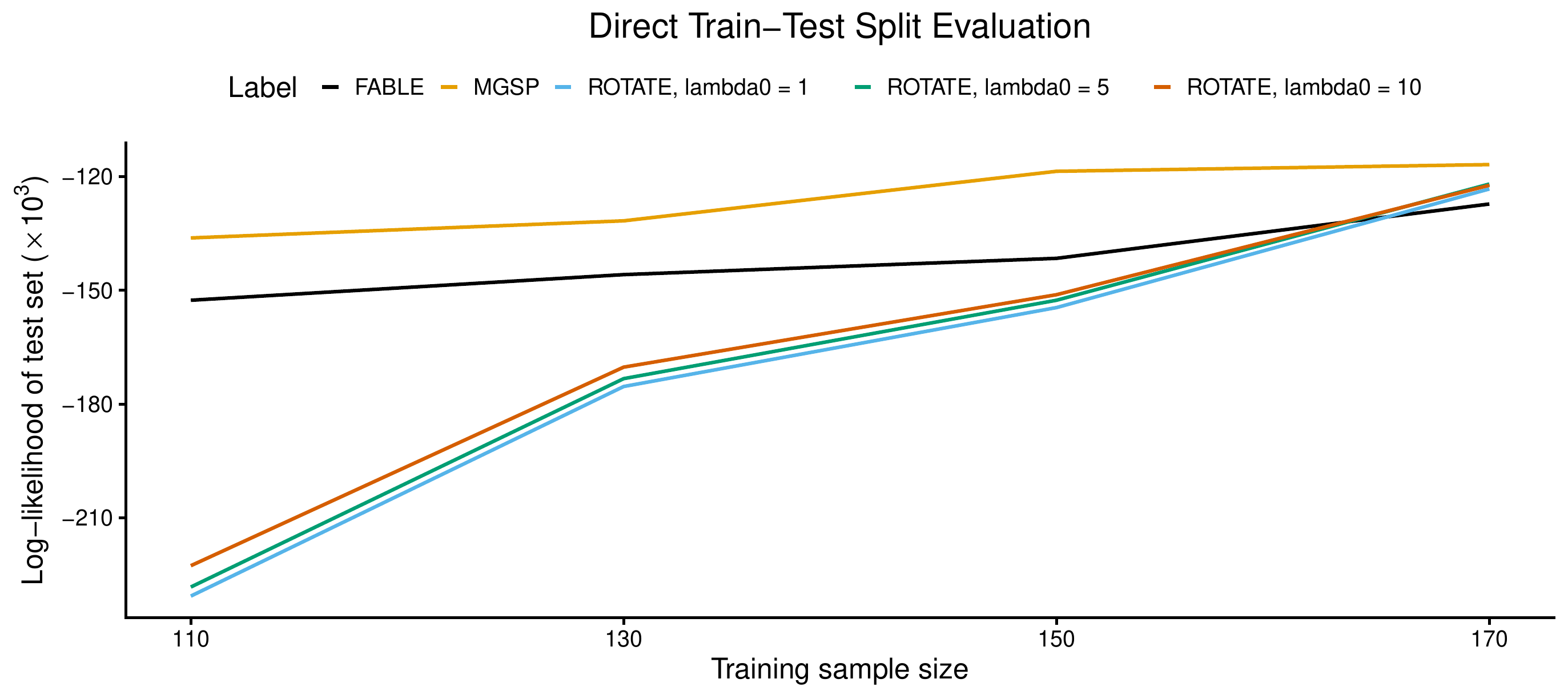}
    \caption{Out-of-sample log-likelihood (y-axis) averaged over $R = 10$ replicates for the direct train-test split exercise with gene expression data, using the FABLE, MGSP, and ROTATE approaches. 
    The training sample size is varied from $110$ to $170$ in increments of $20$. 
    Higher values of test log-likelihood are better.
    }
    \label{fig:trainTestApplication}
\end{figure}

\section{Discussion}
\label{sec:discussion}

In this paper, we develop a computationally scalable approach to fit high-dimensional Bayesian latent factor models with Gaussian data. Using a blessing of dimensions, the proposed approach bypasses the need for MCMC to provide FABLE-posterior samples of the covariance matrix for inference. Due to its embarrassingly parallel nature and reliance on independent samples, FABLE has immense computational benefits over current approaches dependent on MCMC.
The FABLE posterior enjoys desirable theoretical properties, such as consistency and asymptotically accurate UQ of credible intervals. 

This work has already led to several important extensions. \cite{mauri2025factor} extend the FABLE methodology to propose a scalable approach for fitting Bayesian multivariate logistic factor models. Their work is motivated by applications
to joint species distribution modeling in ecology. The number of dimensions (species) for these problems is on the order of $ p \sim 10^4-10^5$, rendering traditional MCMC-based approaches relying on data augmentation Gibbs samplers infeasible \citep{albert1993bayesian, polson2013bayesian}. 
Other important extensions of FABLE include inference on covariance structure in high-dimensional multi-view data \citep{mauri2025inference} and multi-study factor analysis \citep{mauri2025spectral}, with compelling applications in high-throughput multi-omics analyses.

Additional important modeling directions are to modify the FABLE inferential framework to allow more intricate hierarchical modeling. 
In Bayesian factor analysis, it is often of interest to (1) include covariate effects, (2) use more elaborate priors on the loadings to favor sparsity and adaptive selection of the number of factors, (3) consider more flexible latent factor distributions, and/or (4) allow nonlinear measurement structures. 
In exploratory analyses, we have observed good performance for ad hoc FABLE modifications by including covariates and for simulated data having non-Gaussian latent factor distributions. 
Problem (2) can potentially be addressed in the second stage via MCMC algorithms implemented in parallel for the different sparse regressions after inferring the latent factors using a sparse SVD in the first stage. 
Problem (4) may necessitate non-linear dimensionality reduction in the first stage. 
Obtaining concrete implementations with theoretical support for problems (1)-(4) is of future interest.

The FABLE approach crucially relies on the number of dimensions increasing to infinity as the number of samples grows. 
This is essential to ensure the accurate pre-estimation of the unobservable latent factors before leveraging this estimate to obtain exact Monte Carlo samples in an embarrassingly parallel manner.
For smaller numbers of samples and dimensions where such an assumption may not hold, we conjecture that instead of a point pre-estimate of the unknown latent factors, they can be sampled from a distribution with probabilistic uncertainty around this pre-estimate, up to possible rotational constraints.
As an extension of FABLE, we can then sample the latent factors, factor loadings, and error variances independently of one another from approximate posteriors while maintaining an embarrassingly parallel framework.
Just like the proposed FABLE procedure, this modification would also entirely bypass MCMC to obtain posterior samples for the covariance matrix.
We are actively working on this extension and have already obtained some promising results. 

The current work also hints at further investigation into the theoretical properties of Bayesian factor models. 
Although there are existing results focusing on posterior consistency with carefully chosen shrinkage priors on the factor loadings, these are primarily in the case where the idiosyncratic model errors are Gaussian. 
It would thus be interesting to explore how the blessing of dimensions influences both posterior consistency and UQ in different settings, such as when the data are non-Gaussian. 
Furthermore, there remains a clear need to extend existing theoretical results on Bayesian factor models to encompass UQ of the covariance matrix.
Two other future theoretical directions are of immediate interest. 
Firstly, it would be interesting to relax the condition $\sqrt{n}/p_n = o(1)$ when obtaining entrywise guarantees on UQ of the FABLE-posterior.
Secondly, we would like to extend our entrywise UQ result to guarantees that hold uniformly over the entries, thus ensuring joint frequentist coverage of entrywise credible intervals.
\section*{Acknowledgments}

This work was partially supported by Merck, the United States Office of Naval Research (ONR) under grant N00014-21-1-2510, the European Research Council (ERC) under grant 856506, and the National Institutes of Health (NIH) under grant R01ES035625. 
Code used to implement FABLE in the simulation experiments and the gene expression data application is available to download at \url{https://github.com/shounakch/FABLE}.
The authors would also like to thank Sunrit Chakraborty and Peter Dunson for helpful comments and suggestions.

\section*{Data Availability}

The gene expression dataset used in Section \ref{sec:application} of this paper is publicly available for download at \url{https://www.ncbi.nlm.nih.gov/geo/query/acc.cgi?acc=GSE109125} with filename \texttt{GSE109125\textunderscore Genes\textunderscore count\textunderscore table.tsv.gz}.


\bibliography{reference}
\bibliographystyle{apalike}

\newpage

\appendix

\section*{Supplementary Material}

Details on the content of this supplement are as follows.
\begin{enumerate}
    \item Section \ref{app:propIdentify} contains the proof of Proposition \ref{prop:mhat} in the main paper.
    \item Section \ref{proof:post-con} contains the proof of Theorem \ref{theorem:post-con} in the main paper.
    \item Section \ref{proof:bvm} contains the proof of Theorem \ref{theorem:bvm} in the main paper.
    \item Section \ref{proof:asymp-law} contains the proof of Theorem \ref{theorem:asymp-law} in the main paper.
    \item Section \ref{proof:lemmaPostCon} contains the proof of Proposition \ref{proposition:bod} and relevant lemmas used to prove Theorem \ref{theorem:post-con} in the main paper.
    \item Section \ref{proof:lemmaTheorem2and3} contains the proof of additional relevant lemmas used to prove Theorems \ref{theorem:bvm} and \ref{theorem:asymp-law} in the main paper.
    \item Section \ref{suppsec:invgamma} provides results for a sensitivity analysis of FABLE, assessing the impact of varying the inverse-gamma prior hyperparameters.
    \item Section \ref{suppsec:sim} contains additional simulation results for cases described in the main paper, tuning details for the methods HT, SCAD, and LW as implemented in \texttt{R}, as well as the derivation of proportion of variance explained in the simulation examples.
    \item Section \ref{suppsec:app} contains information regarding the gene data application described in the main paper.
    \item Section \ref{suppsec:algorithm} contains a clear description of the FABLE algorithm, along with a computationally convenient expression for the estimated frequentist coverage function $\wh{q}_{uv}(\rho)$ entirely in terms of $(b_{uv})_{1 \leq u, v \leq p}$, defined in \eqref{eq:b_uv}.
    \item Section \ref{suppsec:JICvsAICBIC} contains a comparison of the JIC with the AIC and BIC.
    \item Section \ref{suppsec:runtimeTheory} contains a theoretical analysis of the computational complexity for the FABLE algorithm.
\end{enumerate}
Equations throughout the Supplementary Material are numbered as (S1), (S2), etc.





\section{Proof of Proposition \ref{prop:mhat}}
\label{app:propIdentify}

\begin{proof}
    \begin{enumerate}
        \item[(i)]
        
        We first have $\widehat{\mathbf{M}}^{\top} \widehat{\mathbf{M}} = \widehat{C}^{-1} \mathbf{A}^{\top} \mathbf{A} (\widehat{C}^{\top})^{-1} = \widehat{C}^{-1} (n \widehat{C} \widehat{C}^{\top}) (\widehat{C}^{\top})^{-1} = n \mathbb{I}_k$. 
        Next, we have $\widehat{\mathbf{M}} \widehat{\mathbf{M}}^{\top} = \mathbf{A} (\widehat{C} \widehat{C}^{\top})^{-1} \mathbf{A}^{\top} = (UD/\sqrt{p}) (np D^{-2}) (DU^{\top} / \sqrt{p}) = nUU^{\top}$.
        \item[(ii)] 
        
        It is clear that the FABLE-posterior of $\widetilde{\sigma}_j^2$ only depends on $\widehat{\mathbf{M}}$ through $\widehat{\mathbf{M}}^{\top} \widehat{\mathbf{M}}$ and $\widehat{\mathbf{M}}\widehat{\mathbf{M}}^{\top}$ for all $1 \leq j \leq p$. 
        For $1 \leq u \leq p,$ a FABLE-posterior draw of $\wt{\lambda}_u$ may be represented as
        $$\widetilde{\lambda}_u = \mu_u + \rho \, \widetilde{\sigma}_u \mathbf{K}^{1/2} Z_u,$$
        where $Z_u \overset{ind}{\sim} N_k(0_k,\mathbb{I}_k)$ for $1 \leq u \leq p$ and $Z_u$ is independent of $\widetilde{\sigma}_u^2$. 
        This implies that for $1 \leq u, v \leq p$, $$\widetilde{\lambda}_u^{\top} \widetilde{\lambda}_v = \mu_u^{\top} \mu_v + \rho \, \widetilde{\sigma}_v \mu_u^{\top} \mathbf{K}^{1/2} Z_v + \rho \, \widetilde{\sigma}_u \mu_v^{\top} \mathbf{K}^{1/2} Z_u + \rho^2 \, \widetilde{\sigma}_u \widetilde{\sigma}_v Z_u^{\top} \mathbf{K} Z_v.$$
        We now consider each expression one-by-one. The first term is $\mu_u^{\top} \mu_v = y^{(u) \top} \widehat{\mathbf{M}} \mathbf{K}^{2} \widehat{\mathbf{M}}^{\top} y^{(v)} = (n + \tau^{-2})^{-2} y^{(u) \top} \widehat{\mathbf{M}}\widehat{\mathbf{M}}^{\top} y^{(v)}.$ 
        When $u \neq v$, the second and third terms have  Gaussian distributions with means $0$ and variances  $ \rho^2 \, \widetilde{\sigma}_u^2 \|\mu_v\|_2^2 / (n + \tau^{-2})$ and $\rho^2 \, \widetilde{\sigma}_v^2 \|\mu_u\|_2^2 / (n + \tau^{-2})$, respectively.
        This is because $\mathbf{K} = \left\{\widehat{\mathbf{M}}^{\top} \widehat{\mathbf{M}} + (\mathbb{I}_k / \tau^2)\right\}^{-1} = (n + \tau^{-2})^{-1} \mathbb{I}_k.$
        When $u=v$, both terms are equal and the sum has a Gaussian distribution with mean $0$ and variance $4 \widetilde{\sigma}_u^2 \|\mu_u\|_2^2$. The final term is equal to $(n + \tau^{-2})^{-1} \widetilde{\sigma}_u \widetilde{\sigma}_v Z_u^{\top} Z_v$. 
        Since $\|\mu_u\|_2^2$ only depends on $\wh{\mathbf{M}}\wh{\mathbf{M}}^\top$, all the terms depend on $\widehat{\mathbf{M}}$ only as hypothesized, implying the result for $\widetilde{L}$. 
        The result for $\widetilde{\Psi}$ is obtained by simply observing that $\widetilde{\Psi} = \widetilde{L} + \widetilde{\Sigma}.$
    \end{enumerate}
    This proves the desired result.
\end{proof}

\section{Proof of Theorem \ref{theorem:post-con}}
\label{proof:post-con}

    \subsection{Proof Strategy}

    For part (a), we first decompose $\wt L - L_0$ (where $L_0 = \Lambda_0\Lambda_0^\top$) into a number of terms.
    Among the terms in this decomposition, one term arises from the randomness in $M_0^{\top} M_0$ and another term arises from the cross terms involving the error matrix $E$.
    These two terms provide the final rate, while other terms dealing with shrinkage bias and the residual projection term are asymptotically negligible.
    For part (b), we show $\wt\sigma_j^2$ concentrates around $\sigma_{0j}^2$ by re-centering the inverse-gamma FABLE-posterior as a shifted Gamma variable and then applying a non-asymptotic tail bound.
    Part (c) then follows from parts (a) and (b) via the triangle inequality and an equivalence bound on $\|\Psi_0\|$.
    Proofs of relevant technical lemmas are available in Section \ref{suppsubsec:lemmathm1}.

    \subsection{Proof of part (a)}
    \begin{proof}
    We first start with the FABLE-posterior contraction of $\wt L = \wt \Lambda \wt \Lambda^\top$ to $L_0 = \Lambda_0 \Lambda_0^\top.$ 
    Let $G_0 = [\mu_1, \ldots, \mu_p]^\top = \sqrt{n} \, \Y^{\top} U / (n + \tau^{-2})$ and $G = G_0 G_0^{\top} = n\Y^{\T} UU^{\T} \Y / (n + \tau^{-2})^2.$  
    Write $\wt \Lambda = G_0 + \widetilde{E}$, where $\widetilde{E} = [\tilde{e}_1, \ldots, \tilde{e}_p]^{\top}$ with $$\tilde{e}_j \overset{ind}{\sim} N_k\left(0, \frac{\rho^2 \wt\sigma_{j}^2}{n + \tau^{-2}} \mathbb{I}_k \right) \text{ for $j = 1, \ldots, p$.}$$ 
    We have
    \begin{equation}
    \label{suppeq:lambdaMainSplit}
        \begin{split}
            \wt L - L_0 & = \wt{\Lambda} \wt{\Lambda}^\top - \Lambda_0\Lambda_0^\top \\ 
            & = \left[\frac{\Y^\top \wh{ \mathbf{M}}\widetilde{E}^\top + \widetilde{E}\wh{ \mathbf{M}}^\top \Y }{n+\tau^{-2}} + \widetilde{E}\widetilde{E}^\top \right] + \left(G - \Lambda_0\Lambda_0^{\top}\right) \\
            & = \text{(I) + (II), \quad say.}
        \end{split}
    \end{equation}
    We individually bound each of (I) and (II). 
    First, we decompose (II) as follows:
    \begin{align}
    \label{suppeq:lambdaIISplit}
        \begin{split}
            & G - \Lambda_0\Lambda_0^\top = \frac{n}{(n+\tau^{-2})^2}\left(\Y^\top \Y - \Y^\top U_{\perp}U_{\perp}^\top \Y\right) - \Lambda_0\Lambda_0^\top\\
            & =  \frac{-n}{(n+\tau^{-2})^2}\Y^\top U_{\perp}U_{\perp}^\top \Y + \frac{n}{(n+\tau^{-2})^2}\left(M_0\Lambda_0^\top + E\right)^\top \left(M_0\Lambda_0^\top + E\right) - \Lambda_0\Lambda_0^\top \\
            & = \frac{-n}{(n+\tau^{-2})^2}\Y^\top U_{\perp}U_{\perp}^\top \Y + \frac{n}{(n+\tau^{-2})^2}\Lambda_0(M_0^\top M_0-n\mathbb{I}_k)\Lambda_0^\top +  \\
            & \left(\frac{n^2}{(n+\tau^{-2})^2}-1\right) \Lambda_0\Lambda_0^\top + \frac{n}{(n+\tau^{-2})^2}(\Lambda_0 M_0^{\top} E + E^\top M_0\Lambda_0^\top + E^\top E) \\
            & = \text{(II.I) + (II.II) + (II.III) + (II.IV), \quad say.}
        \end{split}
    \end{align}
    To develop an upper bound of $\| \wt L - L_0 \|$, we aim to develop an upper bound for the spectral norm of (I) and each term in \eqref{suppeq:lambdaIISplit}, namely (II.I) -- (II.IV).
    We enumerate them as follows. 
    \begin{itemize}
        \item[(i)] First, we develop a probabilistic upper bound for $\wt{\sigma}_{j}^2,$ which will be used throughout this proof.
        Recall that $\|y^{(j)}\|_2^2 = \sum_{i=1}^n y_{ij}^2 \sim (\sigma_{0j}^2 + \|\lambda_{0j}\|_2^2) \, \chi_n^2.$ 
        From Lemma \ref{lemma:laurent-massart}, we have $\|y^{(j)}\|_2^2 \lesssim n$ with probability at least $1-o(1).$
    Thus, with probability at least $1-o(1)$ we have
    $$\frac{\gamma_0\delta_0^2 + \|y^{(j)}\|_2^2 -  (n+\tau^{-2}) \|\mu_j\|_2^2}{2} \lesssim n.$$
    As $$\wt{\sigma}_{j}^2 \sim \mbox{IG}\left(\frac{\gamma_0+n}{2}, \frac{\gamma_0\delta_0^2 + \|y^{(j)}\|_2^2 - (n+\tau^{-2}) \|\mu_j\|_2^2}{2}\right),$$
    the Gamma distribution tail bound in Theorem 5 of \cite{zhang2020non} implies there exists a constant $C'>0$ such that
\begin{equation}\label{ineq:tilde-E}
        \begin{split}
            \wt \Pi\left(\max_{1\leq j\leq p}\wt{\sigma}_{j}^2 \leq C'\right) = 1- o_{P_0}(1).
        \end{split}
    \end{equation}
    \item[(ii)] Now, we start with (II.III).
    Since $\|\Lambda_0\| \asymp \sqrt{p_n}$ from Assumption \ref{assumption:norm}, we have
\begin{equation}\label{ineq:evaluation-Lambda_0}
    \left\|\left\{\frac{n^2}{(n+\tau^{-2})^2}-1\right\}\Lambda_0\Lambda_0^\top\right\| = \left|
\frac{n^2 - (n+\tau^{-2})^2}{(n+\tau^{-2})^2}\right|\|\Lambda_0\|^2 \lesssim \dfrac{p_n}{n}.
    \end{equation}
    \item[(iii)] Now consider (II.II). 
    Recall $M_0$ is an $n$-by-$k$ matrix with i.i.d. Gaussian entries. 
    With probability at least $1-o(1)$, we have
    \begin{align*}
        \left\|\frac{n}{(n+\tau^{-2})^2}\Lambda_0(M_0^\top M_0-n\mathbb{I}_k)\Lambda_0^\top\right\| & \lesssim \frac{1}{n}\|\Lambda_0\|^2 \|M_0^\top M_0 - n\mathbb{I}_k\| & \\ 
        & \lesssim \frac{p_n \, \|M_0^\top M_0 - n \mathbb{I}_k\|}{n} & \text{[Assumption \ref{assumption:norm}]} \\
        & \lesssim \dfrac{p_n \sqrt{n \log n}}{n} & \text{[Lemma \ref{lemma:concentration-iid-Gaussian}]}\\
        & = p_n \sqrt{\dfrac{\log n}{n}}.
    \end{align*}
    \item[(iv)] \textcolor{black}{Next, we consider (II.IV):}
    \begin{equation*}
        \begin{split}
            & \left\|\frac{n}{(n+\tau^{-2})^2}(\Lambda_0 M_0^\top E + E^\top M_0\Lambda_0^\top + E^\top E)\right\| \lesssim \frac{1}{n}\left(2\left\|\Lambda_0 M_0^\top E\right\| + \left\|E^\top E\right\|\right).
        \end{split}
    \end{equation*}
    By Lemma \ref{lemma:spectral-norm-heteroskedastic-entries}, $\|E^\top E\| = \|E\|^2 \lesssim n+p_n+\sqrt{np_n}$ with probability at least $1- o(1)$. 
    \textcolor{black}{By Assumption \ref{assumption:norm} and Lemma \ref{lemma:spectral-norm-heteroskedastic-entries}, $\|\Lambda_0 M_0^{\top} E\| \leq \|\Lambda_0\|\,\|M_0\| \,  \|E\| \lesssim \sqrt{np_n} (\sqrt{n} + \sqrt{p_n})$ with probability at least $1-o(1)$.} 
    Thus, with probability at least $1-o(1)$, 
    \begin{align*}
        \left\|\frac{n}{(n+\tau^{-2})^2}(\Lambda_0^\top M_0 E + E^\top M_0\Lambda_0^\top + E^\top E)\right\| & \lesssim\frac{n+p_n + p_n\sqrt{n} + n \sqrt{p_n} + \sqrt{np_n}}{n}\\
        & \lesssim \dfrac{p_n \sqrt{n} + n \sqrt{p_n}}{n}.
    \end{align*}
    \item[(v)] \textcolor{black}{Now, we bound (II.I).}
    Since $U$ corresponds to the top $k$ left singular vectors of $\Y$, we have $\|\Y^\top U_{\perp}U_{\perp}^\top \Y\| = \|\Y^\top U_{\perp}\|^2 = s_{k+1}^2(\Y)$, where $s_{k+1}^2(\Y)$ is the $(k+1)$-th singular value of $\Y$. By Lemma \ref{lm:Eckart–Young}, we have
    \begin{align*}
        \|\Y^\top U_{\perp}U_{\perp}^\top \Y\| & = s_{k+1}^2(\Y)\\
        & = \min_{B: \text{rank}(B)\leq k}\|\Y-B\|^2 \\
        & \leq \|\Y - M_0\Lambda_0^\top\|^2 = \|E\|^2 \\
        & \lesssim n+p_n+\sqrt{np_n} & \textcolor{black}{\text{[Lemma \ref{lemma:spectral-norm-heteroskedastic-entries}]}}
    \end{align*}
    with probability at least $1-o(1)$.
    Thus, 
    $$\left\|\frac{-n}{(n+\tau^{-2})^2}\Y^\top U_{\perp}U_{\perp}^\top \Y\right\| \lesssim \frac{n+p_n + \sqrt{np_n}}{n}.$$
    \item[(vi)] \textcolor{black}{Finally, we consider (I):}
    \begin{equation*}
        \begin{split}
            \left\|\frac{\Y^\top \wh{ \mathbf{M}}\widetilde{E}^\top + \widetilde{E}\wh{ \mathbf{M}}^\top \Y}{n+\tau^{-2}} + \widetilde{E}\widetilde{E}^\top\right\|\lesssim \dfrac{2}{n}\left\|\Y^\top \wh{ \mathbf{M}}\widetilde{E}^\top\right\| + \|\widetilde{E}\widetilde{E}^\top\|.
        \end{split}
    \end{equation*}
    By Lemma \ref{lemma:spectral-norm-heteroskedastic-entries}, the following hold with probability at least $1-o_{P_0}(1)$:
    \begin{align*}     \|\widetilde{E}\widetilde{E}^\top\| & = \|\widetilde{E}\|^2 \lesssim \frac{p_n}{n} \rho^2 \left(\max_{1\leq j \leq p}\wt{\sigma}_{j}^2 \right) \lesssim \frac{p_n}{n}, \text{and} \\
    \frac{1}{n}\left\|\Y^\top \wh{ \mathbf{M}} \widetilde{E}^\top \right\| & \lesssim \frac{1}{n}\|\Y^\top \sqrt{n}U\| \left\|\widetilde{E}\right\| \lesssim \frac{\rho}{\sqrt{n}}\|\Y\| \sqrt{\frac{p_n}{n} \left(\max_{1\leq j\leq p}\wt{\sigma}^2_{j}\right)}\\
        & \lesssim \left(\frac{\sqrt{np_n} + \sqrt{n} + \sqrt{p_n}}{\sqrt{n}} \right) \sqrt{\frac{p_n}{n}} \lesssim \frac{p_n}{\sqrt{n}}.
    \end{align*}
    Thus, with probability at least $1-o_{P_0}(1)$, we have 
    $$\left\| \frac{\Y^\top \wh{ \mathbf{M}}\widetilde{E}^\top + \widetilde{E}\wh{ \mathbf{M}}^\top \Y}{(n+\tau^{-2})} + \widetilde{E}\widetilde{E}^\top\right\| \lesssim \dfrac{p_n}{\sqrt{n}}.$$
    \end{itemize}
    \textcolor{black}{Combining the previous steps (i)-(vi) along with $\|\Lambda_0\| \asymp \sqrt{p_n}$, we have obtained bounds for both (I) and (II). 
    Plugging the bounds into \eqref{suppeq:lambdaMainSplit}, we have
    $$\dfrac{\|\wt L - L_0\|}{\|L_0\|} \lesssim \dfrac{1}{p_n} \left( p_n \sqrt{\dfrac{{\log n}}{n}} + \dfrac{n\sqrt{p_n} + p_n \sqrt{n}}{n}\right) \lesssim \sqrt{\dfrac{\log n}{n}} + \dfrac{1}{\sqrt{p_n}},$$
    we have proved the desired result.
    }
\end{proof}
    
    \subsection{Proof of part (b)}
\begin{proof}
    We now show the contraction result for the FABLE-posterior of $\wt \Sigma$. 
    We start out by observing that $||\wt{\Sigma} - \Sigma_0|| = \underset{1 \leq j \leq p}{\max} |D_j|$, where $D_j = \wt{\sigma}_{j}^2 - \sigma_{0j}^2$ for $j=1,\ldots,p.$ 
    Let $\kappa_j = \wt{\sigma}_{j}^{-2}$, so that $\kappa_j \overset{ind}{\sim} G(\gamma_n / 2, \gamma_n \delta_j^2 / 2)$ under $\wt \Pi$. 
    Let 
$$\Delta_j = \dfrac{\gamma_n \delta_j^2}{2} \kappa_j - \dfrac{\gamma_n}{2}.$$
We can now express $D_j$ as
\begin{equation}
\label{suppeq:postConSigmaDj}
    D_j = \left(1 + \dfrac{2 \Delta_j}{\gamma_n}\right)^{-1} \, 
\left\{(\delta_j^2 - \sigma_{0j}^2) - \dfrac{2 \Delta_j}{\gamma_n}\sigma_{0j}^2\right\}.
\end{equation}
Using Lemma \ref{lemma:gamma-conc}, we have $\max_{1 \leq j \leq p}|\Delta_j| / \gamma_n \lesssim \{(\log p_n) / n\}^{1/3}$ with probability at least $1 - o_{P_0}(1).$ 
Thus, $\min_{1 \leq j \leq p}|1 + (2 \Delta_j / \gamma_n)| \gtrsim 1/2$ with probability at least $1-o_{P_0}(1).$ 
Therefore, with probability at least $1-o_{P_0}(1),$ \eqref{suppeq:postConSigmaDj} implies
\begin{equation}
\label{suppeq:DjBound}
    \max_{1 \leq j \leq p}|D_j| \lesssim 2 \max_{1 \leq j \leq p}|\delta_j^2 - \sigma_{0j}^2| + 4 \left(\max_{1 \leq j \leq p}\sigma_{0j}^2 \right) \max_{1 \leq j \leq p}\dfrac{|\Delta_j|}{\gamma_n}.
\end{equation}
Let $Q_j = \delta_j^2 - \sigma_{0j}^2$.
Using Lemma \ref{lemma:deltasq}, we can represent
$$Q_j = \dfrac{\sigma_{0j}^2}{n}\left\{\dfrac{Z_j}{\sigma_{0j}^2} - (n-k)\right\} - \dfrac{k\sigma_{0j}^2}{n} + F_j,$$
where $Z_j / \sigma_{0j}^2 \sim \chi_{n-k}^2 \equiv G\{(n-k)/2, 1/2\}$ and $\underset{1 \leq j \leq p_n}{\max}|F_j| \lesssim \dfrac{1}{\sqrt n} + \dfrac{1}{\sqrt{p_n}}$ with probability at least $1-o(1).$ \textcolor{black}{We now consider bounding the individual terms in $Q_j$.}
Using Lemma \ref{lemma:gamma-conc}, we obtain 
$$\max_{1 \leq j \leq p}\left|\dfrac{Z_j}{\sigma_{0j}^2} - (n-k)\right| \lesssim n \left(\dfrac{\log p_n}{n}\right)^{1/3}$$
with probability at least $1 - o(1)$. Thus, with probability at least $1-o(1)$, we get 
\begin{equation}
\label{suppeq:thm1Deltaj}
    \max_{1 \leq j \leq p}|Q_j| =  \max_{1 \leq j \leq p}|\delta_j^2 - \sigma_{0j}^2| \lesssim \left(\dfrac{\log p_n}{n}\right)^{1/3} + \dfrac{1}{\sqrt{p_n}}.
\end{equation}
We now bound the second term in \eqref{suppeq:DjBound}. 
From Lemma \ref{lemma:gamma-conc}, we have $\max_{1 \leq j \leq p} |\Delta_j| / \gamma_n \lesssim \left({\log p_n}/{n}\right)^{1/3}$ with probability at least $1-o_{P_0}(1)$. 
Combining the bounds for $Q_j$ and the above into \eqref{suppeq:DjBound}, we obtain $$\max_{1 \leq j \leq p}|D_j| = \|\wt \Sigma - \Sigma_0\| \lesssim \left(\dfrac{\log p_n}{n}\right)^{1/3} + \dfrac{1}{\sqrt{p_n}}.$$
This proves the result.
\end{proof}

\subsection{Proof of part (c)}
\begin{proof}
From Assumption \ref{assumption:norm}, we have $\|\Psi_0\| \asymp \|L_0\| \asymp p_n.$ By the triangle inequality, $\|\wt \Psi - \Psi_0\| \leq \|\wt L - L_0\| + \|\wt \Sigma - \Sigma_0\|$. Thus, under $\wt \Pi$, we can combine the bounds obtained from proofs of parts (a) and (b) to obtain
$$\dfrac{\|\wt{\Psi} - \Psi_0\|}{\|\Psi_0\|} \lesssim \dfrac{1}{\sqrt{p_n}} + \sqrt{\dfrac{\log n}{n}} + \dfrac{1}{p_n} \left[\left(\dfrac{\log p_n}{n}\right)^{1/3} +  \dfrac{1}{\sqrt{p_n}}\right] \lesssim \dfrac{1}{\sqrt{p_n}} + \sqrt{\dfrac{\log n}{n}}$$
with probability at least $1-o_{P_0}(1)$. 
This proves the result.
\end{proof}

\section{Proof of Theorem \ref{theorem:bvm}}
\label{proof:bvm}

\subsection{Proof Strategy}

We prove the result in two parts: first for $u \neq v$ and then for $u = v$.
We can approximate the entrywise distribution of the FABLE-posterior samples with a suitable Gaussian distribution by showing that the remainder of the terms concentrate/converge to $0$.
For $u = v$, we use the decomposition in \eqref{suppeq:postConSigmaDj} along with the Gaussian approximation for $u < v$ to establish the result. 
The decomposition in \eqref{suppeq:postConSigmaDj} provides a $\chi^2_n$ distribution which is further approximated by a Gaussian distribution using the CLT, and is independent of the Gaussian approximation from the proof for $u < v$, and thus the variances add up for the final result.

\subsection{Proof for $u \neq v$:}
    \begin{proof}
        \textcolor{black}{Under the FABLE-posterior, we can represent $\lambda_u = \mu_u + \rho \left(\wt \sigma_u \tilde e_u / \sqrt{n + \tau^{-2}}\right)$ with $\tilde{e}_{u}, \tilde{e}_{v} \overset{ind}{\sim} N_k(0, \mathbb{I}_k)$, so that}
        \begin{equation}
            \label{suppeq:thm2Luv}
            \wt L_{uv} = \wt \lambda_u^\top \wt \lambda_v = \mu_u^{\top} \mu_v + \dfrac{\rho \left(\wt{\sigma}_{v} \mu_u^{\top} \tilde{e}_{v} + \wt{\sigma}_{u} \mu_v^{\top} \tilde{e}_{u} \right)}{\sqrt{n + \tau^{-2}}} + \dfrac{\rho^2 \wt{\sigma}_{u} \wt{\sigma}_{v} \tilde{e}_{u}^{\top} \tilde{e}_{v}}{n + \tau^{-2}},
        \end{equation}
    First let $u \neq v$, so that $\wt \Psi_{uv} = \wt L_{uv}$.
    Since $\tilde e_u, \tilde e_v \overset{\text{ind}}{\sim} N_k(0, \mathbb{I}_k)$, we can express the second term in the RHS of \eqref{suppeq:thm2Luv} as
    \begin{align*}
        \rho\left(\wt \sigma_{v} \mu_u^{\top} \tilde{e}_{v} + \wt \sigma_{u} \mu_v^{\top} \tilde{e}_{u}\right) & = 
         l_{0, uv}(\rho) R_{uv} + d_{uv}(\rho) R_{uv},
    \end{align*}
    for $R_{uv} \sim N(0,1)$ such that $R_{uv}$ is independent of $\mathbf{Y}, \wt\sigma_{u}^2, $ and $\wt \sigma_{v}^2$, with $$l_{0,uv}^2(\rho) = \rho^2 \left(\sigma_{0v}^2 \|\lambda_{0u}\|_2^2 + \sigma_{0u}^2 \|\lambda_{0v}\|_2^2\right),  \text{ and}$$ $$d_{uv}(\rho) = \rho(\wt \sigma_{v}^2 ||\mu_u||_2^2 + \wt \sigma_{u}^2 ||\mu_v||_2^2)^{1/2} - l_{0, uv}(\rho).$$ 
    Therefore, we can express \eqref{suppeq:thm2Luv} as
    \begin{align}
    \label{suppeq:thm2main}
        & \sqrt{n}(\wt L_{uv} - 
        \mu_u^{\top} \mu_v) = \dfrac{\sqrt{n}}{\sqrt{n + \tau^{-2}}} l_{0, uv}(\rho) R_{uv} + d_{uv}^{*}(\rho),
    \end{align}
    where 
    $$d_{uv}^{*}(\rho) = \sqrt{\dfrac{n}{n+\tau^{-2}}}  d_{uv}(\rho) R_{uv}  + \sqrt{n} \left( \dfrac{\rho^2 \wt \sigma_{v} \wt\sigma_{u} \tilde{e}_{u}^{\top} \tilde{e}_{v}}{n + \tau^{-2}}\right).$$
    Due to FABLE-posterior concentration of $\wt\sigma_{j}^2$ to $\sigma_{0j}^2$ (from part (b) of Theorem \ref{theorem:post-con}) and convergence in probability of $||\mu_j||_2^2$ to $||\lambda_{0j}||_2^2$ from Lemma \ref{lemma:post-mean-conv}, we have FABLE-posterior concentration of $d_{uv}(\rho)$ around $0$ as $n \to \infty$ for any finite $\rho > 0$. 
    Furthermore, with probability at least $1-o_{P_0}(1)$, we have
    $$\sqrt{n} \left(\dfrac{\wt{\sigma}_{u} \wt{\sigma}_{v} \tilde{e}_{u}^{\top} \tilde{e}_{v}}{n + \tau^{-2}}\right) \lesssim \dfrac{1}{\sqrt{n}}.$$
    Thus, the FABLE-posterior of $d_{uv}^{*}(\rho)$ concentrates around $0$ asymptotically. As $n \to \infty$, Lemma \ref{lemma:bvm-approx} with \eqref{suppeq:thm2main} leads to
    $$\underset{x}{\sup} \, \left|\wt{\Pi}\left\{ \dfrac{\sqrt{n}(\wt L_{uv} - \mu_u^{\top} \mu_v)}{l_{0, uv}(\rho)} \leq   x\right\} - \Phi(x)\right| \overset{P_0}{\to} 0.$$


\subsection{Proof for $u=v$:}
For $u=v$, we first observe that $\sqrt{n}(\wt \Psi_{uu} - T_{uu}) = \sqrt{n} (\wt L_{uu} - \|\mu_u\|_2^2) + \sqrt{n}(\wt \sigma_{u}^2 - \delta_{u}^2).$ 
We can use the decomposition in \eqref{suppeq:postConSigmaDj} to express
\begin{equation}
\label{suppeq:thm2sigdecomp}
    \sqrt{n}(\wt \sigma_{u}^2 - \delta_u^2) = -\dfrac{2 \sqrt{n} \, \Delta_u}{\gamma_n}\sigma_{0u}^2 + \mathcal{G}_u,
\end{equation}
where
\begin{equation}
\label{suppeq:thm2Gj}
    \mathcal{G}_u = -\dfrac{2\sqrt{n}(\delta_u^2 - \sigma_{0u}^2) (\Delta_u / \gamma_n)}{1 + (2 \Delta_u / \gamma_n)} + \dfrac{4\sqrt{n} \sigma_{0u}^2 (\Delta_u^2 / \gamma_n^2)}{1 + (2 \Delta_u / \gamma_n)}.
\end{equation}
Letting $ \mathcal{N}_u = -\dfrac{2 \sqrt{n} \, \Delta_u}{\gamma_n} \sigma_{0u}^2,$ we have 
$\mathcal{N}_u / \sqrt{2 \sigma_{0u}^4} \implies N(0, 1)$, using the normal approximation to the Gamma distribution. 
We now consider bounding each of the terms on the RHS of \eqref{suppeq:thm2Gj}. 
The first term is bounded as
\begin{align*}
    \left|-\dfrac{2\sqrt{n}(\delta_u^2 - \sigma_{0u}^2) (\Delta_u / \gamma_n)}{1 + (2 \Delta_u / \gamma_n)}\right| & \asymp |\mathcal{N}_u| \, |\delta_u^2 - \sigma_{0u}^2| \\
    & \lesssim \left(\dfrac{\log p_n}{n}\right)^{1/3} + \dfrac{1}{\sqrt{p_n}}
\end{align*}
with probability at least $1 - o_{P_0}(1)$, by the result on $\mathcal{N}_u$ and the bound on $|\delta_u^2 - \sigma_{0u}^2|$ from \eqref{suppeq:thm1Deltaj}. 
The second term is analogously bounded:
\begin{align*}
    \left|\dfrac{4\sqrt{n} \sigma_{0u}^2 (\Delta_u^2 / \gamma_n^2)}{1 + (2 \Delta_u / \gamma_n)}\right| \asymp \dfrac{1}{\sqrt{n}}|\mathcal{N}_u|^2 \lesssim \dfrac{1}{\sqrt{n}},
\end{align*}
which provides $|\mathcal{G}_u| \lesssim o(1)$ with probability at least $1 - o_{P_0}(1)$.
Thus, we can combine \eqref{suppeq:thm2Luv} and \eqref{suppeq:thm2sigdecomp} to write
$$\sqrt{n}\left(\wt\Psi_{uu} - T_{uu}\right) = \left(2\rho \, \sigma_{0u} \|\lambda_{0u}\|_2 \right) R_{uu} + \mathcal{N}_u + \mathcal{G}_u,$$
as $l_{0,uu}^2(\rho) = 4 \rho^2 \sigma_{0u}^2 \|\lambda_{0u}\|_2^2$.
By construction, $R_{uu}$ is independent of $\mathcal{N}_u$ for $1 \leq u \leq p$. Combining all these observations and using Lemmas \ref{lemma:slutsky-extension} and \ref{lemma:bvm-approx}, we have the desired result for $u=v$, with $l_{0,uu}^2(\rho) = 2\sigma_{0u}^4 + 4\rho^2 \sigma_{0u}^2 \|\lambda_{0u}\|_2^2$ for $1 \leq u \leq p$.
\end{proof}

\section{Proof of Theorem \ref{theorem:asymp-law}} 
\label{proof:asymp-law}

\subsection{Proof Strategy}

    For $1 \leq u , v \leq p$, let 
\begin{align}
\label{suppeq:Suv}
    S_{uv} = \dfrac{n}{(n+\tau^{-2})^2}y^{(u) \top}UU^{\top} y^{(v)} - \lambda_{0u}^{\top} \lambda_{0v}.
\end{align}
For $u \neq v$, the proof will involve decomposing $S_{uv}$ into a number of components, with some components contributing to the asymptotic Gaussian law, and the rest converging to $0$ in probability. 
There are two components that (asymptotically) independently contribute to the limiting Gaussian distribution.
The first component involves the randomness of the true latent factor matrix $M_0$ itself.
The law of the second component depends on the error distributions $\epsilon^{(u)}$.
We will use different parts of Proposition \ref{proposition:bod} to establish the result, along with other Lemmas (highlighted in this document) and well-known results.
A key observation is that the bound on the spectral norm of $UU^\top - U_0 U_0^\top$ obtained from part (a) of Proposition \ref{proposition:bod} does not suffice to establish the UQ result, as $\sqrt{n}(n^{-1/2} + p_n^{-1/2})$ does not converge to $0$ as $n \to \infty$.
\textbf{Instead, we require a more refined leave-one-out argument} leading to part (c) of Proposition \ref{proposition:bod}, that helps bound some of the terms obtained from the decomposition of $S_{uv}$.
Some of the remainder terms are bounded using part (b) of Proposition \ref{proposition:bod}.

For $u = v$, the proof is considerably simpler.
We establish that the  estimator $T_{uu}$ is asymptotically equivalent to $(1/n) \|y^{(u)}\|_2^2 \sim (\|\lambda_{0u}\|_2^2 + \sigma_{0u}^2) \, \chi_n^2$, as one would expect, and thus a direct application of the CLT can be used to obtain the result.

\subsection{Proof for $u \neq v$:}
\begin{proof} 
We first assume $ u \neq v$. 
Since $y^{(u)} = M_0 \lambda_{0u} + \epsilon^{(u)}$, we can decompose $S_{uv}$ in \eqref{suppeq:Suv} as
$S_{uv} = D_{uv} + R_{uv}$ after some algebraic manipulation, where
\begin{align}
    D_{uv} & = \lambda_{0u}^{\top} \left(\dfrac{1}{n} M_0^{\top} U_0 U_0^{\top} M_0 - \mathbb{I}_k\right)\lambda_{0v} + \dfrac{1}{n}\left(\lambda_{0u}^{\top}M_0^{\top}U_0 U_0^{\top} \epsilon^{(v)} + \lambda_{0v}^{\top} M_0^{\top}U_0 U_0^{\top} \epsilon^{(u)}\right), \notag \\
    & = \text{(D.I) + (D.II), say.} \label{suppeq:thm3Duv}\\
    R_{uv} & = \, \Delta_n \lambda_{0u}^{\top} M_0^{\top}U_0 U_0^{\top} M_0 \lambda_{0v} \notag \\
    & + f_n \, \lambda_{0u}^{\top}M_0^{\top}(UU^{\top} - U_0U_0^{\top})M_0 \lambda_{0v} \notag \\
    & + \Delta_n \left[\lambda_{0u}^{\top} M_0^{\top}U_0 U_0^{\top} \epsilon^{(v)} + \lambda_{0v}^{\top} M_0^{\top} U_0 U_0^{\top} \epsilon^{(u)}\right] \notag \\
    & + f_n\left[\lambda_{0u}^{\top} M_0^{\top} (UU^{\top} - U_0 U_0^{\top})\epsilon^{(v)} + \lambda_{0v}^{\top} M_0^{\top}(UU^{\top} - U_0U_0^{\top})\epsilon^{(u)}\right] \notag \\
    & + f_n \, \epsilon^{(u) \top} UU^{\top} \epsilon^{(v)} \notag\\
    & = \text{(R.I) + (R.II) + (R.III) + (R.IV) + (R.V), say.} \label{suppeq:thm3Ruv}
\end{align}
with $f_n = n / (n + \tau^{-2})^2 = (1/n) + \Delta_n$ such that $|\Delta_n| \asymp (1/n^2).$ 
We will show that $\sqrt{n} R_{uv} = o_{P_0}(1)$ by individually bounding each of its terms.
We then derive an asymptotic Gaussian limit distribution of $\sqrt n D_{uv}$.
We now consider each of the terms in $\sqrt n R_{uv}$ as in \eqref{suppeq:thm3Ruv}.
\textcolor{black}{Throughout, we use the identity 
 \begin{equation}
    \label{suppeq:thm3projection}
    \Lambda_0 M_0^\top U_0 U_0^\top = V_0 D_0 U_0^\top U_0 U_0^\top = V_0 D_0 U_0^\top = \Lambda_0 M_0^\top
 \end{equation}
 along with $\|\lambda_{0u}\|_2 = \mathcal{O}(1) \text{ and } \|\lambda_{0v}\|_2 = \mathcal{O}(1)$, from Assumption \ref{assumption:norm}.
 We also remark that Assumptions \ref{assumption:dim} and \ref{assumption:uq} together imply that 
 $$\dfrac{\log n}{p_n} = \dfrac{\log n}{\sqrt{n}} \times \dfrac{\sqrt{n}}{p_n} = o(1).$$
 }
\begin{itemize}
    \item[(i)] First consider (R.V).
    We bound $\sqrt{n} f_n \epsilon^{(u) \top} UU^{\top} \epsilon^{(v)}$ as follows:
\begin{align*}
    \left|\sqrt{n}f_n \epsilon^{(u) \top } UU^{\top} \epsilon^{(v)}\right|
    & \leq \sqrt{n}f_n \left|\epsilon^{(u) \top} U_0 U_0^{\top} \epsilon^{(v)}\right| + \sqrt{n} f_n \left|\epsilon^{(u) \top} (UU^{\top} - U_0 U_0^{\top}) \, \epsilon^{(v)}\right|.
\end{align*}
The first term is $\sqrt{n} f_n \left|\epsilon^{(u)} U_0 U_0^{\top} \epsilon^{(v)} \right| \, = \mathcal{O}_{P_0}(\sqrt{n} f_n) = \mathcal{O}_{P_0}(1/\sqrt{n})$, since $U_0^{\top} \epsilon^{(u)} \sim N(0, \sigma_{0u}^2 \mathbb{I}_k)$ with finite $k$ and $f_n \asymp (1/n)$. 
The second term is handled by observing that
{\color{black}
\begin{align*}
    \sqrt{n} f_n & \left|\epsilon^{(u) \top} (UU^{\top} - U_0 U_0^{\top}) \epsilon^{(v)}\right|\\ 
    & \leq \sqrt{n} f_n \|\epsilon^{(v)}\|_2 \|(UU^{\top} - U_0 U_0^{\top})\epsilon^{(u)}\|_2 \\
    & \lesssim \|(UU^{\top} - U_0 U_0^{\top})\epsilon^{(u)}\|_2 & \textcolor{black}{\text{[Lemma \ref{lemma:laurent-massart}]}} \\
    & \lesssim \sqrt{\log n}\left(\dfrac{1}{\sqrt{n}} + \dfrac{1}{\sqrt{p_n}}\right) + \dfrac{\sqrt{n}}{p_n} \quad & \text{[Proposition \ref{proposition:bod}, (c)]}\\ 
    & = o(1) \quad & \text{[Assumption \ref{assumption:uq}]} &,
\end{align*}
with probability at least $1-o(1).$
}
Thus, 
we have $\sqrt{n} \left|f_n \epsilon^{(u)\top} UU^{\top} \epsilon^{(v)}\right| = o_{P_0}(1)$.
    \item[(ii)] Next, we consider (R.I):
    \begin{align*}
     \sqrt{n} \Delta_n \left|\lambda_{0u}^{\top} M_0^{\top} U_0 U_0^{\top} M_0 \lambda_{0v}\right| & = \sqrt{n} \Delta_n \left|\lambda_{0u}^\top M_0^{\top}M_0 \lambda_{0v}\right| & \textcolor{black}{\text{[Using \eqref{suppeq:thm3projection}]}} \\
&  \leq \sqrt{n} \Delta_n \|M_0\|^2 \|\lambda_{0u}\|_2 \|\lambda_{0v}\|_2 \\
& \lesssim \dfrac{1}{\sqrt{n}} & \textcolor{black}{\text{[Lemma \ref{lemma:laurent-massart}]}}\\ 
& = o(1),  
    \end{align*}
    with probability at least $1-o(1).$ 
    Thus, $\sqrt{n} \Delta_n \left|\lambda_{0u}^{\top} M_0^{\top} U_0 U_0^{\top} M_0 \lambda_{0v}\right| = o_{P_0}(1).$
    \item[(iii)] 
    {\color{black}
    Next, we consider (R.II):
    \begin{align*}
     & \sqrt{n} f_n \left|\lambda_{0u}^{\top} M_0^{\top}(UU^{\top} - U_0U_0^{\top})M_0 \lambda_{0v}\right|\\ 
     & \leq \sqrt{n}f_n \|M_0^{\top}(UU^{\top} - U_0U_0^{\top})M_0\| \|\lambda_{0u}\|_2 \|\lambda_{0v}\|_2\\
     & \lesssim \sqrt{n}f_n \|M_0^{\top}(UU^{\top} - U_0U_0^{\top})M_0\| \quad & \text{[Assumption \ref{assumption:norm}]}\\
     & \lesssim \dfrac{1}{\sqrt{n}} \left(1 + \dfrac{n}{p_n}\right) \quad & \text{[Proposition \ref{proposition:bod}, (b)]}\\
     & = \dfrac{1}{\sqrt{n}} + \dfrac{\sqrt{n}}{p_n} \\ 
     & = o(1), \quad & \text{[Assumption \ref{assumption:uq}]}
    \end{align*}
    with probability at least $1-o(1)$.
    Thus, $$\sqrt{n} f_n \left|\lambda_{0u}^{\top} M_0^{\top}(UU^{\top} - U_0U_0^{\top})M_0 \lambda_{0v}\right| = o_{P_0}(1).$$
    }
    \item[(iv)] Next, we consider (R.III):
    \begin{align*}
        & \sqrt{n} \Delta_n \left|\lambda_{0u}^{\top}  M_0^{\top} U_0 U_0^{\top} \epsilon^{(v)} + \lambda_{0v}^{\top}  M_0^{\top} U_0 U_0^{\top} \epsilon^{(u)}\right|\\ 
        = & \, \sqrt{n} \Delta_n \left|\lambda_{0u}^{\top} M_0^{\top} \epsilon^{(v)} +
        \lambda_{0v}^{\top} M_0^{\top} \epsilon^{(u)}\right| & \textcolor{black}{\text{[Using \eqref{suppeq:thm3projection}]}} \\
        \leq & \, \sqrt{n} \Delta_n (\|\lambda_{0u}\|_2 \|M_0\| \|\epsilon^{(v)}\|_2 + \|\lambda_{0v}\|_2 \|M_0\| \|\epsilon^{(u)}\|_2) \\
        \lesssim & \, \dfrac{1}{n^2}  (\sqrt{n})^3 & \textcolor{black}{\text{[Lemmas \ref{lemma:laurent-massart} and \ref{lemma:concentration-iid-Gaussian}]}} \\
        = & \, \dfrac{1}{\sqrt{n}},
    \end{align*}
    with probability at least $1-o(1).$
    Thus, 
    $$\sqrt{n} \Delta_n \left|\lambda_{0u}^{\top}  M_0^{\top} U_0 U_0^{\top} \epsilon^{(v)} + \lambda_{0v}^{\top}  M_0^{\top} U_0 U_0^{\top} \epsilon^{(u)}\right| = o_{P_0}(1).$$
    \item[(v)] Finally, 
    we consider (R.IV):
    {\color{black}
    \begin{align*}
        & \sqrt{n} f_n \left|\lambda_{0u}^{\top} M_0^{\top} (UU^{\top} - U_0 U_0^{\top})\epsilon^{(v)} +   \lambda_{0v}^{\top} M_0^{\top}(UU^{\top} - U_0U_0^{\top})\epsilon^{(u)}\right| \\
        \lesssim & \, \dfrac{1}{\sqrt{n}} \|M_0^{\top}\| \|(UU^{\top} - U_0U_0^{\top}) \epsilon^{(v)}\|_2\\
        \lesssim & \sqrt{\log n}\left(\dfrac{1}{\sqrt{n}} + \dfrac{1}{\sqrt{p_n}}\right) + \dfrac{\sqrt{n}}{p_n} \hspace{10em} \text{[Proposition \ref{proposition:bod}, (c)]}\\ 
        = & \, o(1) \hspace{22.35em} \text{[Assumption \ref{assumption:uq}]},
    \end{align*}
    with probability at least $1-o(1)$.
    Thus, 
    $$\sqrt{n} f_n \left|\lambda_{0u}^{\top} M_0^{\top} (UU^{\top} - U_0 U_0^{\top})\epsilon^{(v)} + \lambda_{0v}^{\top} M_0^{\top}(UU^{\top} - U_0U_0^{\top})\epsilon^{(u)}\right| = o_{P_0}(1).$$
    }
\end{itemize}

\noindent 
Putting steps (i)-(v) together implies
    $\sqrt{n} |R_{uv}| = o_{P_0}(1).$ 
    
    Next, we consider obtaining the limiting law of $\sqrt{n} D_{uv}$. 
    We first consider (D.I):
    \begin{align*}
        \sqrt{n} \lambda_{0u}^{\top} \left(\dfrac{1}{n} M_0^{\top} U_0 U_0^{\top} M_0 - \mathbb{I}_k\right) \lambda_{0v} & = \sqrt{n} \lambda_{0u}^{\top} \left(\dfrac{1}{n} M_0^{\top}  M_0 - \mathbb{I}_k\right) \lambda_{0v}\\
        & = \sqrt{n}(\overline{V}_n - \lambda_{0u}^{\top} \lambda_{0v})\\
        & {\implies} N(0, \xi_{0, uv}^2),
    \end{align*}
    as $n \to \infty$, using the central limit theorem (CLT), where $\overline{V}_n = (1/n) \sum_{i=1}^{n} V_i$ with $V_i := (\lambda_{0u}^{\top} \eta_{0i})(\lambda_{0v}^{\top} \eta_{0i})$ for $i=1,\ldots,n$, and 
    \begin{align*}
        \mathcal{\xi}_{0,uv}^2 = \mbox{var}\left[(\lambda_{0u}^{\top} \eta_{01})(\lambda_{0v}^{\top} \eta_{01})\right] & = (\lambda_{0u}^{\top} \lambda_{0v})^2 + ||\lambda_{0u}||_2^2 ||\lambda_{0v}||_2^2. 
    \end{align*}
    Next, we consider (D.II):
$$\sqrt{n}\left\{\dfrac{1}{n}\left(\lambda_{0u}^{\top}M_0^{\top}U_0 U_0^{\top} \epsilon^{(v)} + \lambda_{0v}^{\top} M_0^{\top}U_0 U_0^{\top} \epsilon^{(u)}\right)\right\} = \sqrt{n}\left\{\dfrac{1}{n}\left(\lambda_{0u}^{\top}M_0^{\top}\epsilon^{(v)} + \lambda_{0v}^{\top} M_0^{\top} \epsilon^{(u)}\right)\right\}.$$
    Let $$l_{uv}^2(M_0) = \sigma_{0v}^2\lambda_{0u}^{\top} \dfrac{M_0^{\top} M_0}{n} \lambda_{0u} + \sigma_{0u}^2\lambda_{0v}^{\top} \dfrac{M_0^{\top} M_0}{n} \lambda_{0v}.$$
    Since $\epsilon^{(u)}$ and $\epsilon^{(v)}$ are independent for $u \neq v$, we have
$$\left.\dfrac{\sqrt{n}}{n}\left(\lambda_{0u}^{\top}M_0^{\top}\epsilon^{(v)} + \lambda_{0v}^{\top} M_0^{\top} \epsilon^{(u)}\right) \, \right \vert M_0 \sim N(0, l_{uv}^2(M_0)) = l_{uv}(M_0) Z_{uv},$$
    where $Z_{uv} \sim N(0,1)$ and $Z_{uv}$ is independent of $M_0.$ Let
    $$l_{0, uv}^2 = \sigma_{0v}^2 ||\lambda_{0u}||_2^2 + \sigma_{0u}^2 ||\lambda_{0v}||_2^2$$
    be the limit (in probability) of $l_{uv}(M_0)$ as $n \to \infty$.
    Then, combining the expressions for (D.I) and (D.II) in \eqref{suppeq:thm3Duv}, we have
    \begin{align}
    \label{suppeq:thm3DuvDecomp}
        \sqrt{n} D_{uv} & = \sqrt{n} \lambda_{0u}^{\top} \left(\dfrac{1}{n} M_0^{\top} U_0 U_0^{\top} M_0 - \mathbb{I}_k\right) \lambda_{0v} + l_{uv}(M_0)Z_{uv} \notag\\
        & = \sqrt{n} \lambda_{0u}^{\top} \left(\dfrac{1}{n} M_0^{\top}  M_0 - \mathbb{I}_k\right) \lambda_{0v} + l_{0, uv}Z_{uv} + \{l_{uv}(M_0) - l_{0, uv}\}Z_{uv}.
    \end{align}
    Since $\|M_0\| \asymp \sqrt{n}$, we have $\left|l_{uv}(M_0) + l_{0, uv}\right| = \mathcal{O}_{P_0}(1)$ and $\left|l_{uv}(M_0) + l_{0,uv}\right| \geq l_{0,uv} \geq \sqrt{2c_1^2c_2} > 0$, from Assumptions \ref{assumption:norm} and \ref{assumption:var}. 
    Thus, the third term in the RHS of \eqref{suppeq:thm3DuvDecomp} can be bounded in probability by observing
    \begin{align*}
         \{l_{uv}(M_0) - l_{0, uv}\}Z_{uv} & = \dfrac{l_{uv}^2(M_0) - l_{0, uv}^2}{l_{uv}(M_0) + l_{0, uv}} Z_{uv}\\
         & = \mathcal{O}_{P_0}\left(\sqrt{\dfrac{\log n}{n}}\right),
    \end{align*}
    as $\left\|\dfrac{M_0^{\top}M_0}{n} - \mathbb{I}_k\right\| = \mathcal{O}_{P_0}\left(\sqrt{\dfrac{\log n}{n}}\right)$ from Lemma \ref{lemma:concentration-iid-Gaussian}.
Since $Z_{uv}$ is independent of $M_0$, Lemma \ref{lemma:slutsky-extension} immediately implies that
    $$\sqrt{n} D_{uv} \overset{d}{=} N(0, \xi_{0, uv}^2 + l_{0, uv}^2) + o_{P_0}(1).$$
Let $\mathcal{S}_{0,uv}^2 = l_{0,uv}^2 + \xi_{0,uv}^2$. 
Since $\sqrt{n} R_{uv} = o_{P_0}(1)$, we now invoke Slutsky's theorem to obtain
    $$\sqrt{n} S_{uv} = \sqrt{n} D_{uv} + \sqrt{n} R_{uv} \overset{}{\implies} N(0, \mathcal{S}_{0, uv}^2),$$
which shows the desired result.
\end{proof}

\subsection{Proof for $u = v$:}
\begin{proof}
    {\color{black}
    We now prove the result for $u = v$. 
    Firstly, we can rewrite $\sqrt{n} S_{uu}$ as in \eqref{suppeq:Suv} as
    \begin{align*}
        \sqrt{n} S_{uu} &= \sqrt{n} \left[\dfrac{1}{n} y^{(u)\top} UU^\top y^{(u)} - \|\lambda_{0u}\|_2^2\right] + \sqrt{n} \,\Delta_n y^{(u)\top} UU^\top y^{(u)},
    \end{align*}
    where $\Delta_n = \dfrac{n}{(n+\tau^{-2})^2} - \dfrac{1}{n}$ so that $|\Delta_n| \asymp \dfrac{1}{n^2}$.
    Using Lemma \ref{lemma:laurent-massart}, the second term above can be bounded as $\sqrt{n}|\Delta_n y^{(u)\top} UU^\top y^{(u)}| \leq \sqrt{n}|\Delta_n| \|y^{(u)}\|_2^2 \|UU^\top\| \lesssim \sqrt{n}\dfrac{1}{n^2}(\sqrt{n})^2 = \dfrac{1}{\sqrt{n}},$
    with probability at least $1-o(1)$, and therefore
    \begin{align}
        \label{suppeq:SuvFinal}
        \sqrt{n} S_{uu} = \sqrt{n} \left[\dfrac{1}{n} y^{(u)\top} UU^\top y^{(u)} - \|\lambda_{0u}\|_2^2\right] + o_{P_0}(1).
    \end{align}
    Define $F_{uu} = \delta_{u}^2 - \sigma_{0u}^2$ and consider, from the definition of $\delta_u^2$ in \eqref{eq:hyper},
    \begin{align}
        \sqrt{n} F_{uu} 
        = & \, \sqrt{n} \left\{\dfrac{\gamma_0 \delta_0^2}{\gamma_n} + \dfrac{1}{\gamma_n}\left(\|y^{(u)}\|_2^2 - \dfrac{ny^{(u)\top} UU^\top y^{(u)}}{n+\tau^{-2}}\right) - \sigma_{0u}^2\right\} \notag\\
        = & \, \sqrt{n} \left\{\dfrac{1}{n} y^{(u)\top}(\mathbb{I}_n - UU^\top) y^{(u)} - \sigma_{0u}^2 \right\} + \notag\\
        & \sqrt{n} y^{(u)\top} (\mathbb{I}_n - UU^\top) y^{(u)}\left(\dfrac{1}{\gamma_n} - \dfrac{1}{n}\right) + \notag\\
        & \sqrt{n} \left\{\dfrac{\gamma_0 \delta_0^2}{\gamma_n} + \dfrac{\tau^{-2}}{\gamma_n(n+ \tau^{-2})} y^{(u)\top}UU^\top y^{(u)}\right\}. \label{suppeq:Duu}
    \end{align}
    The second term in \eqref{suppeq:Duu} can be bounded as
    \begin{align*}
        \left|\sqrt{n} y^{(u)\top} (\mathbb{I}_n - UU^\top) y^{(u)}\left(\dfrac{1}{\gamma_n} - \dfrac{1}{n}\right)\right| & \lesssim \sqrt{n} (\sqrt{n})^2 \dfrac{1}{n^2} = \dfrac{1}{\sqrt{n}},
    \end{align*}
    with probability at least $1-o(1)$. 
    The third term in \eqref{suppeq:Duu} can be bounded as 
    \begin{align*}
        \left|\sqrt{n} \left\{\dfrac{\gamma_0 \delta_0^2}{\gamma_n} + \dfrac{\tau^{-2}}{\gamma_n(n+ \tau^{-2})} y^{(u)\top}UU^\top y^{(u)}\right\}\right| & \lesssim\dfrac{1}{\sqrt{n}} + \dfrac{\sqrt{n} (\sqrt{n})^2}{n \times n} \lesssim \dfrac{1}{\sqrt{n}},
    \end{align*}
    with probability at least $1-o(1).$
    Combining these results into \eqref{suppeq:Duu}, we have
    \begin{align}
        \label{suppeq:DuuFinal}
        \sqrt{n} F_{uu} = \sqrt{n} \left\{\dfrac{1}{n} y^{(u)\top}(\mathbb{I}_n - UU^\top) y^{(u)} - \sigma_{0u}^2 \right\} + o_{P_0}(1).
    \end{align}
    Combining \eqref{suppeq:SuvFinal} and \eqref{suppeq:DuuFinal}, we have
    \begin{align*}
        \sqrt{n}(S_{uu} + F_{uu}) & = \sqrt{n}\left\{\dfrac{1}{n} \|y^{(u)}\|_2^2 - \left(\|\lambda_{0u}\|_2^2 + \sigma_{0u}^2 \right)\right\} + o_{P_0}(1),
    \end{align*}
    and the desired result is obtained by observing that (i) $T_{uu} - (\|\lambda_{0u}\|_2^2 + \sigma_{0u}^2) = S_{uu} + F_{uu}$ and (ii) $\|y^{(u)}\|_2^2 \sim (\|\lambda_{0u}\|_2^2 + \sigma_{0u}^2) \, \chi_n^2$ and invoking the Gaussian approximation of the $\chi_n^2$ distribution via the CLT.
    }
\end{proof}

\section{Proofs of Proposition \ref{proposition:bod} and Related Lemmas for Theorem \ref{theorem:post-con}}
\label{proof:lemmaPostCon}

\subsection{Proof of Proposition \ref{proposition:bod}}
\label{proofSubSec:bod}

\begin{proof}
{\color{black}
    We will use the statistician-friendly Davis-Kahan theorem (henceforth referred to as the SFDKT), available as Theorem 2 of \cite{yu2015useful}, throughout the proof. 
    We restate the result below for convenience:
    
    \vspace{1em}
    
    \noindent \textbf{[Theorem 2 in \cite{yu2015useful}.]}
Let $\Sigma, \widehat{\Sigma} \in \mathbb{R}^{p \times p}$ be symmetric, with eigenvalues
$\lambda_1 \ge \cdots \ge \lambda_p$ and
$\widehat{\lambda}_1 \ge \cdots \ge \widehat{\lambda}_p$, respectively.
Fix $1 \le r \le s \le p$ and assume that
\[
\min\{\lambda_{r-1}-\lambda_r,\,
      \lambda_s-\lambda_{s+1}\}>0,
\]
where $\lambda_0:=\infty$ and $\lambda_{p+1}:=-\infty$.

Let $d=s-r+1$, and define
\[
V=(v_r,\ldots,v_s)\in\mathbb{R}^{p\times d},
\qquad
\widehat V=(\widehat v_r,\ldots,\widehat v_s)\in\mathbb{R}^{p\times d},
\]
where
\[
\Sigma v_j=\lambda_jv_j,
\qquad
\widehat\Sigma\,\widehat v_j=\widehat\lambda_j\widehat v_j,
\qquad
j=r,\ldots,s.
\]

Then
\[
\|\sin\Theta(\widehat V,V)\|_\text{F}
\le
\frac{
2\min\!\left\{
\sqrt{d}\,\|\widehat\Sigma-\Sigma\|_{},
\;
\|\widehat\Sigma-\Sigma\|_\text{F}
\right\}
}{
\min\{\lambda_{r-1}-\lambda_r,\,
      \lambda_s-\lambda_{s+1}\}
}.
\]

Moreover, there exists an orthogonal matrix
$\widehat O\in\mathbb{R}^{d\times d}$ such that
\[
\|\widehat V\widehat O-V\|_\text{F}
\le
\frac{
2^{3/2}\min\!\left\{
\sqrt{d}\,\|\widehat\Sigma-\Sigma\|_{},
\;
\|\widehat\Sigma-\Sigma\|_\text{F}
\right\}
}{
\min\{\lambda_{r-1}-\lambda_r,\,
      \lambda_s-\lambda_{s+1}\}
}.
\]
    We also use Lemma 1 of \cite{cai2018rate}, showing the equivalence of $\sin \Theta$ and subspace distances between projection operators (upto universal constants); that is
    $$\| \sin\Theta(U_1, U_2) \| \leq \| U_1U_1^\top - U_2 U_2^\top \|\leq 2\| \sin\Theta(U_1, U_2) \| $$
    for any two left subspace singular matrices $U_1, U_2.$
    We also use the fact that $\|A\|\leq \|A\|_\text{F}$ for any matrix $A$, where $\|A\|_\text{F}$ denotes the Frobenius norm of $A$.
    \begin{enumerate}
        \item[(a)] \, From the SVD of $\Y$ in \eqref{eq:SVD}, we have
        $$\Y \Y^\top = UD^2 U^\top + U_{\perp} D_{\perp}^2 U_{\perp}^{\top}.$$
        Recall that $X_0 := M_0 \Lambda_0^\top$ has the SVD $X_0= U_0 D_0 V_0^\top$. 
        Thus, we have $X_0 X_0^\top = U_0 D_0^2 U_0^\top$, 
        with eigengap $\Delta_n \asymp (\sqrt{np_n})^2 = np_n$
        with probability at least $1-o(1),$ from Assumption \ref{assumption:norm} and Lemma \ref{lemma:matrix_normal_norm}.
        Thus, we have
        $$\|UU^\top - U_0 U_0^\top\| \leq 2 \|\sin\Theta(U, U_0)\| \leq 2 \|\sin \Theta(U, U_0)\|_\text{F} \lesssim \dfrac{\|\Y \Y^\top - X_0 X_0^\top\|}{\Delta_n}$$
        with probability at least $1-o(1)$, using the SFDKT, as $k = \mathcal{O}(1)$. 
        To bound the numerator, note that $\Y = X_0 + E$ and thus
        \begin{align}
        \label{suppeq:eigenGapArgument}
        \begin{split}
            \|\Y \Y^\top - X_0 X_0^\top\| & = \|X_0 E^\top + E X_0^\top + EE^\top\|\\
            & \lesssim 2 \|X_0 E^\top \| + \|EE^\top \|\\
            & \lesssim \sqrt{np_n}(\sqrt{n} + \sqrt{p_n}) + (\sqrt{n} + \sqrt{p_n})^2 \quad [\text{Lemma \ref{lemma:spectral-norm-heteroskedastic-entries}}]\\
            & \lesssim \sqrt{np_n} (\sqrt{n} + \sqrt{p_n})
        \end{split}
        \end{align}
        with probability at least $1-o(1)$. Thus, we have
        \begin{align*}
            \|UU^\top - U_0 U_0^\top\| & \lesssim \dfrac{\|\Y \Y^\top - X_0 X_0^\top\|}{\Delta_n}\\
            & \lesssim \dfrac{\sqrt{np_n}(\sqrt{n} + \sqrt{p_n})}{np_n}\\
            & = \dfrac{1}{\sqrt{n}} + \dfrac{1}{\sqrt{p_n}}
        \end{align*}
        with probability at least $1-o(1)$, proving the claim.

        \item[(b)] \, We start with $\|X_0^{\top}UU^{\top}X_0 - X_0^{\top}U_0 U_0^{\top}X_0\|$ and observe that
    \begin{align*}
        \|X_0^{\top}UU^{\top}X_0 - X_0^{\top}U_0 U_0^{\top}X_0 \| & = \|\Lambda_0 M_0^{\top} (UU^{\top} - U_0 U_0^{\top}) M_0 \Lambda_0^{\top} \|.
    \end{align*}
    Next, we remark that $\|AB\| \geq s_{\text{min}}(A) \|B\|$ for any matrices $A, B$ with $A$ having full column rank, with $s_{\min}(A)$ denoting the smallest positive singular value of $A$. 
    From Assumption \ref{assumption:norm}, $\Lambda_0$ has exactly $k$ positive singular values and thus has full column rank. 
    Thus, we have
    \begin{align*}
        & \|X_0^{\top}UU^{\top}X_0 - X_0^{\top}U_0 U_0^{\top}X_0 \|\\
        & = \|\Lambda_0 M_0^{\top} (UU^{\top} - U_0 U_0^{\top}) M_0 \Lambda_0^{\top} \|\\
        & \geq s_{\min}(\Lambda_0) \|M_0^{\top}(UU^{\top} - U_0 U_0^{\top}) M_0 \Lambda_0^\top \|\\
        & \geq  s_{\min}(\Lambda_0) \|\Lambda_0 M_0^\top (UU^\top - U_0 U_0^\top) M_0\| \quad \text{[since $\|A\| = \|A^\top\|$]}\\
        & \geq s_{\min}^2(\Lambda_0) \|M_0^\top (UU^\top - U_0 U_0^\top) M_0\|.
    \end{align*}
    Since the Schatten $\infty$-norm is identical to the spectral norm, Theorem 2 in \cite{luo2021schatten} implies
$$\|X_0^{\top}UU^{\top}X_0 - X_0^{\top}U_0 U_0^{\top}X_0\| = \|X_0^{\top}UU^{\top}X_0 - X_0^{\top}X_0\| = \|(\mathbb{I}_n - UU^{\top}) X_0\|^2 \leq 4\|E\|^2.$$
Thus with probability at least $1-o(1)$,
$$\|M_0^\top (UU^{\top} - U_0 U_0^{\top}) M_0 \| \, s_{\min}^2(\Lambda_0) \leq {4\|E\|^2}.$$ 
Under Assumption \ref{assumption:norm}, $s_{\min}(\Lambda_0) \equiv s_k(\Lambda_0) \asymp \|\Lambda_0\| \asymp \sqrt{p_n}$. 
Define $\sigma_{\text{sum}}^2 = \sum_{j=1}^{p} \sigma_{0j}^2$ and $\sigma_{\max}^2 = \max_{1 \leq j \leq p} \sigma_{0j}^2 = \mathcal{O}(1)$ by Assumption \ref{assumption:var}. 
Lemma \ref{lemma:spectral-norm-heteroskedastic-entries} implies $\|E\| \lesssim (\sigma_{\text{sum}} + \sqrt{n}\sigma_{\max})$ with probability at least $1-o(1)$. 
Since $\sigma_{\text{sum}} \leq \sqrt{p_n} \sigma_{\max} \lesssim \sqrt{p_n}$, one has $\| E \| \lesssim \sqrt{n} + \sqrt{p_n}$ with probability at least $1 - o(1)$.
Thus, with probability at least $1-o(1)$,
\begin{align*}
    \|M_0^\top (UU^{\top} - U_0 U_0^{\top}) M_0\| & \lesssim \dfrac{(\sqrt{n} + \sqrt{p_n})^2}{p_n}
    \lesssim 1 + \dfrac{n}{p_n},
\end{align*}
which proves the desired result.
    \item [(c)] \, Let $\mathcal{P} = UU^\top$, $\mathcal{P}_0 = U_0 U_0^\top$, and fix $1 \leq u \leq p_n.$ 
    Since $\mathcal{P}$ and $\epsilon^{(u)}$ are \textbf{dependent}, we will use a leave-one-out argument to show this result, uncoupling the influence of $\epsilon^{(u)}$ on a perturbed version of $\mathcal{P}$.
    Define $E^{(u)} \in \mathbb{R}^{n \times p_n}$ to have all columns identical to that of $E$, except the $u$th column, replaced by $\nu^{(u)}$ such that $\nu^{(u)}$ and $\epsilon^{(u)}$ are independent and $\nu^{(u)} \sim N_n(0, \sigma_{0u}^2 \mathbb{I}_n).$  
    Consider the perturbed copy
    $$\wt \Y = X_0 + \wt E$$
    and carry out its SVD analogous to that of $\Y$, given by
    $$\wt \Y = \wt U \wt D \wt V^{\top} + \wt U_{\perp} \wt D_{\perp} \wt V_{\perp}^{\top}.$$
    Let $\wt{\mathcal{P}} := \wt U \wt U^{ \top}.$ We start with
    \begin{align*}
        \|(UU^\top - U_0 U_0^\top) \epsilon^{(u)}\|_2 & = \| (\mathcal{P} - \mathcal{P}_0) \epsilon^{(u)}\|_2 \\
        & = \|(\mathcal{P} - \wt{\mathcal{P}})\epsilon^{(u)} + (\wt{\mathcal{P}} - \mathcal{P}_0)\epsilon^{(u)}\|_2\\
        & \leq \|(\mathcal{P}-\wt{\mathcal{P}})\epsilon^{(u)}\|_2 + \|(\wt{\mathcal{P}} - \mathcal{P}_0)\epsilon^{(u)}\|_2\\
        & = \text{(I) + (II)}, \text{ say.}
    \end{align*}
    First, consider bounding (I), for which we have $\|(\mathcal{P} - \wt{\mathcal{P}}) \epsilon^{(u)}\|_2 \leq \|\mathcal{P} - \wt{\mathcal{P}}\| \|\epsilon^{(u)}\|_2$, using the definition of the operator norm.
    Let us bound each of the terms individually. 
    \textcolor{black}{Following Lemma \ref{lemma:laurent-massart}}, we have $\underset{1 \leq j \leq p_n}{\max}\|\epsilon^{(j)}\|_2 \lesssim \sqrt{n}$ and $\underset{1 \leq j \leq p_n}{\max}\|y^{(j)}\|_2 \lesssim \sqrt{n}$ with probability at least $1-o(1).$
    We can bound
    \begin{align}
    \label{suppeq:dkt1}
        \|\mathcal{P} - \wt{\mathcal{P}}\| \leq 2 \, \|\sin\Theta(U, \wt U)\| \leq 2 \, \|\sin \Theta(U, \wt U)\|_\text{F} & \lesssim \dfrac{\|\Y\Y^\top  - \wt \Y \wt \Y^{\top}\|}{\Delta_{n1}},
    \end{align}
    where $\Delta_{n1}$ is the eigengap of $\Y \Y^{ \top}$.
    In \eqref{suppeq:dkt1}, the first inequality is due to Lemma 1 of \cite{cai2018rate}, while the third inequality is due to the SFDKT with $k = \mathcal{O}(1)$.
    To obtain the rate of $\Delta_{n1}$, we first note that 
    \begin{align*}
        \dfrac{1}{np_n}\|\Y \Y^\top - X_0 X_0^\top\| \lesssim \dfrac{1}{\sqrt{n}} + \dfrac{1}{\sqrt{p_n}}
    \end{align*}
    with probability at least $1-o(1)$, following the argument provided to derive \eqref{suppeq:eigenGapArgument}.
    Thus, we have
    $$\|\Y \Y^\top - X_0 X_0^\top\| = o_{P_0}(np_n).$$
    Let $\lambda_j(A)$ denote the $j$th eigenvalue of $A$ for a generic matrix $A$. 
    Since $\lambda_k(X_0 X_0^\top) \asymp (\sqrt{np_n})^2 = np_n$ and $\lambda_{k+1}(X_0 X_0^\top) = 0$, by Weyl's inequality (Theorem 4.3.1 in \cite{horn2012matrix}), we have
    \begin{align*}
        \lambda_k(\Y \Y^\top) & \asymp np_n\\
        \lambda_{k+1}(\Y \Y^\top) & = o_{P_0}(np_n),
    \end{align*}
    with probability at least $1-o(1)$. 
    Thus, the relevant eigengap is given by 
    \begin{equation}
        \label{suppeq:loo1denom}
        \Delta_{n1} \asymp  np_n,
    \end{equation}
    with probability at least $1-o(1).$
    The numerator of \eqref{suppeq:dkt1} is easily bounded by observing that
    \begin{align}
    \label{suppeq:loo1num}
        \|\Y\Y^\top  - \wt \Y \wt \Y^{\top}\| & = \|y^{(u)}y^{(u)\top} - \wt y^{(u)} \wt y^{(u)\top}\| \notag \\
        & \leq \|y^{(u)}\|_2^2 + \|\wt y^{(u)}\|_2^2\lesssim n
    \end{align}
    with probability at least $1-o(1)$, \textcolor{black}{invoking Lemma \ref{lemma:laurent-massart}}.
    The leave-one-out argument allows us to obtain a desirable bound on this perturbation.
    Combining the above bounds \eqref{suppeq:loo1num} and \eqref{suppeq:loo1denom} into \eqref{suppeq:dkt1}, we have
    \begin{align*}
        \|\mathcal{P} - \wt{\mathcal{P}}\| \lesssim \dfrac{n}{np_n} = \dfrac{1}{p_n}
    \end{align*}
    with probability at least $1-o(1).$ This implies that 
    \begin{align}
    \label{suppeq:loo1final}
    \text{(I)} \lesssim \dfrac{\sqrt{n}}{p_n}
    \end{align}
    with probability at least $1-o(1).$

    To bound (II), we first remark that the leave-one-out construction ensures that $\wt{\mathcal{P}} - \mathcal{P}_0$ is \textbf{now independent} of $\epsilon^{(u)} \sim N(0, \sigma_{0u}^2 \mathbb{I}_n)$.
    Thus, we have $$(\wt{\mathcal{P}} - \mathcal{P}_0) \, \epsilon^{(u)} \,\Big|\, \wt{\mathcal{P}}, \mathcal{P}_0 \;\sim\; N\left(0, \sigma_{0u}^2 (\wt{\mathcal{P}} - \mathcal{P}_0)^2\right),$$
    so that $E \left[\|(\wt{\mathcal{P}} - {\mathcal{P}}_0) \, \epsilon^{(u)}\|_2^2 \middle| \wt{\mathcal{P}}, {\mathcal{P}}_0 \right] = \sigma_{0u}^2 \|\wt{\mathcal{P}} - {\mathcal{P}}_0\|_\text{F}^2,$
    where $\|A\|_{\text{F}}$ is the Frobenius norm of a generic matrix $A$. 
    Using the Gaussian concentration inequality (Proposition 5.34 in \cite{vershynin2010introduction}), we have
    \begin{align*}
        P_0\left[(\text{II}) \geq \sigma_{0u} \|\wt{\mathcal{P}} - \mathcal{P}_0\|_\text{F} + t \middle| \wt{\mathcal{P}}, \mathcal{P}_0 \right] & \leq \exp\left(-\dfrac{ct^2}{\sigma_{0u}^2 \|\wt{\mathcal{P}} - \mathcal{P}_0\|^2}\right)\\
        & \leq \exp\left(-\dfrac{ct^2}{\sigma_{0u}^2 \|\wt{\mathcal{P}} - \mathcal{P}_0\|_\text{F}^2}\right),
    \end{align*}
    for a universal constant $c > 0$, since $\|A\| \leq \|A\|_\text{F}$ for any matrix $A$.
    Thus, with $t = \sigma_{0u} \, \|\wt{\mathcal{P}} - {\mathcal{P}}_0\|_\text{F} \, \sqrt{\dfrac{\log n}{c}}$ and for sufficiently large $n,$ we have
    $$P_0\left[\text{(II)} \geq  \dfrac{2 \sigma_{0u}}{\sqrt{c}} \sqrt{\log n} \, \|\wt{\mathcal{P}} - \mathcal{P}_0\|_\text{F} \middle| \wt{\mathcal{P}}, \mathcal{P}_0\right] \leq \dfrac{1}{n}.$$
    Since the bound on the conditional probability does not depend on $\wt{\mathcal{P}}, \mathcal{P}_0$, we have
    \begin{align}
    \label{suppeq:loo2bound}
        \text{(II)} & \lesssim \sqrt{\log n} \, \|\wt{\mathcal{P}} - \mathcal{P}_0\|_\text{F}
    \end{align}
    unconditionally, with (unconditional) probability at least $1-o(1)$. 
    
    We next develop a bound for $\| \wt{\mathcal{P}} - \mathcal{P}_0 \|_\text{F}$ in \eqref{suppeq:loo2bound} as follows.
    First, we observe that since $\text{rank}(A_1 - A_2) \leq \text{rank}(A_1) + \text{rank}(A_2)$ for any two matrices $A_1, A_2$, we have $\text{rank}(\wt{\mathcal{P}} - {\mathcal{P}}_0) \leq k + k = 2k = \mathcal{O}(1)$ from Assumption \ref{assumption:hyper}.
    We now use the standard bound $\|A\|_\text{F}^2 \leq \text{rank}(A) \cdot \|A\|^2$ for any matrix $A$, to obtain
    \begin{equation}
\label{suppeq:loo2frob}
\begin{split}
\| \wt{\mathcal{P}} - \mathcal{P}_0 \|_\text{F} & \leq \sqrt{2k} \| \wt{\mathcal{P}} - \mathcal{P}_0 \|\\
& \leq  2 \sqrt{2k} \|\sin \Theta(\wt{U}, U_0)\|\\
& \leq 2 \sqrt{2k} \|\sin \Theta(\wt{U}, U_0)\|_\text{F}\\
& \leq \dfrac{4 \sqrt{2} \, k\, \|\wt \Y \wt \Y^\top - X_0 X_0^\top\|}{\Delta_{n2}}.
\end{split}
\end{equation}
    In \eqref{suppeq:loo2frob}, the second inequality above is from Lemma 1 in \cite{cai2018rate}, the third is from $\|A\| \leq \|A\|_\text{F}$ for any matrix $A$, and the fourth inequality is due to the SFDKT. 
    To bound the numerator in \eqref{suppeq:loo2frob}, we follow steps analogous to the proof of part (a) of Proposition \ref{proposition:bod}
    to obtain 
    $$4 \sqrt{2} \, k\, \|\wt \Y \wt \Y^\top - X_0 X_0^\top\| \lesssim \sqrt{np_n}(\sqrt{n} + \sqrt{p_n})$$
    with probability at least $1-o(1)$, as $k = \mathcal{O}(1).$
    The denominator in \eqref{suppeq:loo2frob} is bounded easily by observing that $\Delta_{n2} \asymp (\sqrt{np_n})^2 = np_n$
    following Assumption \ref{assumption:norm} and Lemma \ref{lemma:matrix_normal_norm}. 
    Combining the bounds of the numerator and denominator in \eqref{suppeq:loo2frob}, we have
    \begin{align}
    \label{suppeq:loo2final}
        \text{(II)} \lesssim \sqrt{\log n} \left(\dfrac{1}{\sqrt{n}} + \dfrac{1}{\sqrt{p_n}}\right)
    \end{align}
    with probability at least $1-o(1)$. 
    Combining the bounds for (I) and (II) from \eqref{suppeq:loo1final} and \eqref{suppeq:loo2final}, respectively, we obtain the desired result.
    \end{enumerate}
    }
\end{proof}

\subsection{Relevant Lemmas for Theorem 1}
\label{suppsubsec:lemmathm1}

    \begin{lemma}
\label{lemma:concentration-iid-Gaussian}
        Let $E_0 \in \mathbb{R}^{n_1 \times n_2}$ have i.i.d. standard Gaussian entries. Then for every $t>0$, one has
        $$\mathbb{P}\left(\|E_0\|\leq \sqrt{n_1}+\sqrt{n_2}+t\right) \geq 1 - 2\exp(-t^2/2).$$
        As a corollary, for the true factor matrix $M_0 \in \mathbb{R}^{n \times k}$ with $k = \mathcal{O}(1),$ we have
        $$\mathbb{P}\left(\left\|\dfrac{M_0^\top M_0}{n} - \mathbb{I}_k\right\| \leq 3\sqrt{2} \, \sqrt{\dfrac{\log n}{n}}\right) \geq 1 - \dfrac{2}{n}.$$
    \end{lemma}

\begin{proof}
    Refer to Corollary 5.35 and Lemma 5.36 in \cite{vershynin2010introduction}.
\end{proof}
    

    \begin{lemma}
    \label{lemma:spectral-norm-heteroskedastic-entries}
Suppose $E_0 \in \mathbb{R}^{n_1 \times n_2}$ has independent entries such that $E_{0,ij} = g_{ij} b_{ij}$ for $g_{ij} \overset{\text{iid}}{\sim} N(0,1)$ and $\{b_{ij} : 1 \leq i \leq n_1 , 1 \leq j \leq n_2 \}$ are fixed scalars. Let $\sigma_1 = \underset{i}{\max} \sqrt{\sum_j b_{ij}^2}$, $\sigma_2 = \underset{j}{\max} \sqrt{\sum_i b_{ij}^2},$ and $\sigma_{*} = \underset{i,j}{\max} |b_{ij}|.$ Then for every $\epsilon \in (0, 1/2]$ there exists a $c_{\epsilon}^{'} > 0$ such that for all $t \geq 0$,
$$\mathbb{P}\left\{\|E_0\|\geq (1 + \epsilon)(\sigma_1 + \sigma_2) + t \right\} \leq (n_1 \wedge n_2)\exp\{-t^2/(c_{\epsilon}^{'} \sigma_{*}^2)\}.$$
    \end{lemma}
    \begin{proof}
         Refer to Corollary 3.11 in \cite{bandeira2016sharp}.
    \end{proof}
    \begin{lemma}\label{lm:Eckart–Young}
        For any matrix $Z \in \mathbb{R}^{n_1 \times n_2}$, let its $l$th largest singular value be $s_l(Z)$ for $l=1,\ldots,n_1\wedge n_2$. Then for all $1 \leq l < n_1 \wedge n_2,$
        $$s_{l+1}(Z) = \min_{\text{rank}(M)\leq l}\|Z-M\|.$$
    \end{lemma}
    \begin{proof}
        Refer to \cite{eckart1936approximation}.
    \end{proof}

    {\color{black}
    \begin{lemma}\label{lemma:laurent-massart}
        Suppose Assumptions \ref{assumption:dim}-\ref{assumption:hyper} hold.
        For all $1 \leq j \leq p_n$, we have $\|y^{(j)}\|_2^2 \sim (\|\lambda_{0j}\|_2^2 + \sigma_{0j}^2) \, \chi_n^2$ and $\|\epsilon^{(j)}\|_2^2 \sim \sigma_{0j}^2 \, \chi_n^2$, and $$\max\left(\underset{1 \leq j \leq p_n}{\max} \|y^{(j)}\|_2^2, \underset{1 \leq j \leq p_n}{\max} \|\epsilon^{(j)}\|_2^2\right) \lesssim n$$
        with probability at least $1-o(1).$
        \begin{proof}
            The first two statements are immediate from their definitions.
            For the proof of the $\max$ bound, refer to the tail inequality bound of the $\chi^2$ distribution following Lemma 1 of \cite{laurent2000adaptive}.
        \end{proof}
    \end{lemma}
    }
    
{\color{black}
    \begin{lemma}
\label{lemma:post-mean-conv}

Suppose Assumptions \ref{assumption:dim}--\ref{assumption:hyper} hold. Then, as \(n\to\infty\),
\[
\max_{1\leq u,v\leq p_n}
\left|
\frac{1}{n}y^{(u)\top}UU^\top y^{(v)}
-
\lambda_{0u}^{\top}\lambda_{0v}
\right|
=o_{P_0}(1),
\]
and
\[
\max_{1\leq u\leq p_n}
\left|
\frac{1}{n}\|(I_n-UU^\top)y^{(u)}\|_2^2
-
\sigma_{0u}^2
\right|
=o_{P_0}(1).
\]

\end{lemma}
        \begin{proof}
            
            We first note the identity
        \begin{align*}
            & \dfrac{1}{n} y^{(u) \T} UU^{\top} y^{(v)} - \lambda_{0u}^{\top} \lambda_{0v}\\ 
            = & \, \lambda_{0u}^{\top} \left(\dfrac{M_0^{\top} M_0}{n} - \mathbb{I}_k\right) \lambda_{0v} + \dfrac{1}{n} \left[\lambda_{0u}^{\top} M_0^{\T} \epsilon^{(v)} 
            + \lambda_{0v}^{\top} M_0^{\T} \epsilon^{(u)}\right] \, \\
            & +\dfrac{1}{n} \epsilon^{(u) \top} U_0 U_0^{\top} \epsilon^{(v)} + \dfrac{1}{n} y^{(u) \top} (UU^{\top} - U_0 U_0^{\top}) y^{(v)}.
        \end{align*}
        We consider bounding the maximum of the absolute value of each term on the right hand side to bound the maximum of the left hand side over $u,v$.
        The first term is bounded as
        \begin{align*}
            & \max_{u,v} \left|\lambda_{0u}^{\top} \left(\dfrac{M_0^{\top} M_0}{n} - \mathbb{I}_k\right) \lambda_{0v}\right| \\
            & \leq \max_{u,v} \|\lambda_{0u}\|_2 \|\lambda_{0v}\|_2 \left\| \dfrac{M_0^{\top} M_0}{n} - \mathbb{I}_k \right\| \\
            & \lesssim \sqrt{\dfrac{\log n}{n}}, & \text{[Assumption \ref{assumption:norm} and Lemma \ref{lemma:concentration-iid-Gaussian}]}
        \end{align*}
        with probability at least $1-o(1)$.

        \noindent The second term is bounded as follows. 
        First observe that $$\dfrac{1}{n} \, \lambda_{0u}^{\top} M_0^{\top} \epsilon^{(v)} \mid M_0 \sim N\left(0, \dfrac{1}{n^2}\sigma_{0v}^2 \lambda_{0u}^{\top} M_0^{\top} M_0 \lambda_{0u}\right) \overset{d}{=} \dfrac{\sigma_{0v}}{\sqrt{n}} \sqrt{\lambda_{0u}^{\top} \dfrac{M_0^{\top} M_0}{n} \lambda_{0u}} \, \mathcal{Z}_{uv}$$ with $\mathcal{Z}_{uv} \sim N(0,1)$, for any $1 \leq u,v \leq p_n.$
        The Gaussian tail concentration bound combined with a straightforward union bound yields $\underset{u,v}{\max} |\mathcal{Z}_{uv}| \lesssim \sqrt{\log p_n}$ with probability at least $1-o(1)$.
        Since $\|M_0^\top M_0 / n\| \leq 2$ with probability at least $1-o(1)$ from Lemma \ref{lemma:concentration-iid-Gaussian} and $\underset{u}{\max} \|\lambda_{0u}\|_2 = \mathcal{O}(1)$ from Assumption \ref{assumption:norm}, we have
        $$\max_{u,v} \left|\dfrac{1}{n} \, \lambda_{0u}^{\top} M_0^{\top} \epsilon^{(v)}\right| \lesssim \sqrt{\dfrac{\log p_n}{n}} = o(1) \quad \text{[Assumption \ref{assumption:dim}]}$$
        with probability at least $1-o(1).$

        \noindent Now consider the third term. 
        We first observe that for any $1 \leq u \leq p_n$, we have $U_0^\top \epsilon^{(u)} \mid U_0 \sim N(0, \sigma_{0u}^2 U_0^\top U_0) \equiv N(0, \sigma_{0u}^2 \mathbb{I}_k)$, so that $U_0^\top \epsilon^{(u)} \sim  N(0, \sigma_{0u}^2 \mathbb{I}_k)$ unconditionally. 
        Since $k = \mathcal{O}(1)$ from Assumption \ref{assumption:hyper}, $\underset{u}{\max}\|U_0^\top \epsilon^{(u)}\|_2 \lesssim \sqrt{\log p_n}$ with probability at least $1-o(1)$, following another use of the Gaussian concentration inequality and the union bound.
        Thus,
        \begin{align*}
            & \max_{u,v}\left|\dfrac{1}{n}\epsilon^{(u)\top} U_0 U_0^\top \epsilon^{(v)}\right| \\
            & \leq \max_{u,v} \left[\dfrac{1}{n} \|U_0^\top \epsilon^{(u)}\|_2 \|U_0^\top \epsilon^{(v)}\|_2 \right] \\
            & \lesssim \dfrac{\log p_n}{n},
        \end{align*}
        with probability at least $1-o(1).$

        \noindent To bound the final term, we proceed as follows:
        \begin{align*}
            & \max_{u,v} \left|  \dfrac{1}{n} y^{(u) \top} (UU^{\top} - U_0 U_0^{\top}) y^{(v)}\right| \\
            & \leq \dfrac{1}{n} \left(\max_u \|y^{(u)}\|_2^2 \right) \|UU^\top - U_0 U_0^\top\| \\
            & \lesssim \|UU^\top - U_0 U_0^\top \| & \text{[Lemma \ref{lemma:laurent-massart}]}\\
            & \lesssim \dfrac{1}{\sqrt{n}} + \dfrac{1}{\sqrt{p_n}} & \text{[Part (a) of Proposition \ref{proposition:bod}]}
        \end{align*}
        with probability at least $1-o(1).$
        Combining the bounds above, we obtain
        $$\max_{1\leq u,v\leq p_n}
\left|
\frac{1}{n}y^{(u)\top}UU^\top y^{(v)}
-
\lambda_{0u}^{\top}\lambda_{0v}
\right| \lesssim \sqrt{\dfrac{\log(n \vee p_n)}{n}} + \dfrac{1}{\sqrt{p_n}}$$
with probability at least $1-o(1),$ proving our claim.
        

        To derive the second result, we first observe that for any fixed $1 \leq u \leq p_n$, we have $ \|y^{(u)}\|_2^2 \sim (\|\lambda_{0u}\|_2^2 + \sigma_{0u}^2) \, \chi_n^2$. 
        Thus, an application of the chi-square concentration inequality (Theorem 5 in \cite{zhang2020non}) yields
        $$\max_{u}\left\|\dfrac{1}{n}\|y^{(u)}\|_2^2 - (\|\lambda_{0u}\|_2^2 + \sigma_{0u}^2)\right| \lesssim \sqrt{\dfrac{\log p_n}{n}}$$
        with probability at least $1-o(1)$, since $\underset{1 \leq u \leq p_n}{\max} (\|\lambda_{0u}\|_2^2 + \sigma_{0u}^2) = \mathcal{O}(1)$ by Assumptions \ref{assumption:norm} and \ref{assumption:var}.
        Thus, we can write
        \begin{align*}
            & \max_{u}
\left|
\frac{1}{n}\|(I_n-UU^\top)y^{(u)}\|_2^2
-
\sigma_{0u}^2
\right| \\
& \leq \max_{u}\left\|\dfrac{1}{n}\|y^{(u)}\|_2^2 - (\|\lambda_{0u}\|_2^2 + \sigma_{0u}^2)\right| + \max_{u}
\left|
\frac{1}{n}y^{(u)\top}UU^\top y^{(u)}
-
\|\lambda_{0u}\|_2^2
\right| \\
& \leq \max_{u}\left\|\dfrac{1}{n}\|y^{(u)}\|_2^2 - (\|\lambda_{0u}\|_2^2 + \sigma_{0u}^2)\right| + \max_{u,v}
\left|
\frac{1}{n}y^{(u)\top}UU^\top y^{(v)}
-
\lambda_{0u}^\top \lambda_{0v}
\right| \\
& \lesssim \sqrt{\dfrac{\log(n \vee p_n)}{n}} + \dfrac{1}{\sqrt{p_n}}
\end{align*}
with probability at least $1-o(1),$ proving the claim.

        \end{proof}
        }
    
    \begin{lemma}
        \label{lemma:Bmat}
        Let $B = (b_{uv})_{1 \leq u,v \leq p}$ denote the matrix of coefficients as in Section \ref{subsec:cov_correction}. Then, $\|B\|_{\infty} = \mathcal{O}_{P_0}(1)$ under Assumptions \ref{assumption:dim}-\ref{assumption:hyper}.
    \end{lemma}
    \begin{proof}
        First, consider $u,v$ such that $u \neq v$.
        Following from Assumptions \ref{assumption:norm} and \ref{assumption:var}, we have $\underset{u}{\min} \,\|\lambda_{0u}\|_2 \geq c_1 > 0$ and $\underset{u}{\min} \,\sigma_{0u}^2 \geq c_2 > 0$ for some constants $c_1, c_2$, for all $1 \leq u \leq p$. 
        Following from Lemma \ref{lemma:post-mean-conv}, we have for sufficiently large $n$ with probability at least $1-o(1)$: (1) $\underset{u}{\min} \, \|\mu_u\|_2^2 > c_1^2/2$ and $\underset{u}{\min} \, \mathcal{V}_u^2 > c_2 / 2,$ (2) $\underset{u,v}{\max} \, \left|\lambda_{0u}^{\top} \lambda_{0v}\right| \leq \underset{u,v}{\max} \, \|\lambda_{0u}\|_2 \|\lambda_{0v}\|_2 \leq k \|\Lambda_0\|_{\infty}^2$, and (3) $\underset{u}{\max} \|\mu_u\|_2^2 \leq 2k \|\Lambda_0\|_{\infty}^2$ along with $\underset{u,v}{\max} \, |\mu_u^\top \mu_v| \leq 2k \|\Lambda_0\|_{\infty}^2$. 
        This implies that with probability at least $1-o(1)$,
    $$\max_{u \neq v} b_{uv}^2 \leq 1 + \dfrac{16 \,k^2\|\Lambda_0\|_{\infty}^4 }{c_1^2 c_2} < \infty.$$
    The result is derived analogously for $\underset{u}{\max} \, b_{uu}^2$, proving the claim.
    \end{proof}

\begin{lemma}
    \label{lemma:deltasq}
    Suppose Assumptions \ref{assumption:dim}-\ref{assumption:hyper} hold. For each $j=1,\ldots,p_n$, we have
$$\delta_j^2 = \dfrac{Z_j}{n} + F_j,$$
where $Z_j / \sigma_{0j}^2 \sim \chi_{n-k}^2$ and
$\underset{1 \leq j \leq p_n}{\max}|F_j| \lesssim \dfrac{1}{\sqrt n} + \dfrac{1}{\sqrt{p_n}}$
with probability at least $1 - o(1)$.
\end{lemma}
\begin{proof}
Suppose the SVD of $X_0 = M_0 \Lambda_0^{\top} = U_0 D_0 V_0^{\top}$. Let $P_{U_0} = U_0 (U_0^{\top} U_0)^{-1} U_0^{\top} = U_0 U_0^{\top}$ denote the projection matrix onto the column space of $U_0$. We can express $\delta_j^2$ as
$\delta_j^2 = \dfrac{Z_j}{n} + F_j,$
where we have
\begin{align*}
    Z_j & = y^{(j) \top} \left(\mathbb{I}_n - P_{U_0}\right) y^{(j)},\\
    F_j & = \dfrac{\gamma_0 \delta_0^2}{\gamma_n} + \dfrac{1}{n} y^{(j) \top} \left(P_{U_0} - \dfrac{\widehat{\mathbf{M}}\widehat{\mathbf{M}}^{\top}}{n + \tau^{-2}}\right) y^{(j)} - \dfrac{\gamma_0}{n\gamma_n} y^{(j) \top} \left(\mathbb{I}_n - \dfrac{\widehat{\M}\widehat{\M}^{\top}}{n + \tau^{-2}}\right) y^{(j)}.
\end{align*}
Letting $\Y = M_0 \Lambda_0^{\top} + E$ where $E = [\epsilon^{(1)}\ldots \epsilon^{(p)}]$, we have $(\mathbb{I}_n - P_{U_0}) \Y = (\mathbb{I}_n - P_{U_0})E$, implying $(\mathbb{I}_n - P_{U_0}) y^{(j)} = (\mathbb{I}_n - P_{U_0}) \epsilon^{(j)} \sim N_n(0, \sigma_{0j}^2(\mathbb{I}_n - P_{U_0}))$. Since $\mathbb{I}_n - P_{U_0}$ is idempotent (i.e., $(\mathbb{I}_n - P_{U_0})^2 = \mathbb{I}_n - P_{U_0}$), we have $$\dfrac{Z_j}{\sigma_{0j}^2} = \dfrac{\|\left(\mathbb{I}_n - P_{U_0}\right) \, y^{(j)}\|_2^2}{\sigma_{0j}^2} \sim \chi_{\mbox{tr}(\mathbb{I}_n - P_{U_0})}^2 \equiv \chi_{n-k}^2.$$
We now obtain the stated probabilistic upper bound on $|F_j|.$ We observe
$$\max_{1 \leq j \leq p} |F_j| \leq \dfrac{\gamma_0 \delta_0^2}{\gamma_n} + \dfrac{1}{n} \left\|P_{U_0} - \dfrac{nUU^{\top}}{n + \tau^{-2}}\right\| \max_{1 \leq j \leq p} ||y^{(j)}||_2^2 + \dfrac{\gamma_0}{n \gamma_n} \left\|\mathbb{I}_n - \dfrac{nUU^{\top}}{n + \tau^{-2}}\right\| \max_{1 \leq j \leq p} ||y^{(j)}||_2^2,$$
as $\widehat{\M}\widehat{\M}^{\top} = nUU^{\top}$. 
Next, we have $\max_{1 \leq j \leq p}||y^{(j)}||_2^2 \lesssim n$ with probability at least $1 - o(1)$, \textcolor{black}{following Lemma \ref{lemma:laurent-massart}}. 
We start with
\begin{align*}
    G_{1} & :=  \left\|P_{U_0} - \dfrac{nUU^{\top}}{n + \tau^{-2}}\right\|\\
    & \leq  \left\|U_0 U_0^{\top} - UU^{\top}\right\| + \left\|UU^{\top} - \dfrac{nUU^{\top}}{n + \tau^{-2}}\right\|\\
    & \lesssim \dfrac{1}{\sqrt n} + \dfrac{1}{\sqrt{p_n}} + \dfrac{1}{n} \lesssim \dfrac{1}{\sqrt n} + \dfrac{1}{\sqrt{p_n}},
\end{align*}
with probability at least $1 - o(1)$. The second inequality is obtained using part (a) of Proposition \ref{proposition:bod}. 
Next, we work with
\begin{align*}
    G_{2} := & \, \left\|\mathbb{I}_n - \dfrac{nUU^{\top}}{n + \tau^{-2}}\right\|\\
    \leq & \, \|\mathbb{I}_n - UU^{\top}\| + \dfrac{\tau^{-2}}{n} \|UU^{\top}\|\\
    \lesssim & \, 1.
\end{align*}
Combining the bounds for $G_1$ and $G_2$, 
we have with probability at least $1-o(1),$
$\underset{1 \leq j \leq p}{\max}|F_j| \lesssim \dfrac{1}{\sqrt n} + \dfrac{1}{\sqrt{p_n}}.$ This completes the proof.
\end{proof}

\begin{lemma}
    \label{lemma:gamma-conc}
    Let $V_n \sim G(\gamma_n / 2, 1)$ such that $\gamma_n \asymp n$ and let $U_n = V_n - (\gamma_n / 2)$. For any $a_n$ satisfying $a_n \to 0$ and $\sqrt{n} a_n \to \infty$, there exists $c > 0$ such that 
$$P(|U_n| \geq \gamma_n a_n) \lesssim \exp(-c\,n a_n^2).$$
\end{lemma}
\begin{proof}
    Immediate from Theorem 5 in \cite{zhang2020non}.
\end{proof}

\begin{lemma}
\label{lemma:matrix_normal_norm}
    Let $M_0 \in \mathbb{R}^{n \times k}$ be a matrix of iid $N(0,1)$ entries with $n > k$ and let $s_{min}(M_0)$ and $s_{max}(M_0)$ be the smallest and largest singular values of $M_0$, respectively. Then, 
\begin{align*}
    \sqrt{n} - \sqrt{k} \lesssim s_{min}(M_0) \leq s_{max}(M_0) \equiv \|M_0\| & \lesssim \sqrt{n} + \sqrt{k}
\end{align*}
with probability at least $1-(1/n).$ 
In particular, if $k = o(n)$, we have $\|M_0\| \asymp \sqrt{n}$ with probability at least $1-(1/n).$ 
\end{lemma}
\begin{proof}
    Refer to Corollary 5.35 in \cite{vershynin2010introduction} with $t = \sqrt{2 \log (2n)}.$
\end{proof}

\section{Proofs of Additional Related  Lemmas for Theorems \ref{theorem:bvm} and \ref{theorem:asymp-law}}
\label{proof:lemmaTheorem2and3}
\begin{lemma}
    \label{lemma:slutsky-extension}
    Suppose $X_n = Y_n + Z_n,$ where $Y_n \overset{}{\implies} N(0, \sigma_1^2)$ and $Z_n \sim N(0, \sigma_2^2),$ with $Y_n$ independent of $Z_n.$ Then,  $X_n \overset{}{\implies} N(0, \sigma_1^2 + \sigma_2^2).$ 
\end{lemma}
\begin{proof}
    It is immediate from considering the characteristic function of $X_n$ and taking limits.
    \textcolor{black}{We remark here that the same result is obtained even with relaxing $Z_n \implies N(0, \sigma_2^2)$.}
\end{proof}
\begin{lemma}
    \label{lemma:bvm-approx}
    Suppose in the setup of Theorem \ref{theorem:bvm}, we have $T_n = Z_n + Y_n$ for random variables $T_n, Z_n, Y_n$, such that the posterior of $Y_n$ concentrates around $0$ as $n\to \infty$ and the posterior of $Z_n$ is the $N(0,1)$ density. Then as $n \to \infty,$
    $$\sup_{x} \, \left|\wt{\Pi}(T_n \leq x) - \Phi(x) \right| \overset{P_0}{\to} 0.$$
\end{lemma}
\begin{proof}
{\color{black}
Define $\omega(\epsilon) := \underset{x}{\sup} \{\Phi(x+\epsilon) - \Phi(x-\epsilon)\}$ for $\epsilon > 0$.
Since $\Phi(\cdot)$ is Lipschitz-continuous, we have $\omega(\epsilon) \to 0$ as $\epsilon \to 0+$.
Next, we fix $\delta > 0, x \in \mathbb{R}$, and choose $\epsilon > 0$ such that $\omega(\epsilon) < \delta/2.$
We now observe that
$$\{Z_n \leq x-\epsilon\} \cap \{|Y_n| \leq \epsilon\} \subseteq \{Z_n + Y_n \leq x\} \subseteq \{Z_n \leq x+\epsilon\} \cup \{|Y_n| > \epsilon\}.$$
To see why, first note that the left inclusion is trivial, as $|Y_n| \leq \epsilon \implies Y_n \leq \epsilon$ and thus $Z_n \leq x - \epsilon \implies Z_n + Y_n \leq x-\epsilon + \epsilon = x.$
To see why the right inclusion holds, consider the complement $\{Z_n > x+\epsilon\} \cap \{|Y_n| \leq \epsilon\}$.
If $|Y_n| \leq \epsilon$, we have $Y_n \geq -\epsilon$ and thus $Z_n > x+\epsilon \implies Z_n + Y_n > x + \epsilon - \epsilon = x.$
Additionally, we make the following observations: (i) $\wt \Pi[A \cap B] = \wt \Pi(A) - \wt \Pi(A \cap B^c) \geq \wt \Pi(A) - \wt \Pi(B^c)$ and (ii) $\wt \Pi(A \cup B) \leq \wt \Pi(A) + \wt \Pi(B).$
Thus, we obtain
\begin{align*}
    & \wt \Pi(Z_n \leq x-\epsilon) - \wt \Pi(|Y_n| > \epsilon) \leq \wt \Pi(T_n \leq x) \leq \wt \Pi(Z_n \leq x+ \epsilon) + \wt \Pi(|Y_n| > \epsilon)\\
    \implies &  \Phi(x-\epsilon) - \wt \Pi(|Y_n| > \epsilon) \leq \wt \Pi(T_n \leq x) \leq \Phi(x + \epsilon) + \wt \Pi(|Y_n| > \epsilon).
\end{align*}
It follows that
$$\sup_x \left|\wt \Pi(T_n \leq x) - \Phi(x)\right| \leq   \sup_x \max \left\{\Phi(x+\epsilon) - \Phi(x), \Phi(x) - \Phi(x-\epsilon)\right\} + \wt \Pi(|Y_n| > \epsilon).$$
Since $\max\{A, B\} \leq A+B$ for $A, B \geq 0$, we have
$$\sup_x \left|\wt \Pi(T_n \leq x) - \Phi(x)\right| \leq \sup_x \{\Phi(x+\epsilon) - \Phi(x-\epsilon)\} + \wt \Pi(|Y_n|>\epsilon) = \omega(\epsilon) + \wt \Pi(|Y_n|>\epsilon).$$
Then, we have
\begin{align*}
    P_0\left[\sup_x \left|\wt{\Pi}(T_n \leq x) - \Phi(x)\right| > \delta\right] & \leq P_0\left[\omega(\epsilon) + \wt{\Pi}(|Y_n| > \epsilon) > \delta\right] \\
    & \leq P_0\left[\wt \Pi(|Y_n| > \epsilon) > \delta/2\right]\\
    & \to 0,
\end{align*}
as $n \to \infty$, since the concentration property of $Y_n$ implies that $\wt{\Pi}(|Y_n| > \epsilon) \overset{P_0}{\to} 0$ as $n \to \infty.$ 
Since $\delta > 0$ is arbitrary, this proves the claim.
}
\end{proof}

\section{Sensitivity to Inverse-Gamma Prior Hyperparameters}
\label{suppsec:invgamma}

\textcolor{black}{We explore the estimation accuracy of FABLE when estimating the covariance matrix, with loadings generated from a 85-15 spike-and-slab prior with $n=100$ and $k = 10$.
We let $n = 100$ and vary $p \in \{100, 500\}$, $\gamma_0 \in \{0.1, 0.5, 1, 10\}$, and $\delta_0^2 \in \{0.1, 0.5, 1, 10\}$.
The results for $p=100$ and $p=500$ are provided in Tables Tables \ref{supptab:invg1} and \ref{supptab:invg2}, respectively.
The results remain almost identical for the different values of $(\gamma_0, \delta_0^2)$, indicating the robustness of the FABLE procedure to the choice of inverse-gamma prior hyperparameters.}

\begin{table}[H]
\centering
\caption{\normalsize Estimation errors of FABLE for different values of inverse-gamma hyperparameters $(\gamma_0, \delta_0^2)$, with $n=100$, $p = 100$, and $85\%$ sparsity spike-and-slab factor loadings. The ``Mean'' and ``Range'' columns show the average and $2.5\% - 97.5\%$ quantiles across replicates, respectively.}
\label{supptab:invg1}
\normalsize
\begin{tabular}{ccccccccc}
\hline
 & \multicolumn{2}{c}{$\delta_0^2=0.1$} & \multicolumn{2}{c}{$\delta_0^2=0.5$} & \multicolumn{2}{c}{$\delta_0^2=1$} & \multicolumn{2}{c}{$\delta_0^2=10$} \\
\hline
$\gamma_0$ & Mean & Range & Mean & Range & Mean & Range & Mean & Range \\
\hline
0.1 & 0.93 & 0.84 -- 1.05 & 0.93 & 0.84 -- 1.05 & 0.93 & 0.84 -- 1.05 & 0.93 & 0.84 -- 1.06 \\
0.5 & 0.93 & 0.84 -- 1.05 & 0.93 & 0.84 -- 1.05 & 0.93 & 0.84 -- 1.05 & 0.94 & 0.84 -- 1.06 \\
1   & 0.93 & 0.84 -- 1.05 & 0.93 & 0.84 -- 1.05 & 0.93 & 0.84 -- 1.06 & 0.94 & 0.85 -- 1.06 \\
10  & 0.93 & 0.84 -- 1.06 & 0.94 & 0.84 -- 1.06 & 0.94 & 0.85 -- 1.06 & 1.01 & 0.92 -- 1.14 \\
\hline
\end{tabular}
\normalsize
\end{table}

\begin{table}[H]
\centering
\caption{\normalsize Estimation errors of FABLE for different values of inverse-gamma hyperparameters $(\gamma_0, \delta_0^2)$, with $n=100,$ $p = 500$, and $85\%$ sparsity spike-and-slab factor loadings. The ``Mean'' and ``Range'' columns show the average and $2.5\% - 97.5\%$ quantiles across replicates, respectively.}
\label{supptab:invg2}
\normalsize
\begin{tabular}{ccccccccc}
\hline
 & \multicolumn{2}{c}{$\delta_0^2=0.1$} & \multicolumn{2}{c}{$\delta_0^2=0.5$} & \multicolumn{2}{c}{$\delta_0^2=1$} & \multicolumn{2}{c}{$\delta_0^2=10$} \\
\hline
$\gamma_0$ & Mean & Range & Mean & Range & Mean & Range & Mean & Range \\
\hline
0.1 & 0.84 & 0.78 -- 0.90 & 0.84 & 0.78 -- 0.90 & 0.84 & 0.78 -- 0.90 & 0.84 & 0.78 -- 0.90 \\
0.5 & 0.84 & 0.78 -- 0.90 & 0.84 & 0.78 -- 0.90 & 0.84 & 0.78 -- 0.90 & 0.84 & 0.78 -- 0.90 \\
1   & 0.84 & 0.78 -- 0.90 & 0.84 & 0.78 -- 0.90 & 0.84 & 0.78 -- 0.90 & 0.84 & 0.78 -- 0.91 \\
10  & 0.85 & 0.79 -- 0.89 & 0.84 & 0.79 -- 0.89 & 0.84 & 0.79 -- 0.89 & 0.83 & 0.77 -- 0.92 \\
\hline
\end{tabular}
\normalsize
\end{table}

\section{Additional Simulation Results and Tuning Details}
\label{suppsec:sim}


{\color{black}
First, we derive the values of $P_{av}$ for the spike-and-slab case as follows.
Since
$$\Lambda_{0,jl} \overset{\text{iid}}{\sim} \pi_0 \wt{\delta}_0 + (1-\pi_0) N(0, v_0^2), \quad \text{where } v_0^2 = 0.5^2,$$
we have $E^*\|\lambda_{0j}\|_2^2 = k E^*[\Lambda_{0j,1}^2] = k[(0 \times \pi_0) + (1-\pi_0) v_0^2] = k(1-\pi_0)v_0^2$.
By definition, $\bar{S}^2 = E^*[\sigma_{0j}^2] = 0.5 \times (5 + 0.5),$ as we generate $\sigma_{0j}^2 \overset{\text{iid}}{\sim} \mathcal{U}(0.5, 5).$ 
The signal-to-noise ratio (SNR) is given by $R_{av,j} = E^*\|\lambda_{0j}\|_2^2 / E^*(\sigma_{0j}^2) = k(1-\pi_0)v_0^2 / \bar{S}^2.$
Finally, we report $P_{av} = (1/p)\sum_{j=1}^{p} \left\{R_{av,j} / (1 + R_{av,j})\right\}.$ 
}

Next, we provide results for the additional simulation experiments as described in the main document.
\begin{enumerate}
    \item $\,$ Table \ref{supptab:esterror} shows relative errors for the competitors for different $(n,p)$ when the loadings are generated 
    \textcolor{black}{according to setting (ii) of the spike-and-slab setup in Section \ref{subsec:preliminaries}.}
ROTATE performs slightly better than FABLE for $n = 1000$ while FABLE performs slightly better than ROTATE for $n = 500$.
We do not provide the results for MGSP as it did not perform well. 
\textcolor{black}{For this case, $P_{av} = 12\%$.}
\item $\,$ Table \ref{supptab:smallnpsims} provides relative errors for the competitors when loadings are generated 
\textcolor{black}{according to setting (iii) of the spike-and-slab setup in Section \ref{subsec:preliminaries}.}
The ROTATE approach did not converge for the vast majority of replicates; we report summary statistics for the replicates where ROTATE did converge.
In these settings, the blessing of dimensionality does not hold, leading to  inferior performance of FABLE for most cases. \textcolor{black}{The values of $P_{av}$ for $\pi_0 = 0.5$ and $\pi_0 = 0.85$ are $P_{av} = 31\%$ and $P_{av} = 12\%$, respectively.}
\item $\,$ Table \ref{supptab:largeksims} provides relative errors for the setting (iv) in Section \ref{subsec:preliminaries} with $\pi_0 = 0.5$ and a larger $k = 50$. 
In this case, the LW approach performs the best overall across the choices of $(n, p)$.
For reasons of computational feasibility, we did not run MGSP.
\textcolor{black}{For this case, $P_{av} \approx 70\%$.}
\end{enumerate}


\begin{table}[H]
\centering
\caption{\normalsize Comparison of estimation error between multiple approaches, with $85\%$ sparsity in spike-and-slab factor loadings. The ``Mean'' and ``Range'' columns show the average and $2.5\%-97.5\%$ quantiles across replicates, respectively.}
\label{supptab:esterror}
\resizebox{\textwidth}{!}{%
\begin{tabular}{lcccccccc}
\hline
\multicolumn{1}{c}{$(n,p)$} & \multicolumn{2}{c}{$(500,1000)$} & \multicolumn{2}{c}{$(1000,1000)$} & \multicolumn{2}{c}{$(500,5000)$} & \multicolumn{2}{c}{$(1000,5000)$} \\
       Method & Mean & Range & Mean & Range & Mean & Range & Mean & Range \\
\hline
FABLE & 0.40 & 0.37 -- 0.45 
& 0.30 & 0.28 -- 0.34 
& 0.43 & 0.40 -- 0.47 
& 0.31 & 0.29 -- 0.35 \\
ROTATE   & 0.43 & 0.39 -- 0.48           & 0.28 & 0.25 -- 0.31           & 0.47 & 0.41 -- 0.54           & 0.30 & 0.27 -- 0.33           \\
HT       & 0.47    & 0.44 -- 0.52           & 0.32    & 0.28 -- 0.34           & 0.50    & 0.46 -- 0.55           & 0.34    & 0.30 -- 0.36            \\
SCAD     & 0.40    & 0.36 -- 0.44           & 0.40    & 0.37 -- 0.43           & 0.43    & 0.39 -- 0.47           & 0.42    & 0.39 -- 0.45           \\
LW       & 0.59    & 0.56 -- 0.62           & 0.45    & 0.42 -- 0.48           & 0.63    & 0.61 -- 0.66           & 0.48    & 0.45 -- 0.51           \\
\hline
\end{tabular}
}
\end{table}

\begin{table}[H]
\centering
\caption{\normalsize Comparison of estimation error between multiple approaches, with $50\%$ and $85\%$ sparsity in spike-and-slab factor loadings, with smaller $n, p$. 
The ``Mean'' and ``Range'' columns show the average and $2.5\%-97.5\%$ quantiles across replicates, respectively.}
\label{supptab:smallnpsims}
\resizebox{\textwidth}{!}{%
\begin{tabular}{lcccccccc}
\hline
\multicolumn{1}{c}{} & \multicolumn{4}{c}{$\pi_0 = 0.50$} & \multicolumn{4}{c}{$\pi_0 = 0.85$} \\
\cline{2-5} \cline{6-9}
\multicolumn{1}{c}{$(n,p)$} & \multicolumn{2}{c}{$(100,100)$} & \multicolumn{2}{c}{$(100,500)$} & \multicolumn{2}{c}{$(100,100)$} & \multicolumn{2}{c}{$(100,500)$} \\
\cline{2-3} \cline{4-5} \cline{6-7} \cline{8-9}
Method & Mean & Range & Mean & Range & Mean & Range & Mean & Range \\
\hline
FABLE & 0.70 & 0.54 -- 0.85 
& 0.67 & 0.56 -- 0.80 
& 0.93 & 0.84 -- 1.06 
& 0.84 & 0.78 -- 0.90 \\
MGSP  & 0.61 & 0.45 -- 0.87 
& 0.88 & 0.66 -- 1.18 
& 0.68 & 0.49 -- 0.87 
& 1.65 & 1.17 -- 2.31 \\ 
ROTATE   & 0.71 & 0.66 -- 0.76 
& NA & NA -- NA
& 0.61 & 0.57 -- 0.69 
& NA & NA -- NA \\
HT    & 0.68 & 0.52 -- 0.80 
& 0.82 & 0.68 -- 0.93 
& 0.60 & 0.53 -- 0.66 
& 0.82 & 0.77 -- 0.86 \\
SCAD  & 0.60 & 0.50 -- 0.69 
& 0.73 & 0.61 -- 0.81 
& 0.55 & 0.48 -- 0.61 
& 0.74 & 0.69 -- 0.80 \\
LW    & 0.55 & 0.47 -- 0.64 
& 0.63 & 0.57 -- 0.71 
& 0.55 & 0.49 -- 0.62 
& 0.75 & 0.70 -- 0.79 \\
\hline
\end{tabular}
}
\end{table}


\begin{table}[H]
\centering
\caption{\normalsize Comparison of estimation error between multiple approaches, with $50\%$ sparsity in spike-and-slab factor loadings and $k = 50$. 
The ``Mean'' and ``Range'' columns show the average and $2.5\%-97.5\%$ quantiles across replicates, respectively.}
\label{supptab:largeksims}
\resizebox{\textwidth}{!}{%
\begin{tabular}{lcccccccc}
\hline
\multicolumn{1}{c}{$(n,p)$} & \multicolumn{2}{c}{$(500,1000)$} & \multicolumn{2}{c}{$(1000,1000)$} & \multicolumn{2}{c}{$(500,5000)$} & \multicolumn{2}{c}{$(1000,5000)$} \\
       Method & Mean & Range & Mean & Range & Mean & Range & Mean & Range \\
\hline
FABLE & 0.50 & 0.44 -- 0.57 
& 0.35 & 0.32 -- 0.39 
& 0.59 & 0.54 -- 0.65 
& 0.41 & 0.38 -- 0.44 \\
ROTATE   & 0.59 & 0.57 -- 0.63           
& 0.41 & 0.39 -- 0.44           
& 0.55 & 0.52 -- 0.59           
& 0.38 & 0.36 -- 0.41           \\
HT       & 0.52    & 0.46 -- 0.59           & 0.34    & 0.31 -- 0.38           
& 0.62    & 0.57 -- 0.69           
& 0.41    & 0.38 -- 0.44            \\
SCAD     & 0.55    & 0.49 -- 0.61           & 0.37    & 0.33 -- 0.40           
& 0.65    & 0.60 -- 0.73           
& 0.44    & 0.41 -- 0.47           \\
LW       & 0.45    & 0.43 -- 0.48           & 0.33    & 0.31 -- 0.36           
& 0.52    & 0.50 -- 0.55           
& 0.39    & 0.37 -- 0.41           \\
\hline
\end{tabular}
}
\end{table}


Lastly, we present details regarding hyperparameter tuning of the methods HT, SCAD, and LW, as implemented in \texttt{R} with the \texttt{cvCovEst} package \citep{boileau2022cvcovest}.
The HT implementation \texttt{thresholdingEst} requires tuning of the thresholding parameter \texttt{gamma}, while the SCAD implementation \texttt{scadEst} requires tuning of the thresholding parameter \texttt{lambda}.
Both of these thresholding parameters are chosen via 5-fold cross-validation using the matrix Frobenius loss.
The SCAD implementation also requires a shape parameter; we fixed this at its default value $a = 3.7.$
{\small
\begin{verbatim}
    library(cvCovEst)
    thresholdingResults <- cvCovEst(
      dat = ...,
      estimators = c(thresholdingEst),
      estimator_params = list(
        thresholdingEst = list(gamma = seq(0.05, 0.5, length.out = 10))),
      cv_loss = cvMatrixFrobeniusLoss,
      cv_scheme = "v_fold",
      v_folds = 5
    )
    scadResults <- cvCovEst(
      dat = ...,
      estimators = c(scadEst),
      estimator_params = list(
        scadEst        = list(lambda = seq(0.01, 0.5, length.out = 10))),
      cv_loss = cvMatrixFrobeniusLoss,
      cv_scheme = "v_fold",
      v_folds = 5
    )
\end{verbatim}
}
The LW implementation via \texttt{linearShrinkLWEst} does not involve tuning any hyperparameters. 
Instead, its shrinkage hyperparameter is computed analytically from the data.
{\small
\begin{verbatim}
    linearShrinkLWEstResults <- linearShrinkLWEst(...)
\end{verbatim}
}

\section{Gene Data Application Results}
\label{suppsec:app}

{\color{black}
\begin{enumerate}
    \item \, Table \ref{supptab:appblessing} provides negative out-of-sample log-likelihood (OOSLL) values in $10^3$ for the covariance submatrix exercise illustrated in Section \ref{subsec:appblessing}.
    \item \, Table \ref{supptab:apptraintest} provides negative out-of-sample log-likelihood (OOSLL) values in $10^3$ for the direct train-test split exercise illustrated in Section \ref{subsec:apptraintest}.
\end{enumerate}
}

\begin{table}[H]
\centering
\caption{\normalsize Comparison of negative out-of-sample log likelihood ($\times 10^3$) across methods by number of additional variables (Dims). 
The training set consists of $155$ samples and the test set consists of $50$ samples.
Each cell shows Mean and Range ($2.5\% - 97.5\%$ quantiles). Smaller values are better.}
\label{supptab:appblessing}
\resizebox{\textwidth}{!}{%
\begin{tabular}{c|cc|cc|cc|cc|cc}
\toprule
Dims & \multicolumn{2}{c|}{FABLE} & \multicolumn{2}{c|}{MGSP} & \multicolumn{6}{c}{ROTATE} \\
 & & & & & \multicolumn{2}{c|}{\texttt{lambda0=1}} & \multicolumn{2}{c|}{\texttt{lambda0=5}} & \multicolumn{2}{c}{\texttt{lambda0=10}} \\
 & Mean & Range & Mean & Range & Mean & Range & Mean & Range & Mean & Range \\
\midrule
0    & 6.29 & 5.47--7.71 & 5.40 & 5.06--6.09 & 6.31 & 5.46--7.73 & 6.26 & 5.43--7.66 & 6.20 & 5.39--7.54 \\
100  & 6.00 & 5.27--7.34 & 5.40 & 5.06--6.11 & 6.19 & 5.38--7.55 & 6.20 & 5.37--7.58 & 6.20 & 5.36--7.63 \\
200  & 5.70 & 5.13--6.72 & 5.42 & 5.18--6.08 & 6.06 & 5.30--7.43 & 6.09 & 5.32--7.50 & 6.13 & 5.34--7.58 \\
300  & 5.60 & 5.09--6.50 & 5.40 & 5.07--6.03 & 5.98 & 5.25--7.29 & 6.02 & 5.27--7.38 & 6.08 & 5.31--7.49 \\
400  & 5.56 & 5.09--6.44 & 5.46 & 5.16--6.06 & 5.94 & 5.21--7.23 & 5.99 & 5.23--7.34 & 6.07 & 5.27--7.50 \\
500  & 5.58 & 5.18--6.42 & 5.43 & 5.11--6.06 & 5.91 & 5.19--7.19 & 5.96 & 5.21--7.30 & 6.04 & 5.26--7.43 \\
600  & 5.57 & 5.17--6.39 & 5.49 & 5.29--6.15 & 5.89 & 5.18--7.15 & 5.94 & 5.21--7.28 & 6.03 & 5.26--7.44 \\
700  & 5.56 & 5.18--6.38 & 5.50 & 5.24--6.19 & 5.88 & 5.17--7.15 & 5.94 & 5.19--7.25 & 6.03 & 5.24--7.40 \\
800  & 5.56 & 5.16--6.37 & 5.49 & 5.19--6.09 & 5.86 & 5.17--7.14 & 5.92 & 5.19--7.24 & 6.01 & 5.24--7.38 \\
900  & 5.55 & 5.16--6.34 & 5.50 & 5.22--6.10 & 5.85 & 5.17--7.12 & 5.92 & 5.20--7.25 & 6.01 & 5.24--7.41 \\
1000 & 5.55 & 5.15--6.33 & 5.58 & 5.33--6.07 & 5.85 & 5.16--7.12 & 5.91 & 5.17--7.27 & 6.01 & 5.23--7.42 \\
2000 & 5.54 & 5.18--6.26 & 5.64 & 5.40--6.18 & 5.78 & 5.13--7.00 & 5.85 & 5.17--7.10 & 5.96 & 5.23--7.23 \\
4000 & 5.60 & 5.28--6.23 & 5.83 & 5.66--6.21 & 5.76 & 5.10--6.96 & 5.82 & 5.13--7.09 & 5.96 & 5.22--7.29 \\
\hline
\end{tabular}%
}
\end{table}

\begin{table}[H]
\centering
\caption{\normalsize Comparison of negative out-of-sample log likelihood ($\times 10^3$) across methods by training set sizes ($n_{\mathrm{T}}$), evaluated on a test set of $35$ samples. 
Each cell shows Mean and Range (Lower--Upper). Smaller values are better.}
\label{supptab:apptraintest}
\resizebox{\textwidth}{!}{%
\begin{tabular}{c|cc|cc|cc|cc|cc}
\toprule
$n_{\text{T}}$ & \multicolumn{2}{c|}{FABLE} & \multicolumn{2}{c|}{MGSP} & \multicolumn{6}{c}{ROTATE} \\
& & & & & \multicolumn{2}{c|}{\texttt{lambda0=1}} & \multicolumn{2}{c|}{\texttt{lambda0=5}} & \multicolumn{2}{c}{\texttt{lambda0=10}} \\
& Mean & Range & Mean & Range & Mean & Range & Mean & Range & Mean & Range \\
\midrule
110 & 152.61 & 138.07--170.55 & 136.18 & 129.54--147.44 & 230.62 & 195.83--273.07 & 228.23 & 193.56--269.57 & 222.61 & 189.10--263.03 \\
130 & 145.86 & 131.29--162.20 & 131.68 & 121.98--152.25 & 175.36 & 139.47--213.48 & 173.28 & 137.92--210.97 & 170.25 & 135.43--207.49 \\
150 & 141.57 & 126.05--157.70 & 118.61 & 114.80--122.95 & 154.55 & 119.74--186.32 & 152.61 & 118.32--183.79 & 151.15 & 117.75--181.74 \\
170 & 127.24 & 109.65--144.21 & 116.85 & 110.62--126.17 & 123.27 & 103.13--140.77 & 121.97 & 101.82--139.23 & 122.31 & 101.82--139.18 \\
\hline
\end{tabular}%
}
\end{table}

\section{FABLE Algorithm}
\label{suppsec:algorithm}

\begin{algorithm}[H]
\captionsetup{labelfont={sc,bf}, labelsep=newline}
\fontsize{11}{15}\selectfont 
\setlength{\abovedisplayskip}{6pt} 
\setlength{\belowdisplayskip}{6pt} 
  \caption{Steps to obtain $N_0$ samples from the FABLE-posterior for the covariance matrix.}
  \label{algorithm:fable}
    \textbf{Input:} Data matrix $\Y \in \mathbb{R}^{n \times p}$, $N_0$ Monte Carlo (MC) samples with default $N_0 = 1000$, inverse-gamma hyperparameters $(\gamma_0, \delta_0^2)$ with default $\gamma_0 = \delta_0^2 = 1$, and the upper bound on the cumulative singular value proportion ${S}_0$ with default  \textcolor{black}{$S_0 = 0.75$}. Let $r = n \wedge p$.\\
    
    \textbf{Step 1:} Compute the SVD of $\Y$ as $\Y = U^* D^* V^{* \top}$ with $U \in \mathbb{R}^{n \times r}, D \in \mathbb{R}^{r \times r}, V \in \mathbb{R}^{p \times r}$, and $D = \mbox{diag}(s_1,\ldots,s_r)$ such that $s_1 \geq s_2 \geq \ldots \geq s_r \geq 0.$
    
    \noindent
    \textbf{Step 2:} 
    Let $\widehat{k} = \arg \min \, \{\mbox{JIC}(k^*) \, \vert \, k^* = 1, \ldots, \mathcal{K}_0\},$ where 
    $$\mathcal{K}_0 = \min\left\{\mathcal{K}: 1 \leq \mathcal{K} \leq r, \left(\sum_{j=1}^{\mathcal{K}} s_j\right)/\left(\sum_{j=1}^{r} s_j\right) \geq S_0 \right\}.$$

    \noindent
    \textbf{Step 3:} Let $U$ consist of the first $\widehat{k}$ columns of $U^*$. For $1 \leq j \leq p$, let $y^{(j)}$ denote the $j$th column of $\Y$, and obtain
$    \mathcal{L}_j^2  = \|U^{\top} y^{(j)}\|_2^2 / n$ and 
$        \mathcal{V}_j^2  =  \|(\mathbb{I}_n - UU^{\top}) y^{(j)}\|_2^2 / n$ for $1 \leq j \leq p.$

    \noindent
    \textbf{Step 4:} Estimate $\tau^2$ by $$\widehat{\tau}^2 = \dfrac{1}{p\widehat{k}} \sum_{j=1}^{p} \dfrac{\mathcal{L}_j^2}{\mathcal{V}_j^2}.$$

    \noindent
    \textbf{Step 5:} For $1 \leq j \leq p$, let $\mu_j = \sqrt{n} U^{\top} y^{(j)} / \left(n + \widehat{\tau}^{-2} \right).$
    For $1 \leq u, v \leq p$, let 
\begin{align*}
\begin{split}
b_{uv} & = \left\{\begin{array}{ll}
\left(1 + \dfrac{\|\mu_u\|_2^2 \|\mu_v\|_2^2 + (\mu_u^{\top} \mu_v)^2}{\mathcal{V}_u^2 \|\mu_v\|^2 + \mathcal{V}_v^2 \|\mu_u\|_2^2}\right)^{1/2}, & \text{if $u \neq v,$}\\
\left(1 + \dfrac{\|\mu_u\|_2^2}{2 \mathcal{V}_u^2}\right)^{1/2}, & \text{if $u = v,$}
\end{array}
\right.
\end{split}
\end{align*}
and let $\rho$ be the solution to \eqref{eq:rhoChoiceMean}; a convenient expression for $\wh{q}_{uv}(\rho)$ is provided in \eqref{suppeq:estCovConvenient}.
Other approaches to choose $\rho$ are highlighted in Section \ref{subsec:cov_correction}.

\textbf{Step 6:} Let $\gamma_n = \gamma_0 + n$. For $j=1, \ldots, p$, evaluate
$$\gamma_n \delta_j^2 = \gamma_0 \delta_0^2 + y^{(j) \T} \left(\mathbb{I}_n - \dfrac{nUU^{\T}}{n + \widehat{\tau}^{-2}}\right) y^{(j)}.$$

\textbf{Step 7:} For each $t = 1, \ldots, N_0$, independently sample $(\wt{\lambda}_j^{(t)}, \wt{\sigma}_{j}^{(t) 2})$ for $j=1,\ldots,p$ following
$$\wt{\sigma}_{j}^{(t) 2}  \sim \mbox{IG}\left(\dfrac{\gamma_n}{2}, \dfrac{\gamma_n \delta_j^2}{2}\right),\qquad 
    \wt{\lambda}_{j}^{(t)} \mid \wt{\sigma}_{j}^{(t) 2}  \sim N_k\left(\mu_j, \dfrac{\rho^2 \wt{\sigma}_{j}^{(t) 2}}{n + \widehat{\tau}^{-2}} \mathbb{I}_k \right).$$ 
Form $\wt{\Lambda}^{(t)} = \left[\wt{\lambda}_{1}^{(t)}, \ldots, \wt{\lambda}_{p}^{(t)}\right]^{\T}$ and $\wt{\Sigma}^{(t)} = \mbox{diag}\left(\wt{\sigma}_{1}^{(t)2}, \ldots, \wt{\sigma}_{p}^{(t)2}\right)$.\\

\textbf{Step 8:} 
For each $t=1,\ldots,N_0$, 
compute the $t$-th coverage-corrected sample as 
$$    \wt{\Psi}^{(t)} = \wt{\Lambda}^{(t)} \wt\Lambda^{(t) \top} + \wt\Sigma^{(t)}. $$

\textbf{Output:} The $N_0$ MC samples of the covariance matrix $\wt\Psi^{(1)}, \ldots, \wt\Psi^{(N_0)}$.\\

\end{algorithm}

\newpage

We remark here that \eqref{eq:rhoChoiceMean} may be conveniently rewritten entirely in terms of the coefficients $(b_{uv})_{1 \leq u, v \leq p}.$
From \eqref{eq:asympcoverage}, the estimated frequentist coverage $\wh{q}_{uv}(\rho)$ for entry $(u,v)$ as a function of $\rho$ is given by
\begin{align}
    \label{suppeq:estCovConvenient}
    \wh{q}_{uv}(\rho) = 2 \Phi\left\{z_{1 - (\alpha/2)} \mathcal{R}_{uv}(\rho)\right\} - 1,
\end{align}
where the ratio $\mathcal{R}_{uv}(\rho) = \wh{l}_{0,uv}(\rho) / \wh{S}_{0,uv}$ is given by 
\begin{align*}
\label{suppeq:covCorrectionRatioFunction}
\begin{split}
\mathcal{R}_{uv}(\rho) & = \left\{\begin{array}{ll}
\dfrac{\rho}{b_{uv}}, &  \text{if $u \neq v,$}\\ \\
\left\{\dfrac{1 + 4\rho^2(b_{uu}^2-1)}{\left(2b_{uu}^2 - 1 \right)^2}\right\}^{1/2}, & \text{if $u = v.$}
\end{array}
\right.
\end{split}
\end{align*}
It is straightforward to verify that $\mathcal{R}_{uv}(b_{uv}) = 1$ for all $1 \leq u, v \leq p.$

\section{Comparisons Between JIC, AIC, and BIC}
\label{suppsec:JICvsAICBIC}


Naive implementations of information-theoretic approaches utilizing the marginal likelihood, such as the Akaike information criterion (AIC) or the Bayesian information criterion (BIC) \citep{akaike1974new, akaike1987factor} scale poorly as $n, p$ increase. 
This is primarily due to the inversion of the $p \times p$ covariance $\Psi^{(k)} = \Lambda^{(k)} \Lambda^{(k)\top} + \Sigma$ with theoretical complexity $\mathcal{O}(p^3)$ for each candidate $k$ when evaluating the marginal likelihood. 
The complexity of this step is reduced to $\mathcal{O}(pk^2)$ with the Sherman-Morrison-Woodbury (SMW) formula \citep{sherman1950adjustment, woodbury1950inverting}
\begin{equation}
    \label{eq:smwfactor}
\left(\Lambda^{(k)}\Lambda^{(k)\top} + \Sigma \right)^{-1} = \Sigma^{-1} - \Sigma^{-1} \Lambda^{(k)}\left(\mathbb{I}_k + \Lambda^{(k)\top} \Sigma^{-1} \Lambda^{(k)}\right)^{-1} \Lambda^{(k)\top} \Sigma^{-1}.
\end{equation}
However, the numerical error associated with the use of the general SMW formula
\begin{equation}
\label{eq:smw}
    (\mathcal{A} + \mathcal{U}\mathcal{V}^\top)^{-1} = \mathcal{A}^{-1} - \mathcal{A}^{-1} \mathcal{U}\left(\mathbb{I} + \mathcal{V}^\top \mathcal{A}^{-1} \mathcal{U}\right)^{-1}\mathcal{V}^\top \mathcal{A}^{-1}
\end{equation}
can be indefinitely large if either $\mathcal{A}$ is ill-conditioned or the capacitance matrix $\mathcal{C} := \mathbb{I} + \mathcal{V}^\top \mathcal{A}^{-1} \mathcal{U}$ is ill-conditioned (or both).
In factor models, we have $\mathcal{A} = \Sigma$ and $\mathcal{C} \equiv \mathcal{C}^{(k)} = \mathbb{I}_k + \Lambda^{(k)\top} \Sigma^{-1} \Lambda^{(k)}$.
Thus, if some of the diagonal entries of $\Sigma$ are close to $0$, both $\Sigma$ and $\mathcal{C}^{(k)}$ could be ill-conditioned for inversion in \eqref{eq:smwfactor}.
Additionally, $\mathcal{C}^{(k)}$ could potentially be ill-conditioned if we set $\mathcal{U} = \mathcal{V} = \Lambda^{(k)}$ in \eqref{eq:smw}, 
since a naive choice of $\mathcal{U}$ and $\mathcal{V}$ when implementing the SMW formula may lead to an ill-conditioned $\mathcal{C}$ and numerical instability \citep{hager1989updating}.
Furthermore, \cite{ma2025note} 
demonstrate that when $\mathcal{C}$ is ill-conditioned, even using an approximate inverse of $\mathcal{C}$ in \eqref{eq:smw} can substantially amplify both the forward and backward numerical errors corresponding to matrix inversion with the SMW formula.

In contrast, since the JIC is based on the joint likelihood of $ (\mathbf{M}^{(k)}, \Lambda^{(k)}, \Sigma)$,
it avoids the inversion of $\Psi^{(k)}$ altogether.
As noted by \cite{chen2022determining}, neither AIC nor BIC can consistently estimate the true rank of the model from joint likelihoods, as the number of parameters diverges with increasing sample size.
The authors show that their proposed JIC approach alleviates such issues, ensuring both computational scalability and consistency of the estimate $\wh{k}.$
Additionally, as noted in Section \ref{subsubsec:tunek}, the JIC is applicable in those factor models where the marginal likelihood is not explicitly available, but the joint likelihood is. 
For instance, such situations arise with non-Gaussian errors or non-linear measurement structures.

Next, we carry out simulation studies to assess the performance of $\wh{k}$ with the JIC 
and compare with the AIC and BIC.
The results, available in Table \ref{supptab:kEst}, consist of the proportion of replicates for which $\wh{k}$ equals, is less than, and greater than the true rank $k_0$ of the true signal $M_0 \Lambda_0^\top$, with $\wh k$ estimated using the JIC, AIC, and BIC.
For small $n = 100$, the results can be substantially improved by relying on more principled estimators of the signal matrix instead of a direct spectral estimate.
All three penalties perform similarly across the vast majority of simulation scenarios, with the AIC providing slight improvements for some cases with small $n$.

\begin{table}[H]
\caption{\normalsize Proportion across replicates for which the estimate $\wh{k}$ equals, is less than, and is greater than the true rank $k_0$ for all the simulation cases considered.
BD and SS denote the block-diagonal and spike-and-slab loadings cases, respectively.}
\centering
\label{supptab:kEst}

\resizebox{\textwidth}{!}{%
\begin{tabular}{l l l l | r r r | r r r | r r r}
\hline
& & & &
\multicolumn{3}{c|}{Proportion (JIC)} &
\multicolumn{3}{c|}{Proportion (AIC)} &
\multicolumn{3}{c}{Proportion (BIC)} \\
Model & Sparsity & $k_0$ & $(n,p)$ &
$\wh{k}=k_0$ & $\wh{k}<k_0$ & $\wh{k}>k_0$ &
$\wh{k}=k_0$ & $\wh{k}<k_0$ & $\wh{k}>k_0$ &
$\wh{k}=k_0$ & $\wh{k}<k_0$ & $\wh{k}>k_0$ \\
\hline

\multirow{4}{*}{BD}
 & \multirow{4}{*}{85\%} & 10 & (500,1000)
 & 1.00 & 0.00 & 0.00
 & 1.00 & 0.00 & 0.00
 & 1.00 & 0.00 & 0.00 \\
 & & 10 & (500,5000)
 & 1.00 & 0.00 & 0.00
 & 1.00 & 0.00 & 0.00
 & 1.00 & 0.00 & 0.00 \\
 & & 10 & (1000,1000)
 & 1.00 & 0.00 & 0.00
 & 1.00 & 0.00 & 0.00
 & 1.00 & 0.00 & 0.00 \\
 & & 10 & (1000,5000)
 & 1.00 & 0.00 & 0.00
 & 1.00 & 0.00 & 0.00
 & 1.00 & 0.00 & 0.00 \\
\cline{1-13}

\multirow{16}{*}{SS}
 & \multirow{10}{*}{50\%}
 & 10 & (100,100)
 & 0.00 & 0.00 & 1.00
 & 0.14 & 0.26 & 0.60
 & 0.00 & 1.00 & 0.00 \\
 & & 10 & (100,500)
 & 0.30 & 0.70 & 0.00
 & 0.00 & 0.00 & 1.00
 & 0.00 & 0.00 & 1.00 \\
 & & 10 & (500,1000)
 & 1.00 & 0.00 & 0.00
 & 1.00 & 0.00 & 0.00
 & 1.00 & 0.00 & 0.00 \\
 & & 10 & (500,5000)
 & 1.00 & 0.00 & 0.00
 & 1.00 & 0.00 & 0.00
 & 1.00 & 0.00 & 0.00 \\
 & & 10 & (1000,1000)
 & 1.00 & 0.00 & 0.00
 & 1.00 & 0.00 & 0.00
 & 1.00 & 0.00 & 0.00 \\
 & & 10 & (1000,5000)
 & 1.00 & 0.00 & 0.00
 & 1.00 & 0.00 & 0.00
 & 1.00 & 0.00 & 0.00 \\
\cline{3-13}

 & & 50 & (500,1000)
 & 1.00 & 0.00 & 0.00
 & 1.00 & 0.00 & 0.00
 & 1.00 & 0.00 & 0.00 \\
 & & 50 & (500,5000)
 & 1.00 & 0.00 & 0.00
 & 1.00 & 0.00 & 0.00
 & 1.00 & 0.00 & 0.00 \\
 & & 50 & (1000,1000)
 & 1.00 & 0.00 & 0.00
 & 1.00 & 0.00 & 0.00
 & 1.00 & 0.00 & 0.00 \\
 & & 50 & (1000,5000)
 & 1.00 & 0.00 & 0.00
 & 1.00 & 0.00 & 0.00
 & 1.00 & 0.00 & 0.00 \\
\cline{2-13}

 & \multirow{6}{*}{85\%}
 & 10 & (100,100)
 & 0.00 & 0.98 & 0.02
 & 0.00 & 0.46 & 0.54
 & 0.00 & 0.98 & 0.02 \\
 & & 10 & (100,500)
 & 0.00 & 1.00 & 0.00
 & 0.00 & 0.00 & 1.00
 & 0.00 & 0.00 & 1.00 \\
 & & 10 & (500,1000)
 & 0.98 & 0.02 & 0.00
 & 1.00 & 0.00 & 0.00
 & 0.00 & 1.00 & 0.00 \\
 & & 10 & (500,5000)
 & 0.10 & 0.90 & 0.00
 & 1.00 & 0.00 & 0.00
 & 1.00 & 0.00 & 0.00 \\
 & & 10 & (1000,1000)
 & 1.00 & 0.00 & 0.00
 & 1.00 & 0.00 & 0.00
 & 1.00 & 0.00 & 0.00 \\
 & & 10 & (1000,5000)
 & 1.00 & 0.00 & 0.00
 & 1.00 & 0.00 & 0.00
 & 1.00 & 0.00 & 0.00 \\
\hline
\end{tabular}
}
\end{table}


\section{Theoretical Computational Complexity of FABLE}
\label{suppsec:runtimeTheory}


For the purposes of computing theoretical computational complexity of FABLE, we assume $k$ is fixed and $n, p$ vary. 
We first consider the order complexity when generating $N_0$ posterior samples of the covariance matrix. 
The first stage SVD of the $n \times p$ matrix $\mathbf{Y}$ has $\mathcal{O}(np \min(n,p))$ complexity \citep{vasudevan2017hierarchical}, which needs to be computed only once. 
Computing the $p$ hyperparameters $\mu_j$ and $\delta_j^2$ for $j=1,\ldots,p$ has $\mathcal{O}(np)$ complexity, while generating each Monte Carlo sample of $(\wt{\Lambda}, \wt{\Sigma})$ for $j = 1, \ldots, p$ has $\mathcal{O}(p)$ complexity since $(\wt{\lambda}_{j}, \wt{\sigma}_{j}^2)$ are independently generated for $j=1,\ldots,p$. 
Lastly, obtaining $\wt\Psi = \wt\Lambda \wt\Lambda^{\top} + \wt\Sigma$ has $\mathcal{O}(p^2)$ complexity for each sample, leading to $\mathcal{O}(p^2 N_0)$ complexity to obtain $N_0$ samples. 
Thus, obtaining $N_0$ posterior samples of the covariance matrix after coverage-correction has $\mathcal{O}\{p^2 N_0 + p N_0 + np \min(n,p)\} = \mathcal{O}\{p^2 N_0 + np \min(n,p)\}$ complexity. 

Next, we consider the order complexity to obtain the FABLE-posterior mean of the covariance. 
As before, we require computing the SVD with $\mathcal{O}\{np \min(n,p)\}$ complexity, while computing the hyperparameters $(\mu_j, \delta_j^2)$ for $j = 1,\ldots,p$ has $\mathcal{O}(np)$ complexity. 
Obtaining the matrix $G = G_0 G_0^{\T}$ from $G_0$ has $\mathcal{O}(p^2)$ complexity.
Thus, computing the FABLE-posterior mean of $\Psi$ has asymptotic complexity $\mathcal{O}\{p^2 + np \min(n,p)\}$. 

The runtime for both sampling and the FABLE-posterior mean computation can be substantially improved by computationally efficient implementations of the SVD step, such as the randomized SVD \citep{halko2011finding} when either $n$ or $p$ is very large. Additionally, since FABLE employs embarrassingly parallel sampling, its computational runtime would be further improved by parallelizing the computations across cores or utilizing graphics processing units (GPUs).


\end{document}